\documentclass[11pt]{amsart}

\pdfoutput = 1

\usepackage{fouriernc,macros,slashed,amsaddr}

\numberwithin{equation}{section}
\newcommand{\nocontentsline}[3]{}
\newcommand{\tocless}[2]{\bgroup\let\addcontentsline=\nocontentsline#1{#2}\egroup}

\begin{document}

\title{Twisted characters and holomorphic symmetries}

\author{Ingmar Saberi}
\address{Mathematisches Institut der Universit\"at Heidelberg \\ Im Neuenheimer Feld 205 \\ 69120 Heidelberg \\ Deutschland}
\email{saberi@mathi.uni-heidelberg.de}

\author{Brian R. Williams}
\address{Department of Mathematics, Northeastern University \\ 567 Lake Hall \\ Boston, MA 02115 \\ U.S.A.}
\email{br.williams@northeastern.edu}

\begin{abstract}
We consider holomorphic twists of arbitrary supersymmetric theories in four dimensions. Working in the BV formalism, we rederive classical results characterizing the holomorphic twist of chiral and vector supermultiplets, computing the twist explicitly as a family over the space of nilpotent supercharges in minimal supersymmetry. The BV formalism allows one to work with or without auxiliary fields, according to preference; for chiral superfields, we show that the result of the twist is an identical BV theory, the holomorphic $\beta\gamma$ system with superpotential, independent of whether or not auxiliary fields are included. We compute the character of local operators in this holomorphic theory, demonstrating agreement of the free local operators with the usual index of free fields. The local operators with superpotential are computed via a spectral sequence, and are shown to agree with functions on a formal mapping space into the derived critical locus of the superpotential. We consider the holomorphic theory on various geometries, including Hopf manifolds and products of arbitrary pairs of Riemann surfaces, and offer some general remarks on dimensional reductions of holomorphic theories along the $(n-1)$-sphere to topological quantum mechanics. We also study an infinite-dimensional enhancement of the flavor symmetry in this example, to a recently-studied central extension of the derived holomorphic functions with values in the original Lie algebra that generalizes the familiar Kac--Moody enhancement in two-dimensional chiral theories. 
\end{abstract}

\maketitle
\thispagestyle{empty}

\newpage
\tableofcontents

\section{Introduction}
Twists of supersymmetric theories have been the subject of intense study over the past thirty years. Such twists produce simpler quantum field theories, which are of mathematical interest as sources of, or organizing principles for, invariants of the spacetimes on which they live,   and are often related to interesting gauge-theoretic moduli spaces associated to the spacetime. 
Perhaps one of the most familiar cases is the topological twist of $\N=2$ gauge theory considered in~\cite{WittenTwist}, for which the  relevant moduli space is the space of anti-self-dual instantons. In such a topologically twisted theory, deformations of the metric act trivially up to homotopy, so that the theory depends only on the smooth structure of the spacetime.  

Starting on an affine spacetime, a necessary condition for the existence of a topological twist is a nilpotent supercharge whose image under the bracket contains all of the translation operators. 
More generally, when no such supercharge exists, one can still define a more general class of twists by passing to the cohomology of a chosen nilpotent supercharge, with the caveat that the resulting theory may only make sense on manifolds of restricted holonomy.\footnote{In fact, such a  caveat applies even for topological supercharges; see~\cite{NV} for a general discussion.} An example of such a construction is the \emph{holomorphic twist}, which (due to general properties of Clifford multiplication) exists in any even dimension with any amount of supersymmetry. In these cases, the space of nullhomotopic translations is half-dimensional, and corresponds to a choice of complex structure on the spacetime. 
The theories which result from such a procedure depend only on the structure of the spacetime as a complex manifold, and can be defined generally on K\"ahler (or Calabi--Yau) manifolds; we will refer to such objects---whether or not they arise from the twist of a full $\N=1$ theory---as \emph{holomorphic field theories}.

Holomorphic twists have been previously considered in~\cite{JohansenHolo,NekThesis,CostelloHolomorphic};
as objects in their own right, holomorphic field theories in various guises have been studied in the works of \cite{NekThesis, NekCFT, NekChiral, BaulieuCS}, just for example. 
More recently, a program for studying twisted versions of supergravity and string theory in terms of Kodaira-Spencer theory and holomorphic Chern-Simons is currently being developed by Costello and Li in~\cite{CostelloLiSUGRA, CostelloLiTypeI} and Costello in~\cite{CostelloMtheory1, CostelloMtheory2}.
For related work, including supporting evidence of some conjectures of Costello and Li, see~\cite{CY4, NV}. 
In addition, a foundational mathematical treatment of holomorphic field theory is given in~\cite{BWhol}. 
In this paper, we will most closely borrow notations and conventions from the cited works~\cite{NV, BWhol}.

One advantage of holomorphic twists is that renormalization is significantly better behaved for holomorphic theories than for their untwisted parents.
Indeed, regularization in supersymmetric theories, especially gauge theories, is notoriously difficult. 
A salient feature of holomorphic theories is the existence of a gauge in which analytic difficulties become much more tractable.
Consequently, facets of these theories, such as their anomalies, can be cast in a more algebraic framework.
While we don't delve deeply into a utilization of such features in this paper, we refer the reader to~\cite{BWhol} for work on the renormalization of holomorphic field theories.

Another appealing aspect of holomorphic theories is the theory of their observables.
In complex dimension one, the local observables of a holomorphic theory are mathematically described by a vertex (or chiral) algebra \cite{BD}.
Likewise, on a global Riemann surface there is a rich theory of conformal blocks which describes how such vertex algebras are glued together. 
The partition functions of chiral theories have interesting modular properties, and a chiral theory can be obtained from a supersymmetric theory via a holomorphic twist; the partition  function  of the  holomorphic twist is then  the elliptic genus of the original theory. Moreover,  global symmetries of two-dimensional theories enhance to infinite-dimensional Kac--Moody symmetries of chiral theories. 
One philosophical point that we wish to explore in this note is the idea that these phenomena are not as peculiar to the two-dimensional case as is usually assumed. Rather, most of the above can  be interpreted in the setting of holomorphic twists of minimally supersymmetric theories in any even dimension. There are several good reasons the two-dimensional case is easier to see: firstly, due to the factorization of the two-dimensional wave equation, one does not need the notion of a holomorphic twist to arrive at chiral two-dimensional theories. In particular, supersymmetry is not an essential ingredient in dimension two. Secondly, the Dolbeault cohomology of punctured $\C^d$ is concentrated in a single degree only for $d=1$. For larger values of~$d$, it is therefore essential to work in a derived setting. The relevant analogues of Kac--Moody algebras were studied recently in~\cite{GwilliamWilliams,FHK}.

The theory of factorization algebras, as developed in \cite{CG1, CG2}, is a general mathematical tool that described the algebra of observables in any perturbative field theory. 
While the data of a factorization algebra can be quite unwieldy, in special situations---for example, in chiral CFTs or topological theories---factorization algebras admit efficient algebraic descriptions. 
For higher-dimensional holomorphic theories, it is tempting to speculate that this holomorphic factorization algebra should be analogous to a vertex algebra, and should reflect the structure of  the operator product expansion  in  the twisted theory. 
While we don't develop such a theory in this paper, we gesture at the existence of such a structure through quantities reminiscent of ordinary vertex algebras, such as a higher dimensional $q$-character. 

Throughout this paper, we work with $\N=1$ supersymmetric theories in four dimensions, focusing on the chiral multiplet with superpotential (the Wess--Zumino model).  
While we compute the twist of a general $\N=1$ theory in our language, and make use of the twisted vector multiplet to connect to the holomorphic flavor current multiplet, we do not provide explicit results for characters of local operators in gauge theory here, although the techniques to do so are largely developed. We hope to return to a more complete treatment in future work, with the goal of studying dualities, such as Seiberg duality, at the level of the holomorphic twist. 
The explicit description we give of the local operators in twisted $\N=1$ theories, coupled with an understanding of renormalization group flow in the holomorphic setting, should make it tractable to understand Seiberg dualities completely at the level of twisted local operators.

Here is an outline of the remainder of the paper. 
In~\S\ref{sec:prelim}, we recall basic facts about supersymmetry algebras and holomorphic twists, as well as setting up conventions for the remainder of the paper and giving some general discussion.
The key perspective we'd like to emphasize here is that---even if some of the calculations we perform are for a fixed choice of nilpotent supercharge---the formalism is set up in such a way that all constructions exist as a family over the nilpotence variety. In addition to being of theoretical interest, this will have applications later in the paper, when dimensional reduction is studied~\S\ref{sec:dimred}.

In \S\ref{sec:SUSY}, we discuss theories of $\N=1$ chiral multiplets in the BV formalism, and demonstrate the equivalence of the BV complex obtained in the auxiliary-field formulation with the higher-order BV complex for on-shell supersymmetry constructed by Baulieu \emph{et al.}~\cite{Baulieu-susy}.
The main novelty of this alternative formulation is of an off-shell description of $\N=1$ supersymmetry {\em without} the need to introduce auxiliary fields; the twist can be obtained equivalently from either description. 

The first part of \S\ref{sec:holo} recalls general facts about holomorphic field theories in the BV formalism. 
Then, we go on to compute the holomorphic twist of a theory of chiral matter with $F$-term interactions, both in the usual auxiliary field formulation and in the higher order BV setup.
We show that both BV complexes lead---in somewhat subtle fashion---to the same theory, the $\beta\gamma$ system with $F$-term interactions. 
We perform the twist as a family over the nilpotence variety of the supersymmetry algebra; in fact, both the tangent and normal bundles to the nilpotence variety play a role, the first generating deformations of complex structure on~$\R^4$ and the second being responsible for a holomorphic analogue of topological descent. 

In \S\ref{sec:characters} we give a definition of the higher dimensional local character for holomorphic theories in arbitrary dimension, generalizing the usual $q$-character definition for vertex algebras. 
We proceed to compute the character of the free $\beta\gamma$ system, demonstrating agreement with the $\N=1$ index as studied in~\cite{Rom}.
Furthermore, citing well-known results of supersymmetric localization in three and four dimensions, we show that the local character of the twisted theory agrees with its partition function on Hopf manifolds, which are complex surfaces topologically equivalent to~$S^1 \times S^3$. 
This is strong evidence of a holomorphic state-operator correspondence, which we conjecture that this class of holomorphic field theories satisfies.
Note that, while Hopf manifolds are not K\"ahler, they nevertheless have $SO(3)$ holonomy.
For a very thorough discussion of placing supersymmetric $\cN=1$ field theories on Hopf manifolds, as well as other $4$-manifolds, we refer the reader to~\cite{Closset1,Closset2}.
Our results presented here are complementary to this work, but are different in the sense that we study theories directly at the level of the holomorphic twist, and place the $\beta\gamma$ system on the Hopf manifold by exhibiting it as a quotient of punctured~$\C^2$, which is K\"ahler. We further argue that compactification of a general holomorphic theory in $d$ complex dimensions along~$S^{2d-1}$ gives rise to a topological quantum  mechanics, which is just a single associative algebra: in fact, this always takes the form of a dg Weyl algebra, and a natural module is given by the local operators in the original theory. (The quantum mechanics contains additional local operators, arising from nonlocal operators in the higher-dimensional theory wrapping a holomorphic $(d-1)$-cycle.) Proving an analogue of the Stone--von~Neumann theorem for a broad enough class of such dg~algebras should result in a general argument establishing the holomorphic state-operator correspondence.

In~\S\ref{sec:Fterms}, we  consider the spectral sequence induced by deforming the BV differential of the free theory by superpotential terms, and correspondingly compute the character in this case. 
Further, we introduce a chiral version of the Jacobi ring associated to a holomorphic potential, and show that it agrees with the holomorphic local operators of the theory. 
The fields on the $E_1$ page of the relevant spectral sequence consists of two copies of functions on the formal disk, tensored with dual vector spaces and placed in adjacent homological degrees; the differential arising from interactions, acting on operators, simply witnesses functions on this space as the Koszul complex for the  partial derivatives of the superpotential. One can thus interpret our result as an identification of the holomorphically twisted Wess--Zumino model with the sigma model on the derived critical locus of the superpotential. 

The goal of \S\ref{sec:symmetry} is to introduce an infinite-dimensional symmetry present in the twisted theory on $\CC^2$---by (derived) holomorphic functions on the spacetime, with values in the Lie  algebra of the flavor symmetry.
This type of symmetry was first discussed in~\cite{GwilliamWilliams} as a higher dimensional version of the Kac--Moody symmetry present in chiral CFT; the relevant algebras were also discussed in~\cite{FHK}. 
This symmetry is quite different from other versions of infinite chiral symmetries in four dimensions, for example in~\cite{BeemEtAl, JohansenKM}, in a few regards. 
Firstly, it is present even at the level of the twist of $\cN=1$ supersymmetry, rather than the $\cN=2$ required in~\cite{BeemEtAl}. As such, this algebra will also act in a holomorphic twist of any $\N=2$ theory, from which any other twist (including that considered in~\cite{BeemEtAl}) can be obtained by a further deformation of the differential. 
Secondly, this symmetry does not pick out any preferred $\CC$-plane (or a more general Riemann surface) inside of the four-dimensional manifold $\CC^2$; as such, we view it as more intrinsic symmetry of the twisted theory. 
At the level of algebras, however, our previous caveat is of course still valid: one must work in a derived way in order to see anything nontrivial. 
Indeed, instead of the state space having a symmetry by a Lie algebra (as occurs for affine Kac--Moody in complex dimension one), there is an $L_\infty$ algebra which acts: as previously mentioned, it arises as a central extension of derived sections of holomorphic functions on the complex manifold.

Finally, in \S\ref{sec:dimred} we demonstrate the compatibility of our calculations with dimensional reduction to a theory with in two dimensions along an arbitrary Riemann surfacae. 
In the case of a torus, this procedure produces the holomorphic 
twist of $\cN=(2,2)$; more generally, one finds a twist of a theory with $\cN=(0,2)$ supersymmetry. We also consider dimensional reduction along a plane which may not be a complex subspace of~$\C^2$, using our expression for the twist as a family over the space of complex structures on~$\R^4$; this produces either the holomorphic or $B$-type twist of the resulting $\N=(2,2)$ theory, and witnesses the spectral sequence between them~\cite{BPS-SS} as an instance of the Hodge-to-de-Rham spectral sequence.

\subsection*{Acknowledgements}
We thank K.~Costello, T.~Dimofte, R.~Eager, O.~Gwilliam, V.~Kac, D.~Pei, M.~Szczesny, and J.~Walcher for conversations and helpful advice related to many aspects of this work. I.S. thanks the Kavli Institute for Theoretical Physics, the Erwin-Schr\"odinger-Institut f\"ur mathematische Physik, the Center for Quantum Geometry of Moduli Spaces, and the Mathematisches Forschungsinstitut Oberwolfach for hospitality during the preparation of this work. 
B.W. thanks Northeastern University, the Banff International Research Station, the Aspen Center for Physics, and the Simons Center for hospitality during the preparation of this work.  
The work of I.S. was supported in part by the Deutsche Forschungsgemeinschaft, within the framework of the Exzellenzstrategie des Bundes und der L\"ander.
The work of B.W. was supported by Northeastern University and National Science Foundation Award DMS-1645877.

\section{$\N=1$ supersymmetry in four dimensions}
\label{sec:prelim}

We begin by recalling some standard facts and conventions of supersymmetric field theory.

\subsection{Conventions for operators in field theories}
\label{ssec:ops}

The building blocks of a classical field theory are its fields, which arise from local data on the spacetime manifold. 
The fields of a theory without defects define a locally free sheaf on the spacetime manifold $M$, and are typically given as the smooth sections of some (translation invariant) super vector bundle on~$M$. 
For instance, the \emph{chiral multiplet} in a four-dimensional $\N=1$ theory contains as component fields one complex scalar and one Weyl fermion:
\deq{ 
\phi \in \Map(\R^4,\C) = C^\infty(\RR^4) , \qquad
\psi \in \Map(\R^4,\Pi S_+) = C^\infty(\RR^4) \tensor_{\CC} \Pi S_+ .
}
These are both sections\footnote{By convention, sections or smooth functions are always {\em complex} valued.}
of trivial bundles on $\RR^4$, and we refer to the trivial bundle $\ul{\CC} \oplus \Pi \ul{S_+}$ simply as the chiral multiplet for the remainder of this section.%
\footnote{$\Pi(-)$ denotes parity shift with respect to the super grading.}

Suppose $E$ is a super vector bundle on a spacetime manifold $M$, defining a field theory whose {\em fields} are its sheaf~$\sE$ of smooth sections.
(In full generality, such as in gauge theories, $E$ will also carry a differential.)
If $x \in M$ is any point, we can speak of the {\em local operators} of the theory supported at $x \in M$. 
These are operators that depend algebraically (or formally algebraically) on the fields and their derivatives at the point $x$. 
Mathematically, the definition is the following.

\begin{dfn}
\label{dfn: susy ops}
Let $E$ be a super vector bundle on $M$ and $\sE$ its sheaf of smooth sections. 
The space of {\em local operators} of $\sE$ at $x \in M$ is the super vector space\footnote{Here, $\Hat{\Sym}(W) = \prod_{n \geq 0} \Sym^n(W)$ is the completed symmetric algebra.}
\[
  \Obs_x = \Hat{\Sym}(J^\infty E |_x)^\vee .
\]
Here, $J^\infty E$ denotes the super vector bundle of $\infty$-jets of $E$ and $J^\infty E|_x$ is its fiber at $x \in M$. 
\end{dfn}

In this work, we will mostly consider field theories defined on an affine space $\ST$, which will be $\R^n$ with a metric of either Euclidean or Lorentzian signature. (For most of the paper, $\ST$ will simply be Euclidean $\R^4$.) 
In this setting, it is natural to suppose that the bundle $E$ is translation invariant.
That is, we specify an isomorphism with the trivial bundle $E = \RR^n \times E_0$. (Note that translation invariance of a bundle is data, rather than a property!)
$E_0$ denotes the fiber of $E$ over $0 \in \RR^n$. 
For translation-invariant bundles, the spaces of local functionals over any two points are identified, so that it makes sense to write $O_x$ for the operator at $x\in\ST$ corresponding to $O_0 \in \Obs_0$.

\begin{eg}
  Take $E = \ul{\C}$ to be the trivial complex line bundle on~$\R^4$, so that $\phi \in \sE$ is a complex scalar field. An example of a local operator is given by 
  \deq[eq:exop]{
O_y(\phi) = \phi (y) \pdv{\phi}{x_1} (y) .
}
Since the chiral multiplet in four-dimensional $\N=1$ theories contains a complex scalar field, this expression will also define an operator in any $\N=1$ theory with chiral matter. 
\end{eg}
We note that it is standard in the physics literature to simply refer to an operator like~\eqref{eq:exop} using the notation
\[
\phi \pdv{\phi}{x_1} .
\]
This standard notation has the potential to lead to confusion between fields and operators. 
The distinction is conceptually important, though, as the following remark makes clear: while the fields of a theory are naturally a sheaf over the spacetime, operators with specified support naturally form a sort of \emph{(pre)cosheaf}.

\begin{rmk}
 As  we have emphasized, the fields of the theory are the smooth sections $\sE$ of the super vector bundle $E$. 
Intrinsically, the fields satisfy a sort of locality: they form a sheaf on spacetime, which for us is just $\ST$. 
One can define a more general class of operators (sometimes thought of as ``smeared'' operators in physics) by restricting the support not to  be pointlike, but to lie in a more general open set. That is, we could consider all functions on the sheaf $\sE$:
\[
\sO(\sE) = \Hat{\Sym}(\sE^\vee)  .
\]
Since $\sE$ is a sheaf, and we are taking an appropriate topological linear dual, this object behaves like a cosheaf: it makes sense to evaluate $\sO(\sE)$ on any open set $U \subset \ST$; and if $U \hookrightarrow V$ is an embedding of open sets, then there is a natural map $\sO(\sE)(U) \to \sO(\sE)(\ST)$. 
In fact, a more general structure is present:
the object $\sO(\sE)$ has the structure of a {\em factorization algebra} on $\ST$~\cite{CG1, CG2}.

This more general notion of an observable is related to the local operators we have just defined. 
We can evaluate the factorization algebra on a disk $D(x,r)$ centered at $x \in \ST$ to obtain the super vector space $\sO(\sE)(D(x,r))$. 
The local operators (with pointlike support) embed inside this space:
\deq{
  \Obs_x \hookrightarrow \sO(\sE)(D(x,r)).
}
In fact, there is a more precise relationship in the context of holomorphic theories, as we will point out in~\S\ref{sec: hol op}. 
\end{rmk}

The space of local operators is not quite the home for Lagrangians in a (supersymmetric) field theory. 
The difference is that we want to consider local operators that are only defined up total derivatives. 
The way to say this invariantly is the following. 
Notice that the bundle of jets $J^\infty E$ is not just a super vector bundle; in fact, it is a super $D$-module (in the appropriate super sense). 
In other words, it comes equipped with a canonical flat connection. 
Thus, as we vary the points $x \in \RR^n$ the local operators $\Obs_x$ also carry the structure of a $D$-module. 

Next, we recall the axiomatization of action functionals as integrals over Lagrangian densities. 
We borrow conventions from \cite{CosRenorm, CG2}. 

\begin{dfn}
\label{dfn: localfun}
The space of {\em local functionals} of a super vector bundle $E$ is 
\[
\oloc(\sE) = {\rm Dens}_{M} \tensor_{D_{M}} \sO_\text{red} (J^\infty E) .
\]
Here, $D_M$ is the algebra of differential operators on $M$, and $\sO_\text{red}(J^\infty E) = \prod_{n > 0} \Sym^n_{C^\infty_M} \left(J^\infty(E)^\vee\right)$ is the space of reduced functionals on jets.
\end{dfn}

A local functional encapsulates the data of a Lagrangian density defining a theory.
We will often write local functionals as operators, with the caveat that we are modding out by those functionals that are a total derivative. 

\begin{rmk}
  For a field theory on an affine spacetime $M \cong \R^n$, there is an action by the abelian (complex) Lie algebra of translations,
\begin{equation}
  \STlie = \CC^n = \Span_\CC \{\partial_{x_1}, \ldots, \partial_{x_n}\},
\end{equation}
on the space of local functionals. $\STlie$ is just $\ST_\C$, regarded as an abelian Lie algebra;
we will use $n$ for the real dimension of the spacetime, and later on $d = n/2$ for the complex dimension after the holomorphic twist.
We will mostly be interested in those local functionals that are invariant with respect to this action. 
Further, if $E$ is a translation invariant vector bundle on~$\ST$, there is an isomorphism
\[
\oloc(\sE)^{\CC^n} \cong \CC \cdot \d^n x \tensor^{\LL}_{U(\STlie)} \sO_\text{red} (J^\infty E|_0) .
\]
where we note that the algebra $U(\STlie) = \CC[\partial_{x_1}, \ldots, \partial_{x_n}]$ is precisely the (commutative) algebra of translation invariant differential operators. 
\end{rmk}

\begin{eg}
As an example of a local functional consider the free supersymmetric Lagrangian for the $\N=1$ chiral multiplet on $\RR^4$. 
It consists of the standard kinetic terms:
\deq{
  L = \left( - \partial\bar\phi \cdot\partial\phi + i \bar \psi \dslash \psi \right) \, \d^4x.
}
Note that this functional is manifestly translation invariant. 
\end{eg}

\subsection{Supersymmetry algebras and spinors in four dimensions}
\label{ssec:susyalg}

We now specialize from general field theories to supersymmetric field theories. By definition, these are theories in which the action of the affine transformations of the spacetime~$\ST$ are extended to a super Lie algebra. We also consider only four-dimensional theories here; as such, let $\ST = \R^{1,3}$ or~$\R^4$, corresponding to Lorentzian or Euclidean signature. (We will work with complexified algebras in any case, so that $\ST$ will usually denote $\C^4$; however, the signature will be relevant in a couple of places, which we will point out explicitly when they occur.) 

In order to discuss the supersymmetry algebra, we first fix a couple of conventions related to spinors in four dimensions. In any number of dimensions, the spinors of $\so(\ST)$ can be constructed by choosing a maximal isotropic subspace $L\subset V_\C$; the exterior algebra
\deq{
 D = \Lambda^* (L)
 }
 then carries the structure of a Clifford module. To give this structure, we just need to specify the action of~$V$ on~$D$ by Clifford multiplication; we recall that, in even dimensions,
 \deq{
 V_\C \cong L \oplus L^\vee.
 }
 $L$ is then taken to act by multiplication and $L^\vee$ by contraction. The commutator generates the pairing between~$L$ and~$L^\vee$, which is precisely the inner product on~$\ST$. Since the spin group sits inside of the even Clifford algebra, $D$ acquires a representation of~$\Spin(\ST)$. This representation is reducible, since the splitting
 \deq{
 D = \Lambda^\text{even}(L) \oplus \Lambda^\text{odd}(L)
 }
 is preserved by the action of~$\Lambda^2 V \cong \so(\ST)$. These irreducible spinor representations are called \emph{Weyl spinors} in the physics literature. 
 
We will always work with Weyl spinors, represented by symbols like $\psi$ or $\chi$; these transform in the complex two-dimensional chiral spinor representation $S_+ \cong \Lambda^\text{even}(L)$ of~$\so(4)$, constructed above. After using the exceptional isomorphism
\deq{
\so(4) \cong \su(2) \times \su(2),
}
the representation $S_+$ is the defining representation of the left $\su(2)$, tensored with the trivial representation of the right. 
The antisymmetric square of this representation is the trivial representation; we will write this pairing simply as $\psi\chi$, which (since spinor fields will have odd parity by spin and statistics) is meaningful independent of the order in which the symbols appear. 

An element of the $S_-$ representation will carry a bar, reflecting the fact that---in Lorentzian signature---the $S_-$ is the complex conjugate representation of the~$S_+$. 
For example, $\bar\psi$ denotes the conjugate of~$\psi$. In Euclidean signature, one must also apply  the automorphism of the algebra which interchanges the two $\su(2)$ factors; this is often denoted $\gamma^0$ in the physics literature. (In general, we will work in a complexified setting, and will not need to consider real structures.)
It is thus immediate that $\overline{(\chi \psi)} = \bar\chi \bar\psi$, where the overline denotes complex conjugation, and, just as the case of $S_+$, we have identified the anti-symmetric square of $S_-$ with the trivial representation. 
We will also make frequent use of the isomorphism 
\[
\Gamma : S_+ \otimes S_- \xto{\cong} \ST .
\]
We will also use the Feynman ``slash'' notation for the inverse inclusion $\ST \hookrightarrow S_+ \otimes S_-$, 
so that (for example) there is a linear differential operator
\deq{
\dslash: \Map(\R^4,S_+) \rightarrow \Map(\R^4,S_-).
}

Having settled these conventions, let us now return to the topic of supersymmetry. The four-dimensional $\N=1$ supersymmetry algebra is a super-Lie algebra with underlying super-vector space
\[
\sp = \sp^0 \oplus \sp^1.
\]
The superscripts denote the grading by~$\Z/2\Z$ corresponding to fermion parity; $\sp^1$ is therefore odd. Here, the bosonic part is of the form
\[
\sp^0 = \left[ \STlie \rtimes \lie{so}(\ST) \right] \oplus \lie{r},
\]
where $\STlie \rtimes \lie{so}(\ST)$ is the Poincar\'e algebra, which generates affine transformations of~$\ST$, and~$\lie{r}$, the $R$-symmetry, is in this case a one-dimensional (abelian) Lie algebra, which one may or may not choose to include. As an $\fp^0$-module, the fermionic part is
\[
\sp^1 = \Pi(S_+ \oplus S_-),
\]
where $S_\pm$ has charge $\pm1$ under~$\lie{r}$. Here 
 $\Pi$ denotes parity shift, with respect to the $\Z/2\Z$ grading, and the anticommutator map is just the isomorphism
\[
\Gamma : S_+ \otimes S_- \rightarrow \ST,
\]
extended by zero in the obvious way to a map from $\Sy^2(S_+ \oplus S_-)$.

As with any super-Poincar\'e algebra, $\fp$ has a normal $\Z/2\Z$-graded subalgebra $\st$ of supertranslations, of the form
\[
\st = \STlie \oplus \sp^1 \subset \sp .
\]
That is, $\st^1 \cong \sp^1$. We can think of~$\st$ as arising from~$\sp$ by forgetting the $\lie{so}(\ST)\times\lie{r}$ part of the algebra. However, we can remember the $\lie{so}(\ST)\times\lie{r}$-module structure; $A$ is then the extension of~$\lie{so}(\ST)\times\lie{r}$ by the module~$\st$.

The algebra $\sp$ also sits inside a larger algebra, the $\N=1$ \emph{superconformal} algebra in four dimensions, which is a simple super-Lie algebra:
\[
\sp \subset \sc = \su(2,2|1).
\]
Here $SO(\ST) \cong \su(2) \oplus \su(2)$ sits block-diagonally inside of $\su(2,2)$, and $\STlie$ is one of the off-diagonal blocks.
Notice that, while there is no requirement for the $R$-symmetry to be represented on a super-Poincar{\'e}-invariant theory---and in fact, it is often anomalous---it forms part of the simple algebra $\sc$, and therefore must be present in superconformal field theories.

\subsection{Supersymmetric field theories}
\label{ssec:actions}

By definition, a supersymmetric theory admits an action of the super-Poincar{\'e} algebra, extending the action of affine transformations of the spacetime on the fields. 
Ideally, this means we have an (strict) action of the Lie algebra $\sp$ on the fields of the theory $\sE$ in such a way that the classical Lagrangian is invariant. 
In practice, this is rarely the case, even for the free $\N=1$ chiral multiplet in four dimensions. 
What one can really find is that there is an action of the supersymmetry algebra on the critical locus of the action functional. 

It is natural to require that this action be ``local" in the sense that it is through differential operators.
A way to cast this is to require that the supersymmetry action determine a representation on $\oloc(\sE)$,
\deq{
\rho : \sp \to \End(\oloc(\sE)),
}
in such a way that the Lagrangian $L \in \oloc(\sE)$ defining the classical theory is fixed.%
\footnote{We have seen that $\oloc(\sE)$ is actually a sheaf on $\RR^n$, but for now we consider just its global sections.}

We will find it convenient to repackage the data of this algebra action using standard manipulations of Koszul duality. For physicist readers, this is analogous to the way in which the BRST formalism repackages the procedure of taking gauge invariants: the BRST differential encodes the gauge algebra and its action on the fields. We are free to repackage any symmetry in this fashion; if we do not wish to take (co)invariants, we can just remember the differential without passing to (co)homology.

Recall that the Chevalley--Eilenberg complex of a Lie algebra is defined by
\deq{
  \clie^*(\lie{g}) = \left[ \Sym\left( \lie{g}^\vee[-1] \right), \dCE \right],
}
where $\dCE$ is the \CE{} differential.
It is a degree $+1$ operator, obtained by extending the dual of the Lie bracket on~$\lie{g}$ according to the Leibniz rule; the data of such a degree-one nilpotent differential is precisely equivalent to a Lie algebra structure on~$\lie{g}$ provided $\fg$ is concentrated in degree zero. 
If we relax the condition that $\fg$ is concentrated in degree zero, then we obtain the structure of an $L_\infty$ algebra on~$\lie{g}$. 
Furthermore, if $\lie{g}$ is a super-Lie algebra, the same definition applies (with the condition that overall parity is determined by the sum of the homological degree and intrinsic parity). 
We also note that a Lie algebra structure on~$\lie{g}$ and a $\lie{g}$-module structure on~$M$ are together precisely equivalent to a degree-one nilpotent differential on
\deq{
  \Sym\left( \lie{g}^\vee[-1] \oplus M^\vee \right).
  }
(The physicist reader may think of the collection of all ordinary and ghost fields in the BRST formalism.)

  It is further standard that the data of a map (or more generally, an $L_\infty$ map) of Lie  algebras $\rho : \fg \to \fh$ is equivalent to the data of a Maurer--Cartan element in the dg Lie algebra
\[
\theta_{\rho} \in \clie^*(\fg) \tensor \fh .
\]
Here, we use the commutative dg algebra structure on $\clie^*(\fg)$ together with the Lie bracket on $\fh$.
The Maurer--Cartan equation for $\theta_{\rho}$ is
\[
\dCE_\fg \theta_{\rho} + \frac{1}{2} [\theta_{\rho} , \theta_{\rho}]_{\fh} = 0 .
\] 
where $\dCE_\fg$ is the \CE{} differential for $\fg$ and $[-,-]_{\fh}$ is the Lie bracket on $\fh$. (Further terms would appear if $\lie{h}$ were an $L_\infty$ algebra rather than simply Lie.)

Thus, another way to encode $\N=1$ supersymmetry is to prescribe a Maurer--Cartan element in
\begin{equation}
  \theta_{\rho} \in \clie^*(\sp) \tensor \End(\oloc(\sE)) ,
\end{equation}
or equivalently a BRST-type differential acting in the space
\deq{
  \clie^*(\sp) \otimes \oloc(\sE)
}

One can think of  this as adding ghosts that do not depend on the spacetime; thus, the sheaf is the constant sheaf with value $\sp$, rather than the sheaf of sections of the bundle $\underline{\sp}$ as is typical in gauge theories. See~\cite{Baulieu-susy} for an example of this technique in the physics literature.
Restricting to just the supertranslation algebra, 
we can decompose such a Maurer--Cartan element as
\[
  \theta_{\rho} = \sum_{i=1}^4 \delta_{x_i} + \delta_{Q} + \delta_{\smash{\bar Q}} .
\]

\subsection{Twisting, the nilpotence variety, and~$B\st$}
\label{ssec:Bg}

We will be interested in \emph{holomorphic twists} of supersymmetric field theories. At root, this means that one passes to the cohomology of a chosen nilpotent element in the supersymmetry algebra. One can also think of this as giving a vacuum expectation value to the corresponding ghost; although the ghosts in our setting are spacetime-independent and nondynamical, this description also makes sense for theories with local supersymmetry, where it recovers the proposal of Costello and Li for defining twists of supergravity theories~\cite{CostelloLi}. For recent reviews of the twisting procedure and classifications of possible twists in different dimensions, see for example~\cite{NV,Chris}.

The moduli space of allowed ``vacuum expectation  values'' for the ghosts of supertranslations is nothing other than the space of nilpotent elements in~$\st^1$, which is an algebraic variety $Y$ defined by homogeneous quadratic equations:
\deq{
  Y = \{Q \in \st^1 : Q^2 = 0 \}.
}
For four-dimensional minimal supersymmetry, it is easy to see that
\deq{
  Y = S_+ \cup_{\{0\}} S_- \subset S_+ \oplus S_- = \st^1.
}
Thus, if $\O(Y)$ denotes the algebra of functions on $Y$, one finds
\deq[eq:OY]{
  \O(Y) = \O(\st^1[-1]) / \langle S_+^\vee S_-^\vee \rangle .
}
Note that $Y$ is an ordinary (not super) affine variety and $\O(Y)$ is its ordinary algebra of functions---the homogeneous coordinate ring of the corresponding projective variety, which is here the disjoint union of two projective lines in~$P^3$. 

On general grounds, twists of a particular supersymmetric theory can therefore be thought of as a family over the corresponding nilpotence variety~$Y$. We will see this explicitly for minimal supersymmetry in four dimensions, where the result will be a family of holomorphic theories over the space of complex structures on~$\R^4$; in general, nilpotence varieties that do not arise from dimensional reduction are closely related to spaces of complex structures, since only holomorphic twists are present. The nilpotence variety is also closely related to the classifying space of the super-Lie algebra~$\st$~\cite{NV}; this makes sense, since the twist construction to yield a family over~$B\st$ being a module over $\st$. 

Recall that the classifying space~$B\lie{g}$ of a Lie algebra~$\lie{g}$ is the formal derived space whose space of functions is modeled by the \CE{} complex: by definition,
\deq{
  \O(B\lie{g}) \equiv \clie^*(\lie{g}) .
}
For $\lie{g}$ semisimple, $B\lie{g}$ is the de~Rham stack of (a compact real form of) the corresponding group~$G$.

For super-Lie algebras, as we recalled above, the \CE{} complex is graded by both homological degree and intrinsic parity, $\Z \times \Z/2\Z$. In the case of the supertranslation algebra~$\st$, this grading lifts to a bigrading by~$\Z \times \Z$, since $\st^0$ is abelian and $\st^1$ is trivial as a $\st^0$-module. The result can be given in the form 
\deq[eq:BST]{
\O(B\st) 
=  \left[ \O(\st[1]), \dCE \right]
= \left[  \Sy^*(\st^1)^\vee \otimes \Lambda^* (\st^0)^\vee, \dCE \right],
}
where $\dCE$ carries degree $(2,-1)$ with respect to the two indicated gradings, and the original homological grading is their sum. Of course, everything here is equivariant with respect to the $\so(\ST) \times \lie{r}$-module structure. Moreover, the differential is in fact $\O(\st^1[-1])$-linear, so that~\eqref{eq:BST} becomes the Koszul complex (with respect to the second grading) of the defining ideal of the nilpotence variety, in free $\O(\st^1[-1])$-modules; its zeroth homology is then just the homogeneous coordinate ring~\eqref{eq:OY}. 

\subsubsection{Pure spinor superfields}
\label{sssec:PSSF}
In addition to classifying the possible twists of a supersymmetric field theory, the nilpotence variety can also be used to construct field representations of the corresponding super-Poincar{\'e} algebra. We briefly recall this technique here; later, we will show how this formalism simplifies the computation of the holomorphic twist in four-dimensional theories. For more detail, we refer the reader to~\cite{Cederwall,NV}.

The construction begins by observing that there is a canonical scalar element
\deq{
  \id \in \End(\st^1) \cong \st^1 \otimes (\st^1)^\vee.
}
We can then push this element forward to $\st \otimes_\C \O(Y)$, using the natural inclusion maps on either side. The result is a scalar nilpotent operator, the \emph{Berkovits differential}, which we will denote $\Berk$. It acts in the tensor product of any $\st$-module $M$ with any $\O(Y)$-module $\Gamma$; in the case where the latter is just $\O(Y)$ itself, one can think of this as forming the trivial bundle over~$Y$ with fiber $M$, and then acting by the obvious tautological bundle whose fiber over~$Q\in Y$ is spanned by the operator~$Q$ itself. Moreover, the whole construction is Lorentz invariant if we insist that $M$ and~$\Gamma$ are equivariant modules (so that $M$ is in fact an $\sp$-module). $\Gamma$ can be thought of geometrically as the space of global sections of an equivariant bundle or coherent sheaf.

In general, $H^*(M \otimes \Gamma, \Berk)$ does not admit an action of the full~$\sp$; $\sp^0$ is guaranteed to commute with~$\Berk$ (since it is a scalar), but $\sp^1$ may not. However, we can make a specific choice for~$M$ such that the homology becomes a field representation of~$\sp$ for \emph{any} choice of~$\Gamma$.  Namely, we can take $M$ to be the space of \emph{free superfields}, i.e., the algebra of functions $\O(\st)$.  This admits two commuting actions of~$\st$, by left and right translations. The action of $SO(\ST)$ comes from pullback under the (adjoint) $SO(\ST)$ action on~$\ST$. 
We are therefore free to apply the Berkovits differential using the \emph{right} action, and are guaranteed that $\sp$ will still act on the left on its homology. It follows that 
\deq{
  H^* \left(  \O(\st) \otimes \Gamma, \Berk \right)
}
is a supermultiplet corresponding to the equivariant sheaf $\Gamma$ on~$Y$. Furthermore, $\O(\st)$ is naturally bigraded, and~$\Berk$ decomposes into the sum of two bigraded pieces; one is independent of~$\Sy^*(\st^0)^\vee$, whereas the other involves even translations (derivatives). There is thus a spectral sequence beginning at $\O(\st) \otimes \Gamma$ and abutting to the homology of~$\Berk$; the $E_1$ page contains the component fields of the resulting supermultiplet, and the differentials on following pages correspond to a BV differential. 

In the context of $\N=1$ supersymmetry in four dimensions, the vector multiplet arises from the structure sheaf of~$Y$, whereas the chiral multiplet comes from 
\deq{
  \Gamma = \O (S_+),
}
considered as an $\O(Y)$-module in the obvious way as a quotient of~$\O(Y)$ with respect to the ideal generated by~$S_-^\vee$. The complex $\O(\st) \otimes \Gamma$ thus reduces to
\deq{
  \left[ \left( \O(\ST) \otimes \Lambda^* S_+^\vee \otimes \Lambda^* S_-^\vee \right) \otimes \Sy^* (S_+^\vee), \Berk \right],
}
where $\Berk$ acts on the $E_0$ page by the identity on the two copies of~$S_+$, so that its homology is simply
\deq{
  H^* \left(  \O(\st) \otimes \Gamma, \Berk \right)
  \cong \O( \ST ) \otimes \Lambda^* S_-^\vee.
}
The reader will recognize this as the chiral multiplet of four-dimensional minimal supersymmetry, which we review in more detail in the component formalism in the following section. For degree reasons, no BV differential can appear here.

\section{The chiral multiplet and auxiliary fields}
\label{sec:SUSY}

\subsection{Component formalism}
\label{ssec:components}
The reader will recall that the chiral multiplet in a four-dimensional $\N=1$ theory contains as component fields one complex scalar and one Weyl fermion:
\deq{ 
(\phi, \psi) \in \Gamma (\RR^4, \ul{\CC} \oplus \Pi \ul{S}_+) = \O(\RR^4) \tensor (\CC \oplus \Pi S_+)
}
together with
\deq{
(\Bar{\phi}, \Bar{\psi}) \in \Gamma (\RR^4, \ul{\CC} \oplus \Pi \ul{S}_-) = \O(\RR^4) \tensor (\CC \oplus \Pi S_-)
}
where $\psi, \Bar{\psi}$ have opposite chirality. 

The free Lagrangian for this multiplet just consists of the standard kinetic terms for each field:
\deq{
L_{\rm free} = - \partial\bar\phi \cdot\partial\phi + i \bar \psi \dslash \psi .
}
If $E = \ul{\CC} \oplus \Pi \smash{\ul{S}}_+$ is the underlying super vector bundle of the chiral multiplet, this Lagrangian is a local functional in $L_0 \in \oloc(\sE)$ .

The action $L_{\rm free}$ is invariant (up to a total derivative) under the transformations
\begin{equation}
\label{transf-component}
\begin{aligned}[c]
\delta \phi &= \epar \psi + a^\mu \partial_\mu \phi, \\
\delta \psi &= i \bar\epar (\dslash \phi )  + a^\mu \partial_\mu \psi.
\end{aligned}
\end{equation}
Here $\epar \in (\st^1)^\vee[-1]$ and~$a^\mu \in (\st^0)^\vee[-1]$ are generators of the Chevalley--Eilenberg complex, and the notation $\epar\psi$ represents spinor contraction as defined above in~\S\ref{ssec:susyalg}.
Of course, the complex conjugates of these also hold:
\begin{equation}
\begin{aligned}[c]
\delta \bar\phi &= \bar\epar \bar\psi + a^\mu \partial_\mu \bar\phi, \\
\delta \bar\psi &= -i  (\dslash \bar\phi ) \epar + a^\mu \partial_\mu \bar\psi.
\end{aligned}
\end{equation}
Here, we are encoding the action of the supersymmetry algebra using the techniques of~\cite{Baulieu-susy}, as reviewed above in~\S\ref{ssec:actions}. Global (i.e.\ spacetime-independent), nondynamical ghosts are included, and the differential $\delta$ of \CE{} type  encodes the structure of the supertranslation algebra as well as the module structure on the physical fields. 
While we do not pass to cohomology of this differential, we can conveniently recover the differential arising in any twist of the theory by setting the global supersymmetry ghosts $\epar$ to appropriate nonzero values---i.e., to any point on the nilpotence variety.

The ghosts are \emph{bosonic} spinor variables $\epar$ and~$\bar\epar$, along with a fermionic vector $a^\mu$, all carrying ghost number one.
The differential acts on fields according to~\eqref{transf-component}, and 
the structure of the supersymmetry algebra is encoded in the action of the differential on the ghosts themselves:
\begin{equation}
\begin{aligned}[c]
\delta a^\mu &= \bar\epar \gamma^\mu \epar, \\
\delta \epar &= 0, \\
\delta \bar\epar &= 0.
\end{aligned}
\end{equation}

Now, in order for us to have an action of~$\sp$, the Maurer--Cartan condition $\delta^2 = 0$ must hold. 
However, this is not true using the naive action above, since $\delta^2 \psi = \delta(\delta \psi)$ is in fact nonzero;
in other words, the naive definition of $\delta$ does not define a map of Lie algebras $\sp \to \End(\oloc(\sE))$. 
However, $\delta^2$ lies in the ideal generated by the equations of motion, so that~\eqref{transf-component} \emph{does} define an action of the supersymmetry algebra on-shell (i.e., on the sheaf of solutions to equations of motion). This is a general feature of supersymmetry multiplets; when only physical fields are included in the multiplet, closure of the algebra requires that the equations of motion for the fermion be imposed.

In our case, the well-known work-around for this issue is to introduce an {\em auxiliary field}. 
This has the affect of modifying the space of fields $\sE$ to a larger space $\Tilde{\sE}$ that admits an off-shell action of $\sp$, without changing the sheaf of solutions to equations of motion or its representation of~$\sp$. An auxiliary-field formalism is not always available; we will see a general technique for avoiding the introduction of auxiliary fields below in~\S\ref{sec: bv aux}.

For the chiral multiplet, the auxiliary field is an element
\deq{
F \in C^\infty (\RR^4) ,
}
that appears algebraically in the action functional. 
In particular, the free Lagrangian is extended to
\deq[freeF]{
  \Tilde{L}_{\rm free} = - \partial\bar\phi \cdot\partial\phi + i \bar \psi \dslash \psi + \bar{F} F,
}
so that $F$'s equation of motion in the free theory simply sets $F = \bar F = 0$. 
If the bundle associated to this larger space of fields is 
\deq{
\Tilde{E} = E \oplus \ul{\CC} = \ul{\CC} \oplus \Pi \ul{S}_+ \oplus \ul{\CC},
}
then this Lagrangian is a local functional $\Tilde{L}_0 \in \oloc(\Tilde{\sE})$. The  reader will now recognize that
\deq{
  \Tilde{\sE} = \O(\ST) \otimes \Lambda^* (S_+^\vee),
}
matching the pure spinor superfield construction above.

The modified supersymmetry transformations now read
\begin{align}
\delta \phi &= \epar \psi, \nonumber \\
\delta \psi &= i \bar\epar (\dslash \phi ) +  \epar F, \label{susy-F} \\
\delta F &= - i \bar\epar \dslash \psi. \nonumber
\end{align}
(The complex conjugates of these are also valid.) After restoring the obvious action by ordinary bosonic translations, which  is suppressed above, the differential defined by~\eqref{susy-F} is now nilpotent, signaling closure of the algebra off-shell.
Indeed, $\delta$ now defines a map of Lie algebras
\deq{
\sp \to \End(\oloc(\Tilde{\sE})) 
}
in such a way that $\Tilde{L}_0$ is preserved.

\subsection{Superpotential interactions}
\label{ssec:F-terms}

In this section, we review supersymmetry-preserving interactions of $\N=1$ chiral multiplets, in the formalism with auxiliary fields; these are parameterized by a holomorphic function known as the superpotential. 

We consider a theory with $n$ chiral superfields, labeled with an index $i$.
In the formalism with auxiliary fields, the most general su\-per\-sym\-me\-try-pre\-ser\-ving interaction term that can be added to the free Lagrangian---$n$ copies of~\eqref{freeF}---is 
\deq{
L_\text{int} = - \frac{1}{2} W^{ij} \psi_i \psi_j + W^i F_j + \text{c.c.},
}
where $W^i$ denotes the corresponding derivative of a holomorphic function of the bosonic chiral fields $\phi_i$.
For example, 
\deq{
W^{ij} = \frac{d^2}{d\phi_i \, d\phi_j} W.
}
While $W$ is an arbitrary holomorphic function, only the quadratic and cubic terms in~$W$ will provide relevant interaction terms; these are (respectively) the mass matrix of the $\phi$ fields and the 
scalar self-couplings for the theory, and supersymmetry invariance fixes all other relevant couplings (i.e.\ fermion masses and cubic and quartic scalar vertices) as functions of  these. The terms at linear order only shift the action functional by a constant; we will consider only $W$ of quadratic and higher order.

When superpotential interactions are included, the auxiliary field $F$ still appears algebraically in the Lagrangian, but new terms appear:
\deq{
L = \bar{F}^i F_i + W^i F_i + \bar{W}_i \bar{F}^i +  \text{$F$-independent terms}.
}
Thus, the equations of motion for~$F$ are deformed to
\deq[F-EOM]{
F_i   = \bar{W}_i,  \qquad \bar{F}^i = W^i,
}
and the supersymmetry transformations of the fermionic component fields $\psi_i$  are correspondingly deformed after  $F$ is eliminated, as can be read off by substituting~\eqref{F-EOM} into~\eqref{susy-F}:
\deq[susy-deformed]{
\delta \psi_i = i \bar\epar (\dslash \phi_i ) +  \epar \bar{W}_i.
}

\subsection{The BV formalism}
\label{ssec:BV}

We quickly remind the reader of the basics of the BV formalism, mostly to fix conventions. For more detailed discussion, we refer to~\cite{CG2, CosRenorm}, the article~\cite{Schwarz-BV}, or the review in~\cite{CY4}.

In the BV formalism for classical field theories, one is interested in studying the sheaf of solutions to the classical equations of motion. These are just critical points for the action functional $S = \int L$:
\deq{
  \Crit(S) \subset \sE.
}
One now resolves functions on the critical locus (i.e., classical observables) freely in functions on~$\sE$, using the Koszul complex of the equations of motion. This amounts to constructing the \emph{BV fields} of the theory,
\deq{
  \sB = T^*[-1] \sE,
}
which is a sheaf of shifted symplectic vector spaces, obtained (just as before) as the sections of a dg vector bundle $B \rightarrow \ST$ whose fibers are the shifted cotangent bundles to the fibers of the original bundle~$E$. Since $\sB$ is a shifted symplectic space, its algebra of functions $\O(\sB)$ carries a shifted Poisson structure; this bracket structure is usually called the \emph{antibracket} in the physics literature.

In order to resolve the critical locus, the classical BV observables $\O(\sB)$ are equipped with the differential $\ad_S = \{S,-\}$ generated by the action functional under the antibracket. In the simplest case, when $\sE$ carries no differential, this makes $\O(\sB)$ into the Koszul complex for the equations of motion of the original theory, and no further modification is needed. More generally, though, the action must be modified so as to generate the differential on~$\sE$ while maintaining nilpotence; the latter requirement amounts to imposing the \emph{classical master equation}, $\{\mathfrak{S},\mathfrak{S}\} = 0$. (Here $\mathfrak{S} = \int \mathfrak{L}$ is the BV action functional.)
One can think of this as finding an appropriate lift of $L \in \oloc(\sE)$ to~$\mathfrak{L} \in \oloc(\sB)$, such that its restriction to the zero section returns $L$, and the first-order terms in antifields generate the internal differential on~$\sE$; in the end, the BV differential will do the job of passing  to the quotient of~$\Crit(S)$ by the action  of the gauge group. 

These  two  requirements fix the  terms of the BV Lagrangian~$\mathfrak{L}$ that are of order  zero and one in antifields; higher-order  terms,  if any, are generated by  requiring  the classical master  equation. To   quantize the theory, one tries to  deform the BV action to a  solution of the \emph{quantum  master  equation}, which is a  deformation of the  classical master equation by the BV Laplacian; see~\S\ref{ssec:BVquant} below. In general, for an honestly nilpotent internal differential on~$\sE$, a BV action that is linear in antifields will suffice, and the master equation  will already be satisfied. However, one can even extend the construction of the BV action to cases where the internal differential on~$\sE$ is only nilpotent modulo the defining ideal of~$\Crit(S)$~\cite{Baulieu-susy}. We will see an example of a BV action that is higher order in antifields below.

\subsection{BV actions for $\N=1$ chiral multiplets}
\label{sec: bv aux} 

Using the above procedure, we can now construct the BV Lagrangian at linear order in antifields, beginning with the chiral multiplet in the auxiliary-field formalism: 
\deq{
\mathfrak{L}
= \widetilde{L} + \phi^* \delta \phi + \psi^* \delta \psi + F^* \delta F + \text{c.c.}
}
Note that we are here treating $\delta$, which represents the action of the super-Poincar\'e algebra, as the internal differential on~$\widetilde{\sE}$! This amounts to constructing the BV action equivariantly with respect to the supersymmetry algebra. (The physical chiral multiplet, of course, has no internal differential.)

Since the supersymmetry transformations close on-shell, we are assured that the BV differential (i.e., the adjoint action of $\mathfrak{S}$ under the antibracket) is nilpotent. Indeed, for our purposes, it is sufficient to note that this continues to hold when the $a$-ghosts are set to zero---as long as we will eventually choose values for the bosonic $\epar$ ghosts that lie in the nilpotence variety of the supersymmetry algebra. And taking this choice as internal differential on~$\widetilde{\sE}$ is precisely passing to the holomorphic twist of  the  theory.
Since this is our aim, we discard the $a$-ghosts now, obtaining
\deq[BVaction-F-sum]{
\mathfrak{L}
= \widetilde{L}_\text{free} + \widetilde{L}_\text{int} + \widetilde{L}_\text{a.f.},
}
where
\begin{equation}
\begin{aligned}[c]
  \widetilde{L}_\text{free} &= - \partial \bar\phi \cdot \partial\phi + i \bar\psi \dslash \psi + \bar F F, \\
  \widetilde{L}_\text{int} &= - \frac{1}{2} W^{ij} \psi_i \psi_j + W^i F_j + \text{c.c.}, \\
  \widetilde{L}_\text{a.f.} &= \phi^* \epar\psi + \psi^* \epar F + i \psi^* \bar\epar \dslash \phi - i F^* \bar\epar \dslash \psi + \text{c.c.}
\end{aligned}
\label{BVaction-F}
\end{equation}
An implicit summation over flavors is understood in the first and third terms.

\subsubsection{Eliminating the auxiliary field}

We will now integrate out the auxiliary field, reducing the BV fields from $T^*[-1]\widetilde{\sE}$ to~$T^*[-1]\sE$ and producing a BV action without auxiliary fields. (Since the procedure is analogous to symplectic reduction, we will have to first restrict to the observables that have zero antibracket with the auxiliary field; this amounts to throwing out its antifield.) Notice that in~\eqref{BVaction-F}, the equation of motion for the auxiliary field $F$ is deformed to
\deq{
\bar{F}^i + W^i + \epar \psi^{*i} = 0.
}
(The equation of motion for~$\bar F$ is just the complex conjugate of the above.) Substituting these back into~\eqref{BVaction-F-sum}, and setting the antifields of auxiliary fields to zero, yields the BV action for the chiral multiplet in component formalism:
\begin{multline}
\mathfrak{L} = - \partial \bar\phi \cdot \partial\phi + i \bar\psi \dslash \psi + (W + \epar \psi^*)(\bar W + \bar\epar \bar\psi^*)  \\
- \frac{1}{2} W^{ij} \psi_i \psi_j - W^i (\bar{W}_i + \bar\epar \bar\psi^*_i ) 
+ \phi^* \epar\psi + i \psi^* \bar\epar\dslash \phi - \psi^{*i} \epar ( \bar{W}_i + \bar\epar \bar\psi^*_i ) + \text{c.c.}
\end{multline}
After expanding terms and cleaning this up, one finds
\deq{
  \mathfrak{L} = L_0 + L_1 + L_2,
}
where
\begin{equation}
  \begin{aligned}[c]
    L_0 &= - \partial \bar\phi \cdot \partial\phi + i \bar\psi \dslash \psi - \frac{1}{2} W^{ij} \psi_i \psi_j - \frac{1}{2} \bar{W}_{ij} \bar\psi^i \bar\psi^j -  |W|^2, \\
    L_1 &= \phi^* \epar\psi + i \psi^* \bar\epar\dslash \phi - \epar\psi^{*i} \bar{W}_i + \text{c.c.}, \\
    L_2 &= - \epar\psi^{*i} \bar\epar \bar\psi^*_i.
  \end{aligned}
\label{BVaction-elim}
\end{equation}

A few words of explanation are warranted here. Firstly, $L_0$ consists of terms that are independent of antifields; it reproduces the standard component action of the theory with superpotential interactions, without auxiliary fields. The antifield-dependent terms reflect the supersymmetry transformations of the fields; recall  that, because we have already set the ghost parameters for translations to zero, our differential will be nilpotent only for ghost parameters $\epar$ corresponding to nilpotent supercharges. Nonetheless, the transformations in~\eqref{BVaction-elim} neatly reproduce those in~\eqref{transf-component}, with an additional term that is dependent on the superpotential and give rise to the interaction spectral sequence in this formalism (arising from the terms $\epar\psi^{*i} \bar{W}_i$). This additional term is the one appearing in~\eqref{susy-deformed}; for this reason, the interaction no longer affects only the antifield-independent portion of the action, and we do not separate the free and the interaction terms explicitly. 

Furthermore, the action~\eqref{BVaction-elim} is no longer linear in antifields. A quadratic term $L_2$ has appeared upon elimination of the auxiliary fields, corresponding to the fact that the supersymmetry transformations no longer define a module structure on the space of off-shell fields (although they \emph{do} define a module structure modulo the equation of motion for~$\psi$, as we remarked previously). 

Had we started just with the component action and the transformations~\eqref{transf-component}, we could have used the techniques of~\cite{Baulieu-susy} to produce this BV action directly, without any reference to an auxiliary field formalism. We would have written down $L_0$ and~$L_1$ based on that information, we then would have found that, rather than being zero, $\{S,S\}$ would have been proportional to the Dirac equation for~$\psi$. Choosing $L_2$ so as to cancel that term would have produced the  solution~\eqref{BVaction-elim}, and the process of solving order by order in antifields would then be complete. These techniques apply even in circumstances where no auxiliary-field formalism is available, and produce BV complexes with actions quadratic in antifields that are analogous to~\eqref{BVaction-elim}.

\section{Twisted chiral matter: the $\beta\gamma$ system}
\label{sec:holo} 

The type of twists of supersymmetric theories we study in this paper are not of the familiar topological flavor, but are {\em holomorphic}. 
Like a topological twist, these holomorphic twists do not depend on the underlying metric data as the starting supersymmetric theory does, but unlike a topological twist, they depend on the complex structure. 

The starting point for a holomorphic theory is that of a holomorphic vector bundle equipped with a holomorphic version of the BRST operator. 
All of the theories we consider in this section will be in the BV formalism, and there is a suitable holomorphic version of the BV bracket. 
Given a holomorphic theory there is a natural way to construct a BV theory, which we will refer to as its {\em BV-ification}. 

\noindent {\bf Note:} Unless otherwise stated, we will work in the BV formalism for the remainder of the paper.
Thus, when we refer to fields we mean the full space of BV fields, and the action functional is the full BV action.

\subsection{Holomorphic field theory} \label{sec: holtheory}

In this section we set up notations and conventions for holomorphic field theories. 
We mostly follow the the approach to this subject presented in \cite{BWhol}. 

First, we define the appropriate notion of a ``free" theory. 

\begin{dfn}
A {\em free holomorphic theory} on $\CC^d$ consists of the following data:
\begin{itemize}
\item a $\ZZ$-graded complex vector space $Z^\bullet$;
\item a non-degenerate pairing of cohomological degree $(d-1)$
\[
\omega^\text{hol} : Z^\bullet \times Z^\bullet \to \CC \cdot \d^d z [d-1] ;
\] 
\item a holomorphic differential operator of cohomological degree $+1$:
\[
Q^\text{hol} : \sO^\text{hol}(\CC^d) \tensor Z^\bullet \to \sO^\text{hol}(\CC^d) \tensor Z^\bullet ;
\]
\end{itemize}
such that: $Q^\text{hol}$ is graded skew self-adjoint to $\omega^\text{hol}$, and $(Q^\text{hol})^2 = Q^\text{hol} \circ Q^\text{hol} = 0$. 
\end{dfn}

We use the notation $\CC \cdot \d^d z$ to indicate the fiber of the holomorphic canonical bundle at the origin in $\CC^d$. 
In the definition, we extend $\omega^\text{hol}$ to a pairing $\omega^\text{hol} :  \sO^\text{hol}(\CC^d) \tensor Z^\bu \times \sO^\text{hol}(\CC^d) \tensor Z^\bu \to \sO^\text{hol}(\CC^d) \cdot \d^d z [d-1] = \Omega^{d,hol}(\CC^d) [d-1]$ by $\sO^\text{hol}(\CC^d)$-linearity. 

\begin{rmk}\label{rmk: z2}
There is a weakening of this definition that is relevant for us. 
Instead of starting with a $\ZZ$-graded vector bundle, we can consider a $\ZZ/2$-graded vector bundle $Z^\bullet = Z_\text{even} \oplus Z_\text{odd}$.
We then require that $Q^\text{hol}$ is an odd holomorphic differential operator and that $\omega^\text{hol}$ be odd when $d$ is even, and even when $d$ is odd. 
\end{rmk}

The reader may observe that there is no ``space of fields" in the definition of a free holomorphic theory. 
One obtains the fields of the resulting free field theory, in the BV formalism, by taking the Dolbeault complex with values in the trivial holomorphic vector bundle with fiber $Z^\bullet$.
That is, the space of BV fields (including ghosts, fields, anti-fields, etc.) with its linear BV differential, is the complex 
\[
\sB_Z = \left(\Omega^{0,*}(\CC^d, Z^\bullet), \dbar + Q^\text{hol}\right) .
\] 
If $\alpha \in \sE_X$ denotes a section, the free Lagrangian is
\[
L_{\rm free} = \omega^\text{hol} (\alpha, (\dbar + Q^\text{hol}) \alpha) \d^d z 
\] 
where $\d^d z$ is the standard holomorphic volume form on $\CC^d$. 
We think of the passage from $Z^\bullet$ to $\sB_X = \Omega^{0,*}(\CC^d , Z^\bullet)$ an assignment that takes a holomorphic theory (as we've defined it) to the data of a BV theory.

We remark on a few essential points:
\begin{itemize}
\item[(1)] 
There is a $\ZZ$-grading on $\sB_Z$ given by the totalization of the internal grading of $Z^\bullet$ and the natural grading on Dolbeault forms.
This is the ghost grading of the BV theory.
\item[(2)] The non-degenerate pairing $\omega^\text{hol}$ extends to the space of fields by $\Omega^{0,*}(\CC^d)$-linearity. 
The operator $Q^\text{hol}$ extends to the Dolbeault complex since $\Omega^{0,*}(\CC^d)$ is a resolution for holomorphic functions.
\item[(3)] Since $Q^\text{hol}$ is holomorphic and $(Q^\text{hol})^2 = 0$, the total linear BV differential satisfies $(\dbar + Q^\text{hol})^2 = 0$. 
Thus, we have the following complex of BV fields $(\sB_Z, \dbar + Q^\text{hol})$ which resolves $(\sO^\text{hol}(\CC^d) \tensor Z^\bullet, Q^\text{hol})$. 
\end{itemize}

Here is the most basic, and perhaps most important for this paper, example of a free holomorphic theory. 

\begin{eg}
{\em The $\beta\gamma$ system on $\CC^d$.}
Suppose $V$ is any $\ZZ$-graded vector space.
Given this data, there is a natural holomorphic theory defined on $\CC^d$, for any $d$. 
The graded vector space underlying the fields of the $\beta\gamma$ system with values in $V$ is
\[
Z^\bullet = V \oplus \d^d z \cdot V^\vee [d-1] ,
\]
equipped with $Q^\text{hol} = 0$.
Here, $\d^d z\cdot  V^\vee$ is meant to indicate that we are looking at the fiber of the vector bundle $K_{\CC^d} \tensor V^\vee$ at $0 \in \CC^2$.
The pairing $\omega^\text{hol}$ is given by the obvious evaluation pairing $\<-,-\>$ between $V$ and its linear dual $V^\vee$. 
The resulting fields in the BV formalism are given by the Dolbeault complex
\[
(\gamma, \beta) \in \Omega^{0,*}(\CC^2) \tensor V \oplus \Omega^{d,*}(\CC^d) \tensor V^\vee [d-1] 
\] 
and the free action functional is simply
\[
L_{\rm free} (\beta,\gamma) = \<\beta, \dbar \gamma\> .
\]
Of course, for the holomorphic twist of $4d$ $\N = 1$ we will be most interested in the case $d = 2$. 
\end{eg}

Simply put, holomorphic Lagrangians are the ones that are natural from the point of view of the complex structure on the spacetime manifold we are putting the holomorphic theory on. 
That is, they are Lagrangians that are built from the fields of the BV theory which only involve holomorphic derivatives.
Our main examples of holomorphic Lagrangians will result from twists of supersymmetric interaction terms.
The precise definition is the following. 

\begin{dfn}
The space of {\em holomorphic local functionals} of a free holomorphic theory $(X, \omega^\text{hol}, Q^\text{hol})$ on $\CC^d$ is the $\ZZ$-graded sheaf
\deq{
\oloc^\text{hol}(Z^\bullet) = \Omega^{d,\text{hol}}_{\CC^d} \tensor_{D_{\CC^d}^\text{hol}} \prod_{n > 0} {\rm Hom}_{\sO^\text{hol}_{\CC^d}} (J^\text{hol}_{\CC^d} \tensor Z^\bullet)^{\tensor n}, \sO^\text{hol}_{\CC^d}) 
}
Here $\Omega^{d, \text{hol}}_{\CC^d}$ is the space of holomorphic top forms on $\CC^d$, $D_{\CC^d}^\text{hol}$ is the algebra of holomorphic differential operators, and $J^\text{hol}_{\CC^d}$ is the vector bundle of holomorphic $\infty$-jets of the trivial bundle on $\CC^d$. 
We note that $Q^\text{hol}$ determines a differential on the graded sheaf $\oloc^\text{hol}(Z^\bullet)$, giving it the structure of a sheaf of cochain complexes.
\end{dfn}

\begin{rmk} 
We have seen how every free holomorphic theory based on~$Z^\bullet$ gives rise to a free BV theory $\sB_Z$. 
We can also consider the complex of local functionals on $\sB_Z$ as in Definition \ref{dfn: localfun} equipped with the linear BV differential $(\oloc(\sB_Z), \dbar + Q^\text{hol})$. 
There is a quasi-isomorphism of sheaves on $\CC^d$
\[
\left(\oloc^\text{hol}(Z)[d], Q^\text{hol} \right) \simeq \left(\oloc(\sB_Z), \dbar + Q^\text{hol}\right)
\]
which is compatible with the BV brackets induced by $\omega^\text{hol}$ on both sides. 
This quasi-isomorphism tells us that holomorphic Lagrangians are precisely ordinary local functionals that are closed for the $\dbar$ operator. 
\end{rmk}

We can use the ordinary classical master equation for local functionals of a BV theory to make the following natural definition. 

\begin{dfn}
A {\em classical holomorphic theory} on a complex manifold $Y$ is the data of a free holomorphic theory $(Z^\bullet, Q^\text{hol}, (-,-)_Z)$ plus a holomorphic local functional
\[
I^\text{hol} \in \oloc^\text{hol}(Z^\bullet) 
\]
of cohomological degree $d$ such that the resulting local functional satisfies the classical master equation.
\end{dfn} 

\begin{rmk}\label{rmk: z22}
This is an extension of Remark~\ref{rmk: z2}.
In the case that the holomorphic theory is just $\ZZ/2$ graded, we require that $I^\text{hol}$ be even when $d$ is even, and $I^\text{hol}$ be odd when $d$ is odd. 
\end{rmk}

Before discussing how holomorphic Lagrangians arise from twists, we give an intrinsic holomorphic description of certain interactions one can add to the free $\beta\gamma$ system on $\CC^2$, or more generally on any Calabi--Yau manifold.

\begin{eg}
Let
\[
W \in \Sym^{\geq 2} (V^\vee) = \bigoplus_{n \geq 2} \Sym^{n} (V^\vee)
\]
be a polynomial on the complex vector space $V$ that is at least quadratic. 
This determines a holomorphic Lagrangian on the $\beta\gamma$ system via the formula
\[
L_W (\gamma, \beta) = W(\gamma) \,\d^2 z
\]
Note that in order for this Lagrangian to make sense we have used the obvious holomorphic volume form on $\CC^2$. 

Explicitly, this Lagrangian has the following interpretation. 
Choose a basis $\{e^i\}$ of $V$ and identify $W \in \CC[e_i]$.
Then, we can expand the $\gamma$ field 
\[
\gamma = \gamma^{0} + \gamma^{1} + \gamma^{2} \in \Omega^{0,*}(\CC^2) \tensor V
\]
as
\[
\gamma^0 = \gamma^0_i e^i \;\; , \;\; \gamma^{1} = \gamma^{1}_i e^i \;\; , \;\; \gamma^{2} = \gamma^{2}_i e^i 
\]
where where $\gamma^{a}_i$ is a form of type $(0,a)$.
The Lagrangian, which only remembers the top component, is of the form
\begin{equation}
  L_W(\gamma, \beta) = \frac{1}{2} \left(\partial_i W(\gamma^0) \gamma^{2}_i + \partial_i \partial_j W(\gamma^0) \gamma^{1}_i \wedge \gamma^{1}_j \right) \,\d^2 z .
\end{equation}
For example, if $V = \CC$ and $W = x^3 \in \Sym(V^\vee) = \CC[x]$, then the Lagrangian would read
\[
L_W(\gamma, \beta) = \frac{1}{2} \left((\gamma^0)^2 \gamma^{2} + \gamma^0 (\gamma^{1})^2 \right) \d^2 z .
\]
\end{eg}

\begin{rmk}
We remark is that while the free $\beta\gamma$ system is BRST $\ZZ$-graded, the deformed theory by a nonzero potential $W$ is only $\ZZ/2$-graded; see Remarks~\ref{rmk: z2} and~\ref{rmk: z22}. 
This is because $L_W$ is not homogenous with respect to the $\ZZ$-grading. 
It is, however, even with respect to the obvious forgetful map $\ZZ \to \ZZ/2$. 
In particular, the fields $\gamma^0$, $\gamma^2$, and~$\beta^1$ are of even parity and $\gamma^1$, $\beta^0$, and~$\beta^2$ are of odd parity. 
\end{rmk}

To match with the terminology in supersymmetry, we refer to the $\beta\gamma$ system on $\CC^2$ with values in $V$ equipped with the interaction $L_W (\gamma) = W(\gamma) \,\d^2 z$, as the $\beta\gamma$ system deformed by the superpotential $W$. 

\subsection{Holomorphic operators} 
\label{sec: hol op}

Local operators supported at a point in an arbitrary field theory can be identified with polynomials of derivatives of fields at that that point. 
In a topological twist, since all derivatives are made exact by the supercharge, the only operators left over are given by evaluating the fields at a particular point. 
In a holomorphic theory, not all derivatives are made exact, but all the anti-holomorphic ones are. 
Thus, the operators left over are given by (formal) polynomials in holomorphic derivatives of the fields. 
The precise definition is the following. 

\begin{dfn}\label{dfn: hollocal}
Let $(Z^\bullet, \omega^\text{hol}, Q^\text{hol}, I^\text{hol})$ be a holomorphic theory on $\CC^d$. 
The $\ZZ$-graded vector space of {\em classical holomorphic local operators} at $w \in \CC^d$ is defined by 
\[
\Obs_{w}^\text{hol} = \sO(Z^\bullet [\![z_1 - w_1,\ldots, z_d-w_d]\!]) = \Hat{\Sym} \left(Z^\bullet [\![z_1 - w_1,\ldots, z_d-w_d]\!]\right)^\vee .
\]
Here, as always, $(-)^\vee$ denotes the continuous linear dual. 
\end{dfn}

\begin{rmk}
  \label{rmk:sigma-model}
Let $\hD^d$ be the formal (complex) $d$-disk, whose ring of functions is $\sO(\hD^d) = \CC[\![z_1,\ldots, z_d]\!]$. 
The holomorphic operators are just functionals on the space $\sO(\hD^d) \tensor Z^\bullet$.
If one thinks of $\sO(\hD^d) \tensor Z^\bullet$ as a formal $\sigma$-model of maps $\hD^n \to Z^\bullet$, then we are simply considering the corresponding operators of the $\sigma$-model. 
In other words, the holomorphic operators are the operators on the formal completion of the space of maps $\CC^d \to Z^\bullet$ at the point $w \in \CC^d$.
\end{rmk}

\begin{rmk} 
We have already mentioned that through the program of Costello-Gwilliam \cite{CG1,CG2} one attaches a factorization algebra to any perturbative QFT. 
The local operators, as we've defined them, see only a small piece of this factorization algebra. 
Indeed, $\Obs_w^{\rm hol}$ is the value of the factorization algebra on a disk in $\CC^d$ centered at $w$ as the radius gets infinitesimally small. 
When $d=1$ this is the same relationship between holomorphic factorization algebras and vertex algebras, whereby the state space of the vertex algebra is given by the holomorphic local operators. 
For general $d$, the factorization algebra should endow the space of holomorphic local operators with a higher dimensional analog of the OPE.
\end{rmk}

Every holomorphic local operator determines a functional on the solutions to the equations of motion to the holomorphic theory.
Explicitly, the linear element 
\[
(z_1 - w_1)^{-k_1-1} \cdots (z_d - w_d)^{-k_d - 1} v^\vee \in \Obs_{w}^\text{hol}
\]
can be understood as the operator
\[
\varphi \in \sO^\text{hol} (\CC^d) \tensor V \mapsto \left\<v^\vee , \frac{\partial^{k_1}}{\partial z_1^{k_1}} \cdots \frac{\partial^{k_d}}{\partial z_d^{k_d}} \varphi \right\> (z = w)
\]
Here, the braces denote the contraction between $V$ and its dual $V^\vee$. 
Non-linear operators can be understood similarly. 

The piece of the BV differential $Q^\text{hol} + \{I^\text{hol},-\}$ acts on the space $\Obs_{w}^\text{hol}$.
Inherently, this operator is square zero, so we find that
\[
\left(\Obs_{w}^\text{hol} , Q^\text{hol} + \{I^\text{hol},-\}\right)
\]
is a cochain complex. 
We will refer to this as the {\em cochain complex of holomorphic local operators}. 

\begin{rmk}
\label{rmk: ss1}
We could have started with the full classical BV description of a holomorphic theory with fields
\[
\sE_V = \left(\Omega^{0,*}(\CC^d, Z^\bullet), \dbar + Q^\text{hol}\right)  .
\] 
The space of local operators $\Obs_{w}$ of $\sE_V$ at $w \in \CC^d$, as defined in Definition \ref{dfn: susy ops}, is much bigger than the space of holomorphic local operators. 
Indeed, such operators could involve {\em anti}-holomorphic derivatives ${\partial}/{\partial \zbar_i}$. 
However, $\Obs_w$ is a cochain complex equipped with the classical BV differential $\dbar + Q^\text{hol} + \{I,-\}$.
This results in a filtration on the observables by antiholomorphic form degree, and therefore a corresponding spectral sequence, on whose $E_0$ page we take the cohomology of the classical observables with respect to the $\dbar$ operator. 
Since we are dealing with {\em local operators}, there is no higher $\dbar$-cohomology and the $E_1$-page can be identified with the cochain complex of holomorphic local operators $\left(\Obs_{w}^\text{hol}, Q^\text{hol} + \{I^\text{hol},-\}\right)$.
\end{rmk}

Since the fields of the BV theory associated to a holomorphic theory are built from the Dolbeault complex, there is a natural action by the unitary group $U(d)$ on the space of fields. 
In fact, the group of biholomorphisms on $\CC^d$ acts on the Dolbeault complex of $\CC^d$ simply by pullback of differential forms.
Given any biholomorphism $\varphi : \CC^d \to \CC^d$ there is an automorphism of complexes
\[
\varphi^* : \Omega^{0,*}(\CC^d) \to \Omega^{0,*}(\CC^d) 
\]
sending $\alpha \mapsto \varphi^* \alpha$, which is compatible with $\dbar$ by holomorphicity. 

Moreover, if we additionally assume that $I^\text{hol}$ is $U(d)$-invariant, the resulting BV action has a symmetry by the group $U(d)$. 
In this case, this symmetry determines an action of $U(d)$ on the classical holomorphic local operators $\Obs^\text{hol}_w$. 

\subsection{$U(2)$-equivariant description and holomorphic twist}

The next two sections involve the proof of the following proposition, which completely characterizes the holomorphic twist of the chiral multiplet. 

\begin{prop} \label{prop:F-terms}
Consider the $\N = 1$ chiral multiplet on~$\RR^4$ with holomorphic superpotential $W$. 
Let $Q \in \Pi S_+ \subset \st$ be a chiral element in the Lie algebra of supertranslations. 
The twist of the $\N=1$ chiral multiplet with respect to $Q$ is equivalent to the deformation of the free $\beta\gamma$ system on $\CC^2$ by the interaction
\[
  L_W(\beta,\gamma) = W(\gamma) \, \d^2 z.
\]
Here we are extracting the top component of the mixed form $W(\gamma) \in \Omega^{0,*}(\C^2)$.
\end{prop}

The $\beta\gamma$ system was our typical example of a free holomorphic theory. 
The deformation we are seeing in the twist of the theory with superpotential is a deformation of this theory by a Lagrangian that is holomorphic by assumption. 
Therefore, the twist of the chiral multiplet in the presence of a superpotential arises as the BV-ification of a holomorphic theory on $\CC^2$. 

The next two sections provide two separate proofs of this proposition. 
The first uses the standard description of off-shell supersymmetry via the introduction of auxiliary fields and closely follows~\S\ref{sec: bv aux} above. 
The second uses the off-shell BV description which does not involve auxiliary fields at the expense of introducing some higher interaction terms in the Lagrangian. 
We then further show that the computation of the twist can be packaged neatly in terms of the  pure spinor superfield formalism.

\subsection{The twist in the presence of auxiliary fields}
\label{sec: twistaux}
In this section we compute the holomorphic twist, and prove Proposition~\ref{prop:F-terms}, using the standard auxiliary fields. 
It is a straightforward exercise to compute the decomposition of the spin representation under the subgroup $U(2) \subset SO(4)$, 
As we have recalled above, the complex Dirac spinor (with its decomposition into Weyl spinors) is constructed as 
\deq{
D = \Lambda^* L; \qquad S_+ = \Lambda^0 L \oplus \Lambda^2 L, \quad S_- = \Lambda^1 L \cong L.
}
The $U(2)$ subgroup is nothing other than the stabilizer of the subspace $L \subset \ST_\C$, which corresponds to a complex structure on~$\ST$; its complexification is just $\lie{gl}(L)$. As such, the spinors and the vector decompose into~$U(2)$ representations as
\deq[spindecomp]{
S_+ \rightarrow \rep{1}^{1} \oplus \rep{1}^{-1}, \quad S_- \rightarrow \rep{2}^0, \quad
\st^0 \rightarrow \rep{2}^1 \oplus \rep{2}^{-1}.
}
The dual of the vector transforms identically, but of course the pairing connects opposite $U(1)$ charges. 
Here boldface integers label $SU(2)$ representations by their dimensions, and the exponent represents the charge under the center of~$U(2)$, being $U(1)$. The isomorphism between the vector and the bispinor decomposes into the two obvious maps
\deq{
\rep{1}^1 \otimes \rep{2}^0 \cong  \rep{2}^1, \quad \rep{1}^{-1} \otimes \rep{2}^0 \cong  \rep{2}^{-1}.
}
In what follows, we will write the components of the $\psi$ field under the decompositions~\eqref{spindecomp} as $\psi_+$, $\psi_-$, and~$\bar\psi$, and similarly for $\epar$ and~$\bar\epar$. The supersymmetry transformations~\eqref{susy-F} then become
\begin{align}
\delta \phi &= \epar_+ \psi_- - \epar_- \psi_+, \nonumber\\
\delta \psi_+ &= \epar_+ F + i \bar\epar \left(  \partial \phi \right) , \nonumber\\
\delta \psi_- &= \epar_- F + i \bar\epar \left(  \bar\partial \phi \right), \label{susy-holo-1}\\
\delta F &= - i \bar\epar \left( \partial \psi_- - \bar\partial \psi_+ \right), \nonumber
\end{align}
and their complex conjugates reduce to
\begin{align}
\delta \bar\phi &= \bar\epar \bar\psi, \nonumber\\
\delta \bar\psi &= i \epar_+ \bar\partial \bar\phi  +  i \epar_- \partial \bar\phi 
	+ \bar\epar \bar{F}, \label{susy-holo-2}\\
\delta \bar{F} &= - i \epar_+  \bar\partial \wedge \bar\psi  - i \epar_-  \partial \wedge \bar\psi. \nonumber
\end{align}
(Here, we are again writing a differential of Chevalley--Eilenberg type, acting on linear operators rather than fields.) Using these transformations, it is trivial to read off the twisting differential $\delta = \{ Q_-, \cdot \}$ by setting $\epar_+ = 1$ and other ghosts to zero:\footnote{Of course, choosing $\epar_-$ instead would give an equivalent result with respect to a different complex structure.}
\begin{equation}
 \begin{aligned}[c]
 \delta\phi &= \psi_-,  \\
  \delta \psi_- &= 0,  \\
 \delta \psi_+ &= F,  \\
 \delta F &= 0,   
 \end{aligned}
 \qquad\qquad
 \begin{aligned}[c]
 \delta \bar\phi &= 0,  \\
 \delta \bar\psi &= i \bar\partial \bar \phi,  \\
 \delta \bar{F} &= -  i \bar\partial \wedge \bar\psi. 
 \end{aligned}
 \label{susy-twisting}
\end{equation}
We thus arrive at a cochain complex which is just the Dolbeault resolution of holomorphic functions, $\Omega^{0,*}(\CC^2)$. The twist of the chiral multiplet therefore becomes the field $\gamma$ of the $\beta\gamma$ system; the BV antifields will provide the corresponding $\beta$ field. This is the first half of the proof of Prop.~\ref{prop:F-terms}, in auxiliary-field language.

In fact, the full structure of~\eqref{susy-holo-2} showcases the structure of the twist as a family over the nilpotence variety. It is immediate to see that, for a generic point $(\epar_+,\epar_-)$, one obtains the Dolbeault complex, with a differential corresponding to a deformation of complex structure:
\deq{
  \bar\partial_\epar = d\bar{z}_i \, \left( \epar_+ \pdv{ }{\bar{z}_i} + \epar_- \epsilon_{ij} \pdv{ }{z_j} \right) .
}
The connected component of the nilpotence variety corresponding to the $S_+$, which is just a copy of~$\C P^1$, is thus explicitly identified with the space of complex structures on~$\R^4$, as must happen on general grounds~\cite{NV}. At the point $(\epar_+,\epar_-) = (1,0)$, the $\epar_-$ deformation is just the tangent space to the nilpotence variety. 

Further, the $\bar\epar$ deformations represent the normal bundle to the nilpotence variety. This must always contain a copy of the defining representation of $SU(n)$, witnessing the fact that  antiholomorphic translations are nullhomotopic in the twisted theory. In a topologically twisted theory, the supercharges providing a nullhomotopy for the translations are responsible for the phenomenon of \emph{topological descent;} here, a holomorphic analogue is present, which allows us to construct nonlocal holomorphic operators of ghost number zero out of local operators of nonzero ghost number. While the story of topological descent is classical, such higher structures were considered again recently in~\cite{BBBDN}. In our case, the relevant operad is a holomorphic analogue of the little discs operad that appears in topological field theories; in complex dimension $d$, there is a class in the cohomology of the binary part of the operad in degree $d-1$, corresponding to the Dolbeault cohomology of punctured $\C^d$; this is nothing other than the pairing on the fields of the $\beta\gamma$ system.  We plan to give a detailed discussion of this operad in future work.

Let's return now to the BV action~\eqref{BVaction-F}, decompose it in $U(2)$-equivariant language including antifields,  and then simply set $\epar_+$ to one and all other ghost antifields to zero. This will give the deformation of the action that corresponds to the deformation of the original BV differential by the twisting supercharge. The $U(2)$-equivariant decomposition is straightforward, and the result for the first two terms is
\begin{equation}
\begin{aligned}[c]
  L_\text{free} &= - \partial\bar\phi \wedge \bar\partial\phi  -  \bar\partial\bar\phi \wedge \partial\phi- i \bar\psi \left( \partial \psi_- + \bar\partial \psi_+ \right) + \bar{F} F, \\
L_\text{int} &= - \frac{1}{2} W^{ij} \left( \psi_{i+} \psi_{j-} - \psi_{i-} \psi_{j+} \right) - \frac{1}{2} \bar{W}_{ij} \bar\psi^i \wedge \bar\psi^j + W^i F_i + \bar{W}_i \bar{F}^i.
\end{aligned}
\label{eqs:holoL}
\end{equation}
For the antifield portion of the action, we have
\begin{multline}
L_\text{a.f.} = \phi^* ( \epar_+ \psi_- - \epar_- \psi_+ ) + F (\psi^*_+ \epar_- - \psi^*_- \epar_+ ) - i \bar\epar (\psi^*_+ \bar\partial \phi + \psi^*_- \partial \phi) \\
 - i F^* \bar\epar (\partial \psi_- + \bar\partial \psi_+ ) 
 + \bar\phi^* \bar\epar \bar\psi + F \bar\psi^* \bar\epar
 - i \bar\psi^* (\epar_+ \bar\partial \bar\phi + \epar_- \partial \bar\phi) 
 + i \bar{F}^* (\epar_+ [\bar\partial \bar\psi] + \epar_- [\partial \bar\psi] ) .
 \label{eqs:holoLaf}
\end{multline}
As such, the BV action for the holomorphic twist is easy to write down:
\deq[BVaction-holotwist]{
  L_\text{tw} = L_\text{free} + L_\text{int} + \phi^* \psi_- - F \psi^*_- - i \bar\psi^* \bar\partial \bar\phi + i \bar{F}^* [ \bar\partial \bar\psi],
}
reproducing the differential~\eqref{susy-twisting} above. After discarding the contractible part of this complex, the rest of the action takes a much simpler form: $L_\text{free}$ disappears entirely, reflecting the fact that it arises from a pairing between the chiral field and its $Q$-trivial complex conjugate. This is a reflection of the standard piece of dogma that $D$-terms are not relevant for index computations.

On the other hand, $L_\text{int}$ (which comes from the $F$-term of the original action) does play a significant role: it contributes the terms
\deq[twisted-F-terms]{
L_\text{int} \sim - \frac{1}{2} \bar{W}_{ij} \bar\psi^i \wedge \bar\psi^j + \bar{W}_i \bar{F}^i.
}
Identifying the fields in the twisted theory with a total Dolbeault form $\gamma$, this reduces to the simple expression
\deq{
  L_\text{int} = [W(\gamma)]_\text{top}.
}
Of course, the term in the action pairs this $(0,2)$ form with the Calabi--Yau form on spacetime. It is also ineresting to note that the kinetic terms of the twisted theory do \emph{not} arise from the kinetic terms of the full theory; rather, they come from the terms pairing antifields with fields, and therefore from the supersymmetry transformations of the chiral multiplet (which become the internal differential in the twisted theory).

\subsection{The twist without auxiliary fields} 
Armed with the results of the above calculations, it is straightforward  to compute the twist of the theory in the BV formalism without auxiliary fields. Note that this is explicitly an ``$L_\infty$''  twist, in the sense that we are deforming a Maurer--Cartan element that is not simply a BRST differential. As above, we will first  reduce to holomorphic  language, following \eqref{eqs:holoL} and~\eqref{eqs:holoLaf} above. Then we will integrate out the auxiliary field; in holomorphic language, its equations of motion are
\deq[holo-Fterm]{
\begin{aligned}[c]
\bar{F}^i &= -W^i - \epar_+ \psi^*_- + \epar_- \psi^*_+, \\
F_i &= - \bar{W}_i - \bar\epar \bar\psi^*_i.
\end{aligned}
}
Having done this, we can set the supersymmetry ghosts to the background values corresponding  to  the holomorphically twisted theory: $\epar_+ = 1$, and $\epar_-  = \bar\epar = 0$.  The analogue of~\eqref{eqs:holoL}, imposing~\eqref{holo-Fterm}, is then  
\deq[eqs:holoL-noF]{
\begin{aligned}[c]
  L_\text{free} &\rightarrow - \partial\bar\phi \wedge \bar\partial\phi  -  \bar\partial\bar\phi \wedge \partial\phi- i \bar\psi \left( \partial \psi_- + \bar\partial \psi_+ \right) + (W^i + \psi_-^{*i})\bar W_i, \\
L_\text{int} &\rightarrow - \frac{1}{2} W^{ij} \left( \psi_{i+} \psi_{j-} - \psi_{i-} \psi_{j+} \right) - \frac{1}{2} \bar{W}_{ij} \bar\psi^i \wedge \bar\psi^j - W^i \bar{W}_i  - \bar{W}_i (W^i + \psi^{*i}_-).
\end{aligned}
}
Notice that the terms containing the antifield $\psi^{*i}_-$ cancel in the sum of these two terms! The BV portion of the action correspondingly becomes
\deq[eqs:holoLaf-noF]{
L_\text{a.f.} \rightarrow \phi^{*i} \psi_{-i} + \bar{W}_i \psi^{*i}_- - i \bar\psi^*_i \bar\partial \bar\phi^i. 
}
Here, we have simply set the antifields of~$F$ to zero, as we did above in the full theory.
Note that the term in the BV action quadratic in antifields plays no role in the twisted theory in this dimension, since it contains both $\epar$ and~$\bar\epar$, and is therefore set to zero on any background corresponding to  a twist. In this sense, the $L_\infty$ structure is lost in the holomorphic theory. It will, of course, play an important role in theories where the entire action of supertranslations is gauged---i.e., in the coupling of chiral matter to supergravity.

Furthermore, the portion of the complex that was contractible before is \emph{not} completely contractible now!  The pair $\phi, \psi_-$ still form an acyclic boson--fermion pair, and can be removed; the resulting total BV action is
\deq[BVaction-tot]{
\mathfrak{L} = - i \bar\psi^i \wedge \bar\partial \psi_{+i} -  \frac{1}{2} \bar{W}_{ij} \bar\psi^i \wedge \bar\psi^j  + \bar{W}_i \psi^{*i}_- - i \bar\psi^*_i \bar\partial \bar\phi^i.
}
The quadratic term $|W_i|^2$ disappears because $\phi$ still belongs to an acyclic pair. As above, the fields that survive to the twist can still be identified with two copies of the Dolbeault complex, but the components of the antifield that cannot be integrated out now play the roles of the auxiliary fields. (For the assignments of $U(1)$ charges, see below in Table~\ref{tab:gradings}.) We identify
\deq{
  \gamma = (\bar\phi, \bar\psi, \psi_-^*), \qquad
  \beta = (\psi_+, \bar\psi^*, \bar\phi^*).
}
As above, the antifield of~$\gamma^i$ is $\beta^{2-i}$. 
Making these identifications and integrating once by parts, we can rewrite~\eqref{BVaction-tot} as
\deq[BVaction-ident]{
  \mathfrak{L} = i \langle \beta, \bar\partial \gamma \rangle + [W(\gamma)]_\text{top}, 
}
in perfect agreement with the result above.

\subsection{Twisting in the pure spinor superfield formalism}
\label{ssec:PStwist}

In this section, we arrive at the twist of the chiral multiplet in yet a third way, Recall that superfields are elements of  the coordinate ring of super-Minkowski space,  which (by definition) is a free graded commutative algebra:
\deq{
\O(\st) = \O(\ST) \otimes \C[ \theta^\alpha, \bar\theta^{\dot\alpha}].
}
Here the $\theta\in S_+$ are odd coordinates; we will denote even coordinates by~$x^\mu \in V$, $\mu = 1,\ldots,4$. 

The space of superfields carries the obvious representation of~$\so(4)$. However, we will find it convenient to  reduce immediately to $\lie{u}(2)$-equivariant language, as we are interested in considering the holomorphic twist. The space of superfields then looks like
\deq{
\C[z^i,\bar{z}^i; \theta^+,\theta^-, \bar\theta] , \quad  i =1,2.
}
The $\lie{u}(2)$ representations are the obvious ones: 
\deq{
  \rep{2}^{-1}, \rep{2}^{1}; \rep{1}^1, \rep{1}^{-1}, \rep{2}^0.
}
Note that we take the momentum operator $\partial/\partial z$, rather than $z$ itself, to transform in the fundamental of~$U(2)$. 
The algebra of supertranslations acts on the space of superfields by the left and right regular representations, which mutually commute:
\deq[QD-SO(4)]{
\begin{aligned}[c]
Q_\alpha &= i \pdv{}{\theta^\alpha} - \sigma^\mu_{\alpha\dot\beta} \bar\theta^{\dot\beta} \pdv{}{x^\mu},
\\
\bar{Q}_{\dot\beta} &=- i \pdv{}{\bar\theta^{\dot\beta}} + \sigma^\mu_{\alpha\dot\beta} \theta^{\alpha} \pdv{}{x^\mu},
\end{aligned}
\qquad
\begin{aligned}[c]
D_\alpha &= \pdv{}{\theta^\alpha} - i \sigma^\mu_{\alpha\dot\beta} \bar\theta^{\dot\beta} \pdv{}{x^\mu},
\\
\bar{D}_{\dot\beta} &= - \pdv{}{\bar\theta^{\dot\beta}} + i \sigma^\mu_{\alpha\dot\beta} \theta^{\alpha} \pdv{}{x^\mu}.
\end{aligned}
}
It is easy to  reduce these expressions to holomorphic language, and the results look like
\deq[QD-U(2)]{
\begin{aligned}[c]
&Q_+ = i \pdv{}{\theta^-} -  \bar\theta \wedge \pdv{}{{z}},
\\
&Q_- = i \pdv{}{\theta^+} -  \bar\theta \wedge \pdv{}{\bar{z}},
\\
&\bar{Q} =- i \pdv{}{\bar\theta} + \theta^-  \pdv{}{z} +  \theta^+  \pdv{}{\bar{z}},
\end{aligned}
\qquad
\begin{aligned}[c]
&D_+ = \pdv{}{\theta^-} - i \bar\theta \wedge \pdv{}{{z}},
\\
&D_- = \pdv{}{\theta^+} - i \bar\theta \wedge \pdv{}{\bar{z}},
\\
&\bar{D} = - \pdv{}{\bar\theta} + i \theta^- \pdv{}{z} + i \theta^+  \pdv{}{\bar{z}}.
\end{aligned}
}

As we reviewed above in~\S\ref{sssec:PSSF}, the chiral superfield in the pure spinor superfield formalism can be obtained from the $\O(Y)$-module  $\Gamma = \C[u_+,u_-]$, and is then isomorphic to
\deq{
  H^*(\O(\st) \otimes \Gamma, \Berk) \cong \O(\ST \oplus \Pi S_-).
}
The action of the supersymmetry generators can be obtained  by throwing out nullhomotopic  terms from~\eqref{QD-U(2)} above:
\deq{
  Q_+ = -  \bar\theta \wedge \pdv{}{z}, \quad
  Q_- = - \bar\theta \wedge \pdv{}{\bar z}, \quad
  \bar Q = -i \pdv{}{\bar\theta}.
}
This makes it apparent  that $Q_-$ is nothing other than the Dolbeault differential, and in fact that the differential obtained from generic values $(\epar_+,\epar_-)$ is just the Dolbeault differential for a deformation of the complex structure on~$\ST$. We will return to this point later on. The  $\bar\epar$ differential, on the other hand, is the same cancelling differential we saw above.

\subsection{Gauge interactions and general $\N=1$ theories}
\label{ssec:gauge}
In this section, we compute the twist of the four-dimensional $\N=1$ vector multiplet. Having done this, we will give a general description (Proposition~\ref{prop:gen4d}) for the holomorphic twist of any four-dimensional supersymmetric theory; this applies just as well to the minimal (holomorphic) twist of theories with more supersymmetry, since they are special examples of $\N=1$ theories. While we will not explicitly compute characters for gauge theories in this work, we will use our description of the gauge multiplet in the discussion of the twisted flavor current multiplet below.

The vector multiplet consists of fields $(A, \lambda, \Bar{\lambda}, D)$ of ghost degree zero, where 
\begin{itemize}
\item $A \in \Omega^1(\RR^4) \tensor \fg$ is a connection one-form with values in the Lie algebra~$\fg$;
\item $(\lambda,\bar\lambda) \in C^\infty(\ST) \tensor \fg \tensor \Pi \left( S_+ \oplus S_- \right)$ are a pair of $\fg$-valued spinors of opposite chirality; 
\item $D \in C^\infty(\ST^4) \tensor \fg$ is an auxiliary field.
\end{itemize}
In addition, there is a ghost field $c \in C^\infty(\ST) \tensor \fg$ of ghost degree $-1$.

This data compiles together to define a super dg Lie algebra where the differential is of the form
\begin{equation}
  \begin{tikzcd}[row sep = tiny]
  \ul{0} & & \ul{1} \\
  C^\infty(\RR^4) \tensor \fg \ar[rr, "\d"] & & \Omega^{1} (\RR^4) \tensor \fg  \\
 & & C^\infty(\RR^4) \tensor \fg \tensor \Pi \left( S_+ + S_- \right) \\ 
 & &  C^\infty(\RR^4) \tensor \fg, 
\end{tikzcd}
\label{eq:regraded}
\end{equation}
and the Lie bracket is extended from the matrix commutator. 
Note that it is the shift by one of the fields, rather than simply the fields, that carry a dg Lie structure; as such, the grading in~\eqref{eq:regraded} above is shifted by one from the natural grading on fields. The Chevalley--Eilenberg complex of this dg Lie algebra, which incorporates the shift by one, returns the operators of the theory, together with their BRST differential. 
In this diagram, the outer two lines are of even parity and the last two lines are of even parity (although note that the homological degree means that $c$ has odd parity overall).  
The horizontal degree is the BRST degree. 

Let us now pass to the dual description, in which gauge transformations on fields are represented by a Chevalley--Eilenberg differential on operators. The internal differential of degree $+1$ then acts by
\deq[eq:internal-diff]{
  \begin{aligned}[c]
    \kappa c &= [c, c], \\
    \kappa A &= d c + [A, c],\\
    \kappa \lambda &= [\lambda, c], \\
    \kappa D &= [D,c].
  \end{aligned}
}
(Note that we are abusing notation slightly, by using the same letters in~\eqref{eq:internal-diff} for operators as we used for the corresponding fields.) The cohomology of this differential accomplishes the task of passing to the gauge-invariant sector of the theory.

The action of $\N=1$ supertranslations is the familiar one for the $\N=1$ vector multiplet. This action can be represented by a differential, as we did for the chiral multiplet above, as follows:
\begin{equation}
  \label{eq:V-trans}
  \begin{aligned}[c]
    \delta c &=  0 , \\ 
    \delta A &= \bar\epar \lambda + \bar\lambda \epar, \\
    \delta \lambda &= i F^+ \epar + \epar D, \\
    \delta D &= i \left( -\bar\epar \codslash \lambda + \epar \codslash \bar\lambda \right)
  \end{aligned}
\end{equation}
As  above, we need to either restore the obvious $a^\mu\partial_\mu$ terms, or set one of $\epar$ or~$\bar\epar$ to zero (i.e., restrict to the nilpotence variety).  The notation deserves a word of explanation: When we here write the product~$S_+ \otimes S_-$, as in~$\bar\epar\lambda$, we are implicitly identifying it with the vector representation. Furthermore, in the variation of~$\lambda$, we have identified
\deq{
  \Hom(S_+,S_+) \cong S_+ \tensor S_+^\vee \cong S_+ \tensor S_+,
}
using the antisymmetric pairing on the Weyl spinor; this representation contains the self-dual two-form, denoted by~$F_+$. Finally, in the last line, the vector, spinor, and conjugate spinor are contracted to form a scalar. 

To twist this multiplet, we proceed as in~\S\ref{sec: twistaux}: first, we choose a complex structure on~$M$, and perform the $U(2)$-equivariant decomposition of the fields; then,  we set $\epar$ to a value corresponding to an appropriate point on the nilpotence variety.
The $U(2)$ decomposition of the transformations~\eqref{eq:V-trans} is
\begin{equation}
  \label{eq:V-trans-holo}
  \begin{aligned}[c]
    \delta c &= 0, \\
    \delta A &= \bar\epar \lambda_+ + \epar_+ \bar\lambda, \\
    \delta \bar A &= \bar\epar \lambda_- + \epar_- \bar\lambda, \\
    \delta \lambda_+ &= \epar_+ (D + i F_0) + i\epar_- F_2, \\
    \delta \lambda_- &= \epar_- (D + i F_0) + i\epar_+ F_{-2}, \\
    \delta \bar\lambda &= \bar\epar D + F^-  \bar\epar, \\
    \delta D &=  i \left( - \bar\epar  \wedge \left( \codslash \lambda_- + \codslashbar  \lambda_+ \right) 
    + \epar_+ \codslashbar  \wedge \bar\lambda +  \epar_- \codslash \wedge \bar\lambda \right).
  \end{aligned}
\end{equation}
Here $F$ with a subscript labels the three scalar components of the self-dual two-form, under the $U(2)$-equivariant decomposition, with their $U(1)$ charges. 

As above, to  read off the twisting differential, we can just set $\epar_+$ to one  and  all other parameters to zero. When we have done this, it is immediate that $(A,\bar\lambda)$ form an acyclic pair and can be discarded. Furthermore, $D$ is closed under the differential, but it is set equal to a component $\partial\wedge \bar{A}$ of the field strength by the image of~$\lambda_+$, so that these two generators can  also be discarded. To obtain the full differential of  the twisted theory, we will now have to add back in the original internal differential~\eqref{eq:internal-diff}. The result is
\begin{equation}
  \label{eq:V-trans-reduced}
  \begin{aligned}[c]
    \delta c &= [c,c], \\
    \delta \bar A &= \bar\partial c + [\bar{A},c], \\
    \delta \lambda_- &= \left( \bar\partial + \bar{A} \right) \wedge \bar{A} + [\lambda_-,c].
  \end{aligned}
\end{equation}
To sum up, the fields of the twisted vector multiplet assemble into a single copy of the Dolbeault complex, $\sA \in \Omega^{0,*}(M,\lie{g})[1]$, but shifted so that the $(0,1)$ form---which arises from the physical gauge field---appears in ghost number zero. 
The differential, however, is slightly different: it reads 
\deq{
  \delta \sA = d\sA + [\sA,\sA],
}
where $d$ is the de~Rham differential. This arises from a BV action of the type
\deq{
  \mathfrak{L} = b \left( d\sA + [\sA,\sA] \right)  = b \wedge F_\sA,
}
where $b \in \Omega^{2,*}(M,\lie{g}^\vee)$ is the antifield multiplet. 
The physical values of unbroken $U(1)$ symmetries are given in Table~\ref{tab:vec}; a similar table for the chiral multiplet is displayed below (Table~\ref{tab:gradings}).
Note that in our conventions, the $R$-charges of the fields $(c, A, \lambda, \Bar{\lambda}, D)$ read $(0, 0, -1, 1, 0)$ for the untwisted multiplet.

\begin{table}
  \begin{tabular}{c|ccc|cc|c}
& $R$ &  $L$ &  gh & tw$_1$ & tw$_2$ & \\ \hline
$c$ & 0 & 0 & $1$ & 1 & 0  & $\sA^0$ \\
$A$ & 0 & 1 & 0 & 0 & 1 & $-$ \\
$\bar A$ & 0 & $-1$ & 0 & 0 & $-1$ & $\sA^1$\\
$\lambda_+$  & $-1$ & $1$ & $0$ & $-1$ & 0 & $-$ \\
$\lambda_-$ & $-1$ & $-1$ & $0$ & $-1$ &  $-2$ & $\sA^2$\\ 
$\bar\lambda$ & $1$ & 0 & 0 & 1 & 1 & $-$ \\
$D$ & 0 & 0 & 0 & 0 & 0& $-$ 
  \end{tabular}
  \qquad
  \begin{tabular}{c|ccc|cc|c}
& $R$ &  $L$ &  gh & tw$_1$ & tw$_2$ & \\ \hline
$c^*$ & 0 & 0 & $-2$ & $-2$ & 0  & $b^2$ \\
$A^*$ & 0 & $-1$ & $-1$ & $-1$ & $-1$ & $-$ \\
$\bar A^*$ & 0 & $1$ & $-1$ & $-1$ & $1$ & $b^1$\\
$\lambda_-^*$  & $1$ & $-1$ & $-1$ & $0$ & 0 & $-$ \\
$\lambda_+^*$ & $1$ & $1$ & $-1$ & $0$ &  $2$ & $b^0$\\ 
$\bar\lambda^*$ & $-1$ & 0 & $-1$ & $-2$ & $-1$ & $-$ \\
$D^*$ & 0 & 0 & $-1$ & $-1$ & 0& $-$ 
  \end{tabular}
  \smallskip
  \caption{Gradings on the twist of the vector multiplet. The differential has twisted bidegree $(1,0)$. (Compare Table~\ref{tab:gradings} for the chiral multiplet.)}
  \label{tab:vec}
\end{table}

We can summarize the results of this section with the following proposition, giving a tidy description of the holomorphic twist of any four-dimensional supersymmetric theory in complete generality:
\begin{prop}[Compare~{\cite{JohansenHolo}}]
  \label{prop:gen4d}
  The holomorphic twist of a general $\N=1$ theory in four dimensions on a K\"ahler manifold~$M$, with gauge Lie algebra $\lie{g}$ and chiral matter transforming in a representation~$V$ of~$\lie{g}$, produces the holomorphic BV theory whose fields are
  \begin{equation}
    \label{eq:genBV}
    \begin{aligned}[c]
      \sA &\in \Omega^{0,*}(M,\lie{g})[1], \\
      b &\in \Omega^{2,*}(M,\lie{g}^\vee),
    \end{aligned}
    \qquad
    \begin{aligned}[c]
      \gamma &\in \Omega^{0,*}(M,V), \\
      \beta &\in \Omega^{2,*}(M,V^\vee)[1].
    \end{aligned}
\end{equation}
The dynamics of the theory are specified by the BV action
\begin{equation}
  \label{eq:YMBVaction}
  \mathfrak{L} = \langle b, \bar\partial \sA + [\sA,\sA] \rangle + \langle \beta, (\bar\partial + \sA ) \gamma \rangle,
\end{equation}
encoding the minimal coupling of the $\beta\gamma$ system to holomorphic gauge theory. On a Calabi--Yau manifold, we can also add gauge-invariant superpotential interactions, in the form discussed previously:
\begin{equation}
  \mathfrak{L}_W = W(\gamma) \wedge \Omega,
\end{equation}
where~$\Omega$ is the Calabi--Yau form.
\end{prop}

\begin{rmk}
In \cite{Yangian} the twist of {\em pure} gauge theory is computed using a presentation of the theory in the first-order formalism of super Yang-Mills.
\end{rmk}

We note that, precisely in complex dimension two, the holomorphic twist of the vector multiplet is closely related to holomorphically twisted matter. In general, the $\beta\gamma$ system on an (ungraded) vector space has operators supported in degrees zero and $(d-1)$, while one expects the vector multiplet to have operators supported in degree one, with antifields in degree $(d-2)$. Precisely in complex dimension two, these pairs coincide. Note, however, that it is probably more accurate to think of the twisted vector multiplet as related to the $\beta \gamma$ system on $\lie{g}[1]$, rather than as an ungraded $\beta\gamma$ system with the roles of fields and antifields reversed.

It is a pleasant exercise, performed in~\cite{CostelloLi,CY4}, to check that the dimensional reduction of the holomorphic twist of ten-dimensional supersymmetric Yang--Mills theory (which is the holomorphic Chern--Simons theory with fields $\Omega^{0,*}(\C^5,\lie{g})$) is the twist of maximally supersymmetric Yang--Mills in four dimensions. From that perspective, the $\Z$-grading is broken to~$\Z/2$ due to the presence of a cubic superpotential, which originates from a component of the ten-dimensional cubic Chern--Simons interaction.

\subsection{BV quantization of the holomorphic twist}
\label{ssec:BVquant}
In this section, we turn to the quantization of the $\beta\gamma$ system on $\CC^2$. 
One of the advantages of formulating the holomorphic twist of the supersymmetric theory in the BV formalism is that there is a natural BV quantization. 
In fact, for every {\em free} BV theory there is a unique quantization obtained by deforming the classical BV differential by the BV Laplacian $\hbar \Delta$ given by the contraction with the $(-1)$-shifted symplectic form defining the classical theory.
Even in the case of a superpotential, we will see that no quantum correction arise.
Indeed, the quantization still exists uniquely. 
We will see how this works at the level of holomorphic local operators, as we introduced in~\S\ref{sec: holtheory}. 

\subsubsection{A recollection of the quantum BV formalism}

Classically, we have recalled that in the BV formalism a theory is given by the data of a complex of BV fields $\sB = T^*[-1] \sE$ together with a BV action $\mathfrak{S} = \int \mathfrak{L}$ satisfying the classical master equation. 
We study the BV quantization of the system through the quantization of its observables. 
As we have already mentioned, the full theory of quantization of the observables of a BV theory has been developed in~\cite{CG1,CG2}, using the language of factorization algebras.
While we do not use the full theory here, we remark that the quantization takes place locally on the spacetime manifold; that is, for each open set. 
The main result of \cite{CG2} is that these quantizations glue together according to the axioms of a factorization algebra. 

Schematically, the BV formalism suggests that a quantization of a classical theory is constructed in two steps: 
\begin{enumerate}
\item[(a)] tensoring the underlying graded vector space of observables $\Obs$ with $\CC[\![\hbar]\!]$ and
\item[(b)] modifying the differential to $\{\mathfrak{S}^q,-\} +\hbar \Delta$ where $\Delta$ is the BV Laplacian and $\mathfrak{S}^q = \mathfrak{S} + O(\hbar)$ is a {\em quantum} action satisfying the quantum master equation
\[
\{\mathfrak{S}^q, \mathfrak{S}^q\} + \hbar \Delta \mathfrak{S}^q = 0 .
\] 
\end{enumerate}

Naively, this prescription is incomplete for several reasons. 
First, $\Delta$ is not defined on all of the observables;
the naive formula involves an ill-defined pairing of distributions.
There is a natural way to circumvent this difficulty by introducing a mollification of $\Delta$ instead.
This approach is developed in a very broad context in Chapter 9 of \cite{CG2},
where one introduces a scale dependent BV Laplacian $\Delta_L$.\footnote{More generally, one can associate a BV Laplacian to every {\em parametrix}.}

Secondly, the same infinities that plagued the naive BV operator persist in defining the quantum action $\mathfrak{S}$. 
For a generic theory, when introducing an interaction term one must properly regularize the action functional by introducing counterterms. 
Even if one can introduce counterterms to get a well-defined loop expansion of the quantum action, it may still fail to satisfy the quantum master equation.

\subsubsection{Exactness in the holomorphic theory}

For the holomorphic theory we consider, there is a simple combinatorial reason why no such counterterms arise. 
No loop diagrams can ever contribute to the quantum action; therefore, the classical BV action will also automatically satisfy the quantum master equation!

\begin{prop}
The $\beta\gamma$ system on any complex surface $X$ in the presence of a holomorphic potential $W \in \Sym(V^\vee)$ is exact at tree level. 
In particular, a quantization of the theory exists locally on $\CC^2$ and on Hopf surfaces $X_{q_1,q_2} = \left(\CC^2 \setminus 0\right) / (z_1, z_2) \sim (q_1^\ZZ z_1, q_2^\ZZ z_2)$.  
\end{prop}
\begin{proof}
The quantum theory is constructed out of weights of diagrams are constructed out of the vertices, labeled by $I_W$, and edges, labeled by the propagator. 
The weights are filtered with respect to the parameter $\hbar$, which counts the genus of the graph. 
Since the propagator arises from the BV pairing on the space of fields, it is symbolically of the form
\[
\begin{tikzpicture}
		\draw[fermion] (0,0) -- (2,0);
		\draw (-0.2, 0) node {$\beta$};
		\draw (2.2,0) node {$\gamma$};
		\draw (1,.5) node {$P$};
\end{tikzpicture}
\label{fig:props}
\]
Here, we are ignoring any regularization, since it will play no role for us. 
Since the vertices of the diagrams, labeled by $I_W$, are only functions the fields $\gamma$, we see that only trees will appear in the expansion. 
\end{proof}

This proposition implies that the quantization of the holomorphic theory is easy to understand.
We are most interested in its implication at the level of {\em local} observables.

Just as in the non-BV case, there is a notion of local observables in the classical BV formalism.
We define the cochain complex of {\em local} BV observables at $x \in M$ as
\[ 
\Obs_x :=  \left( \Hat{\Sym}(J^\infty B |_x)^\vee , \{\mathfrak{S},-\} \right) 
\] 
where $B$ is the vector bundle on $M$ underlying the sheaf $\sB$. 

To quantize, we adjoin the parameter $\hbar$, and begin the construction of the quantum action by adding the BV Laplacian $\hbar\Delta$.
For the holomorphic theory, the quantum action receives no quantum corrections and we can disregard terms of order $\hbar^n$ for $n \geq 1$. 
Also, it is easy to see that the BV Laplacian vanishes identically on the {\em holomorphic} local operators.
Thus, the quantum holomorphic operators for the $\beta\gamma$ system in the presence of a superpotential is of the form
\[
  \left(\Obs_0^\text{hol}[\![\hbar]\!] , \{L_W,-\}\right),
\]
where $\Obs_0^\text{hol}$ is the space of classical holomorphic operators defined in~\S\ref{sec: hol op}.

\begin{rmk}
When one couples the $\N=1$ chiral multiplet to a gauge field, where the matter takes values in the adjoint representation, it is known that there are nonzero quantum corrections to the superpotential in perturbation theory \cite{CDSW}. 
We hope to return to seeing these perturbative corrections at the level of the holomorphic twist in future work.
\end{rmk}

\section{Holomorphic characters}
\label{sec:characters}

Consider a quantum field theory defined on affine space $\RR^n$.  
For simplicity, we assume this theory is translation invariant and we denote by $\Obs_0$ the local operators of the theory supported at the origin $0 \in \RR^n$. 
In the case that the theory is conformal, there is the state-operator correspondence, which relates local operators to states in the Hilbert space on $S^{n-1}$.
In the holomorphic case on $\CC^d$, we will find evidence (at least for $d=2$) for such a correspondence by relating a certain $q$-character of the local operators to the partition function over a class of complex manifolds diffeomorphic to $S^{2d-1} \times S^1$. 
This is a generalization to the situation in chiral CFT whereby the $q$-character of a vertex algebra is related to the partition function along an elliptic curve. 

There is a certain class of {\em non-local} operators that will play an essential role for us.
To see them, note that we can restrict the theory to the submanifold 
\[
\RR^n \setminus 0 \subset \RR^n .
\]
The radius of a point in punctured affine space gives a natural projection $r : \RR^n \setminus 0 \to \RR_{>0}$.
We can then reduce the theory along the map $r$ to get a theory of quantum mechanics defined on the positive line $\RR_{>0}$. 
In other words, since $\RR^n \setminus 0 \cong S^{n-1} \times \RR_{>0}$, we can understand this as compactifying the theory along the $(n-1)$-sphere.

We assume that compactification along $S^{n-1}$ results in a topological theory. 
For topological or holomorphic theories, this is certainly true, since the translations that survive the holomorphic twist cannot intersect nontrivially with the translations that survive compactification.\footnote{In fact, there is the small caveat that we must actually consider a dense algebraic subspace of operators of the resulting quantum mechanics that is actually topological.}
Denote by $\sA$ the local operators of the compactified theory on $\RR_{>0}$. 
A topological quantum mechanics is nothing other than a single associative algebra; in our setting, the one-dimensional OPE endows $\sA$ with the natural structure of a (homotopy) associative algebra. 

From the perspective of the full theory on $\RR^n$, the algebra $\sA$ actually contains non-local operators:
it consists precisely of the operators supported on spheres $S^{n-1}$, which can originate either from local operators in the full theory or from nonlocal operators wrapped on a nontrivial cycle. 
The operator product of these sphere operators induced by radial ordering endows $\sA$ with the aforementioned associative product. 
Moreover, the operator product of $S^{n-1}$-operators with local operators implies that $\Obs_0$ is a module for the algebra $\sA$. 

Given any algebra $A$ and a module $M$, finite dimensional over $\CC$, one defines the character by $a \mapsto \Tr_V(\exp(a))$ thought of as a map $HH_0(A) \to \CC$, where $HH_0(A)$ is the zeroth Hochschild homology. 
In our situation, for observables we define the {\em local character} is the the character of the $\sA$-module $\Obs_0$
\[
\ch_{\sA} (\Obs_0) : HH_0(\sA) \to \CC[\![\hbar]\!] .
\]

Of course, the space of local operators $\Obs_0$ is very rarely finite dimensional, so the above definition needs to be properly interpreted. 
In practice, there are additional gradings, or symmetries, present in a QFT which allow one to define the graded dimension of $\Obs_0$. 
For instance, for a chiral conformal field theory, the conformal structure allows one to define the $q$-character of local operators. 
For general holomorphic theories, there is a natural generalization, see Definition~\ref{dfn: holchar}. 

There is another interpretation of this character from the point of view of quantum mechanics. 
Since the original theory is defined on all of~$\RR^n$, the compactified theory on~$\RR_{>0}$ admits a natural boundary condition extending it to a theory on~$\RR_{\geq 0}$.
One can describe this boundary condition by saying that the boundary operators supported at $0 \in \RR_{\geq 0}$---i.e., functions on field configurations compatible with the boundary conditions---are isomorphic to the local operators $\Obs_0$ of the original theory.

In fact, the algebra of local operators in the quantum mechanics is essentially the (differential graded) Weyl algebra, formed from the symplectic vector space which is the cotangent bundle to (the spectrum of) holomorphic local operators in the upstairs theory. One can see this by considering the Dolbeault cohomology of punctured~$\C^d$, which has classes in degree zero and $d-1$ that are paired by integration. (We discuss this further below in~\S\ref{sec:symmetry}.) The fields of the $\beta\gamma$ system on this geometry, after passing to the cohomology of~$\bar\partial$, are $Z^\bullet \otimes H^*_{\bar\partial}(\C^d \setminus 0)$, which has a symplectic pairing in degree zero, and can be thought of as the cotangent bundle to~$Z^\bullet \otimes H^*_{\bar\partial}(\C^d)$. The operator product discussed above on local operators is precisely the quantization of classical local operators with respect to the degree-zero Poisson bracket structure. 

As is familiar from elementary quantum mechanics, such algebras usually admit unique irreducible unitary representations, which are constructed by taking functions on a Lagrangian subspace of the relevant symplectic vector space. It is immediate to see that the operators that are local upstairs define a canonical choice of such a Lagrangian, akin to the zero section of a cotangent bundle. 
Hilbert spaces are associated to boundaries, and this choice of Lagrangian is to be interpreted as a choice of boundary condition in the manner discussed above.  The resulting Hilbert space, over which the trace is taken, then depends on a choice of boundary condition, i.e.~Lagrangian, analogous to polarization data in geometric quantization.
Finally, the partition function of the original theory on $S^{n-1} \times S^1$ can be thought of as a trace over the Hilbert space of this quantum-mechanical system, which is the corresponding  module of~$\sA$.

\subsection{The local character}

We now turn to the holomorphic situation.
For a holomorphic theory (like the ones coming from twists of $\N=1$ in dimension four) on $\CC^d$, we have defined the holomorphic local operators $\Obs^\text{hol}_w$ at $w \in \CC^d$. 
These are simply on-shell local operators of the underlying free BV theory, but as a cochain complex are equipped with the differential $Q^\text{hol} + \{I^\text{hol}, -\}$ where $Q^\text{hol}$ is a linear holomorphic differential operator and $I^\text{hol}$ is the holomorphic interaction. 

\begin{dfn}
\label{dfn: holchar}
Let $\Obs^\text{hol}_0$ be the holomorphic local operators of a holomorphic theory on $\CC^d$ as defined in Definition \ref{dfn: hollocal}. 
The {\em bare ${\bf q}$-character} is defined by the formal series
\[
\chi ({\bf q}) = \sum_{j_1,\ldots,j_d \in \ZZ} q_{1}^{j_1} \cdots q_{d}^{j_d} \dim (\Obs^{(j_1,\ldots,j_d)}_0) \in \CC[\![q_1^{\pm}, \ldots, q_d^{\pm}]\!] .
\]
Here, $\Obs^{(j_1,\ldots,j_d)}$ labels the $(j_1,\ldots,j_d)$-eigenspace corresponding to the action of the maximal torus $T^d \subset U(d)$. 
\end{dfn}

There is the following algebraic way to think about this ${\bf q}$-character. 
Because the space of local operators is a $U(d)$-representation, there is a map of algebras 
\[
U(\fu(d)) \to \End(\Obs_0) .
\]
The character only depends on the Cartan Lie subalgebra $\CC^d \subset \fu(d)$ which we think of as being generated by the scaling operators
\[
L_0^{i} = z^i \frac{\partial}{\partial z_i} \;\;\; , \;\;\; 1 \leq i \leq d
\]
where no summation convention is used. 
Restricting to the Cartan, we obtain a map of algebras $\rho : U(\CC^d) \to \End(\Obs_0)$. 

By definition, this map factors through endomorphisms of the $T^d$-eigenspaces 
\[
U(\CC^d) \to \bigoplus_{(j_1,\ldots,j_d)} \End(\Obs^{(j_1,\ldots,j_d)}_0) .
\]
The character is obtained from the induced map at the level of Hochschild homology:
\[
HH_*(\rho) : HH_*(U(\CC^d)) \to \bigoplus_{(j_1,\ldots,j_d)} HH_*(\End(\Obs^{(j_1,\ldots,j_d)}_0)) .
\]
Indeed, if we assume that each $\Obs^{(j_1,\ldots,j_d)}_0$ is finite dimensional, Morita invariance implies that this map of graded vector spaces is given by a single linear map
\[
HH_*(\rho) : HH_*(U(\CC^d)) \to \bigoplus_{(j_1,\ldots,j_d)} HH_0(\End(\Obs^{(j_1,\ldots,j_d)}_0)) = \bigoplus_{(j_1,\ldots,j_d)} \CC .
\]
Choosing an isomorphism $\oplus_{(j_1,\ldots,j_d)} \CC = \CC[q_1^\pm,\ldots, q^{\pm}]$ we witness the ${\bf q}$-character above as the image of $1 \in HH_*(U(\CC^d))$ under this map
\[
\chi_{\bf q} (\Obs_0) = HH_*(\rho) (1) .
\] 
More concretely, we can express the character as $\chi_{\bf q} (\Obs_0) = \Tr_{\Obs_0} (q_1^{L_0^1} \cdots q_d^{L_0^d})$. 

When a holomorphic theory posses extra symmetries, there are equivariant versions of the $q$-character. 
For instance, if the theory has an additional $U(1)$-symmetry we can define the multi-variable character 
\[
\chi ({\bf q} , u) = \sum_{j_1,\ldots,j_d \in \ZZ} \sum_{k \in \ZZ}  q_{1}^{j_1} \cdots q_{d}^{j_d} u^k \dim (\Obs^{(j_1,\ldots,j_d), k}_0) \in \CC[\![q_1^{\pm}, \ldots, q_d^{\pm}]\!] .
\]
where $\Obs^{(j_1,\ldots,j_d), k}_0$ is the $(j_1, \ldots, j_d), k$-eigenspace of the holomorphic local operators with respect to $T^d \times U(1) \subset U(d) \times U(1)$.

\subsection{Local operators of the free theory on $\CC^2$} \label{sec: bg ops}

\def\bb{\mathsf{b}}
\def\cc{\mathsf{c}}

We now turn our focus to a particular holomorphic theory: the free $\beta\gamma$ system on $\CC^2$. 
This will be our first calculation of a holomorphic character.
Before we proceed, we present a description of the local operators of the theory on $\CC^2$. 

Recall, the BV fields of the $\beta\gamma$ system with values in a complex vector space $V$ is the Dolbeault complex on $\CC^2$ with values in the vector space
\[
V \oplus \d^2 z \cdot V^\vee [1] .
\] 
When $V$ is ungraded, it is thus in cohomological degree zero, and $V^\vee$ in degree $(-1)$.
Fix a basis $\{e_i\}_{i=1}^{N = \dim(V)}$ for $V$ and let $\{e^i\}$ be the dual basis.
Solutions to the classical equations of motion are parametrized by fields 
\begin{equation}
  \begin{aligned}[c]
    \gamma^0_i &\in \sO^\text{hol}(\CC^2) \tensor V , \\
    \beta^{0;j}\, \d ^2 z &\in \Omega^{2,hol}(\CC^2) \tensor V^\vee [1] = \d^2 z \cdot \sO^\text{hol}(\CC^2) \tensor V^\vee [1].
  \end{aligned}
\end{equation}

We label the corresponding linear local holomorphic operators (supported at $w = 0 \in \CC^2$) with bold letters as 
\[
\begin{array}{ccclll} 
\bgamma_{n_1,n_2; i} & : & \gamma^0 & \mapsto & \frac{\partial^{n_1}}{\partial z_1^{n_1}} \frac{\partial^{n_2}}{\partial z_2^{n_2}} \gamma^0_i (z=0) \\
\bbeta_{n_1+1, n_2+1}^j & : & \beta^{0} \d^2 z & \mapsto & \frac{\partial^{n_1}}{\partial z_1^{n_1}} \frac{\partial^{n_2}}{\partial z_2^{n_2}} \beta^{j} (z=0) ,
\end{array}
\]
where $n_1,n_2 \geq 0$ and $i,j \in \{1,\ldots,N\}$.
We use the bold fonts $\bbeta, \bgamma$ to distinguish linear operators from their fields $\beta, \gamma$. 
Note that the ghost degree of $\bgamma_{n_1,n_2; i}$ is $0$ and the ghost degree of $\bbeta_{n_1+1, n_2+1}^j$ is $+1$. 

Using this basis, it is immediate to verify the following description of local holomorphic operators. 

\begin{lem}
Let $\Obs_0^\text{hol}$ be the local holomorphic operators at $w=0$ of the free $\beta\gamma$ system on $\CC^2$.
There is a graded isomorphism
\[
\Obs_0^\text{hol} \cong \Sym \left( (\CC[\![z_1,z_2]\!] \tensor V)^\vee \oplus (\CC[\![z_1,z_2]\!] \tensor V^\vee)^\vee [-1] \right) [\hbar]
\]
which on linear generators sends
\[
z_1^{-n_1}z_2^{-n_2} e_i + z_1^{-m_1} z_2^{-m_2} e^j \mapsto \bgamma_{n_1,n_2;i} + \bbeta_{m_1+1,m_2+1}^j
\]
where $\{e^i\}$ is a basis for $V$ and $\{e_i\}$ is the dual basis. 
\end{lem}

With this description of the local operators of the $\beta\gamma$ system on $\CC^2$ in hand, we move on to present a formula for the character. 

\subsection{Symmetries in $4d$ $\N=1$}

We will present the character of the free $\beta\gamma$ system on $\CC^2$ in two equivalent ways. 
The first is natural from the description of the theory as a holomorphic one. 
The other arises from most naturally from the description of the theory from the twist of $4d$ $\N=1$ supersymmetry. 

\subsubsection{The first description of the holomorphic character}

We can summarize the symmetries present in the free holomorphic theory on $\CC^2$ as follows.
For the various $U(1)$ symmetries, we use the notation $U(1)_{y}$ when we want to stress which variable for the Cartan, or fugacity, used in the expression of the character.

\begin{itemize}
\item The $U(2)$ symmetry, present in any holomorphic theory on $\CC^2$, whose character will decompose with respect to its Cartan $U(1)_{q_1} \times U(1)_{q_2}$ that we label by $q_1,q_2$;
\item The $U(1)_z$-flavor symmetry. 
Here, $\bgamma$ has weight $+1$ and $\bbeta$ has weight $-1$;
\item The $U(1)_u$ symmetry present on the BV complex corresponding to the ghost weight. 
Note that while the action has ghost degree zero, there are local operators of nontrivial ghost degree:
The operator $\bbeta$ has ghost degree $+1$. 
\end{itemize}
For general free $\beta\gamma$ systems, these fugacities are the generalization to arbitrary complex dimension of the regraded fugacities used in the discussion of the elliptic genus in~\cite{GukovGadde}. Note that these are all symmetries of the {\em classical} BV theory. 
In fact, they all extend (uniquely) to symmetries of the quantum theory.

\begin{lem}\label{lem: U(2) equivariance} 
The symmetry by $U(2) \times U(1)_z \times U(1)_u$ on the classical free $\beta\gamma$ system with values in the complex vector space $V$ lifts to a symmetry of the quantization.
\end{lem}

\begin{proof}
The differential on the quantum observables is of the form $\dbar + \hbar \Delta$. 
The operator $\dbar$ is manifestly equivariant for the action of $U(2)$.
Since $U(1)_z \times U(1)_u$ does not act on spacetime, $\dbar$ trivially commutes with its action. Further, the action of $U(2)$ is through linear automorphisms, and since the BV Laplacian $\Delta$ is a second order differential operator, it certainly commutes with the action of $U(2)$. 
Likewise, since $U(1)_z \times U(1)_u$ is compatible with the $(-1)$-symplectic pairing, it automatically is compatible with $\Delta$. 
\end{proof}

In conclusion, each of the bulleted symmetries above extend by $\hbar$-linearity to symmetries of the quantum observables of the free theory.
We now compute the local character with respect to the group $U(2)_{q_1,q_2} \times U(1)_z \times U(1)_u$. 

\begin{prop} 
\label{prop: char}
The local character of the free $\beta\gamma$ system on $\CC^2$ is equal to
\[
\chi(q_1,q_2 ; z ; u) = \prod_{n_1, n_2 \geq 0} \frac{1 - z^{-1} u q_1^{n_1 + 1} q_2^{n_2 + 1}}{1 - z q_1^{n_1} q_2^{n_2}} \in \CC[\![q_1^{\pm},q_2^{\pm}, z^\pm, u]\!] .
\]
The specialization $u = 1$ is well-defined and recovers the elliptic $\Gamma$-function
\[
\chi(q_1,q_2 ; z ; u) |_{u=1} = \Gamma^\text{ell} (q_1,q_2 ; z) .
\] 
\end{prop}

For an introduction to the elliptic $\Gamma$-function and other related hypergeometric series we refer to the textbook reference \cite{Gasper}. 

\begin{proof}
For fixed $n_1,n_2 \geq 0$, let $V^\vee_{n_1,n_2}$ denote the linear span of operators $\{\bgamma_{n_1,n_2; i}\}_{i=1}^N$. 
As a vector space $V^\vee_{n_1,n_2} \cong V^*$, but we want to remember the weights under $U(2)$. 
Likewise, for $n_1 , n_2 > 0$, let $V_{n_1,n_2} \cong V$ be the linear span of the operators $\{\bbeta_{n_1,n_2}^j\}_{j=1}^N$. 

The holomorphic local operators, then, decompose as
\[
\Obs^\text{hol}_{0} = \Sym \left( \left(\bigoplus_{n_1,n_2 \geq 0} V_{n_1,n_2}^*\right) \oplus \left(\bigoplus_{n_1,n_2 > 0}  V_{n_1,n_2}[-1] \right)\right) [\hbar]
\]
The actions of the remaining symmetry groups are easy to read off.
On $V_{n_1,n_2}^\vee$, the group $U(1)_{z} \times U(1)_{u}$ acts by $(+1, 0)$.
On $V^\vee_{n_1,n_2}$, the group $U(1)_{z} \times U(1)_{u}$ acts by $(-1, +1)$.

To compute the character of the local operators it suffices to compute it on the vector space
\[
\Sym \left( \left(\bigoplus_{n_1,n_2 \geq 0} V_{n_1,n_2}^*\right) \oplus \left(\bigoplus_{n_1,n_2 > 0} \oplus V_{n_1,n_2}[-1] \right)\right) \cong \Sym \left(\bigoplus_{n_1,n_2 \geq 0} V_{n_1,n_2}^*\right) \tensor \Wedge \left(\bigoplus_{n_1,n_2 > 0} V_{n_1,n_2} \right) .
\]
We have used the convention that as (ungraded) vector spaces the symmetric algebra of a vector space in odd degree is the exterior algebra.\footnote{For instance, if $W$ is an ordinary vector space, $\Sym(W[-1]) = \Lambda(W)$ as ungraded vector spaces.}
We can further simplify the right-hand side as
\[
\bigotimes_{n_1, n_2 \geq 0} \left(\Sym(V^*_{n_1,n_2})\right) \bigotimes \bigotimes_{n_1,n_2 > 0} \left(\Wedge (V_{n_1,n_2})\right) .
\] 
The character of the symmetric algebra $\Sym(V^\vee_{n_1,n_2})$ contributes
\[
\frac{1}{1-z q_1^{n_1}q_2^{n_2}}
\]
and the character of $\Wedge(V_{n_1,n_2})$ contributes
\[ 
1- z^{-1} u q_1^{n_1+1}q_2^{n_2+1} .
\]
The formula for character in the statement of the proposition follows from the fact that the character of a tensor product is the product of the characters. 
\end{proof}

Our expression for the character of the local operators of the $\beta\gamma$ system on $\CC^2$ agrees with the partition function of the $\N = 1$ supersymmetric chiral multiplet on the manifold $S^3 \times S^1$, computed in \cite{Rom, Closset1, Closset2, SpiriDual}.
For a direct calculation of the partition function in the holomorphically twisted theory, which agrees with our answer here, see below. 

\subsubsection{A physical expression of the character} \label{sec: SU2}

There is another useful way to decompose the character we have just computed. 
It will be useful later on once we introduce the superpotential. 
The variant amounts to decomposing the local operators with respect to the double cover $SU(2) \times U(1)$ of $U(2)$. 
The symmetries are:

\begin{itemize}
\item The Lorentz $SU(2)_p$ symmetry whose Cartan we label by the coordinate $p$. This action arises on operators through the fundamental action of $SU(2)$ on $\CC^2$;
\item The $U(1)_q$ symmetry arising from the grading ${\rm tw}_2$ in Table \ref{tab:gradings}.
For this symmetry, $\bgamma_{n_1,n_2 ; i}$ has weight $n_1+n_2$ and $\bbeta_{n_1,n_2}^i$ has weight $n_1+n_2 + 2$.
\item The $U(1)_u$ symmetry whose corresponding grading we denoted ${\rm tw}_1$ in Table \ref{tab:gradings} (note that this is precisely the ghost degree in the holomorphic twist);  
\item The $U(1)_z$-flavor symmetry. 
Here, $\bgamma$ has weight $+1$ and $\bbeta$ has weight $-1$.  
\end{itemize}

\begin{table}
\begin{center}
\begin{tabular}{c|ccc|cc|c}
& $R$ &  $L$ &  gh & tw$_1$ & tw$_2$ & \\ \hline
$\phi$ & 0 & 0 & 0 & 0 & 0  & $-$ \\
$\bar\phi$ & 0 & 0 & 0 & 0 & 0 & $\gamma^0$ \\
$\psi_+$ & $1$ & $1$ & 0 & $1$ & 2 & $(\beta^0)'$ \\
$\psi_-$ & $1$  & $-1$  & 0 & $1$ & 0 & $-$ \\
$\bar\psi$ & $-1$ & 0 & 0 & $-1$ & $-1$ & $\gamma^1$ \\ \hline
$F$ & $2$ & 0 & 0 & $2$ & $2$ & $-$ \\
$\bar{F}$ & $-2$ & 0 & 0 & $-2$ & $-2$ & $\gamma^2$
\end{tabular}
\qquad
\begin{tabular}{c|ccc|cc|c}
& $R$ &  $L$ &  gh & tw$_1$ & tw$_2$ &  \\ \hline
$\phi^*$ & 0 & 0 & $-1$ & $-1$ & 0 & $-$ \\
$\bar\phi^*$ & 0 & 0 & $-1$ & $-1$ & 0 & $\beta^2$ \\
$\psi_-^*$ & $-1$ & $-1$ & $-1$ & $-2$ & $-2$ & $(\gamma^2)'$  \\
$\psi_+^*$ & $-1$  & $1$  & $-1$ & $-2$ & 0 & $-$ \\
$\bar\psi^*$ & $1$ & 0 & $-1$ & $0$ & $1$ & $\beta^1$ \\ \hline
$F^*$ & $- 2$ & 0 & $-1$ & $-3$ & $-2$ & $-$  \\
$\bar{F}^*$ & $2$ & 0 & $-1$ & $1$ & $2$ & $\beta^0$
\end{tabular} 
\\
\bigskip
\begin{tabular}{c|ccc|cc}
& $R$ &  $L$ &  gh & tw$_1$ & tw$_2$  \\ \hline
$Q_-$ & $1$  & $-1$ & 0 & 1 & 0 \\
$s_0$ & 0 & 0 & 1 & 1 & 0 \\ \hline
$s_\text{int}$ & $-2$  & $0$ & $1$  & $-1$ &  $-2$ 
\end{tabular}
\\
\bigskip
\end{center}
\caption{$U(2)$-equivariant gradings of fields. Here the gradings preserved after the twist are the stabilizers of the line $Q_- + s_0$:  these are the combinations $\text{tw}_1 = R + \text{gh}$ and $\text{tw}_2 = R + L$. Adding $s_\text{int}$ further breaks the grading to $\text{tw}_1 - \text{tw}_2$.}
\label{tab:gradings}
\end{table}

\begin{prop}\label{prop: SU2}
The local character of the holomorphic twist of the free $4d$ $\N=1$ multiplet with respect to the symmetries above is
\[
\chi(p,q; z;u) \prod_{m \geq 0} \prod_{\ell = 0}^m \frac{1 - u z^{-1} q^{m+2} p^{2 \ell - m}}{1- z q^m p^{2 \ell -m}} 
\]
\end{prop}
\begin{proof}
We will first decompose the weights with respect to the grading given by $U(1)_q$.
On linear operators, the decomposition is
\[
\left(\bigoplus_{m \geq 0} V_m^\vee\right) \oplus \left(\bigoplus_{m \geq 0} V_{m+2} \right) [-1]
\]

Let us first compute the contribution of the local operators $\Sym\left(\bigoplus_{m \geq 0} V_m^\vee\right)$ to the local character. 
Since this space has ${\rm tw}_1$ grading zero, it suffices to compute the $SU(2)_p \times U(1)_q \times U(1)_z$ character. 
For each $m$, $V_m^\vee$ is an irreducible representation of $SU(2)$, and hence we have a decomposition
\[
\prod_{m \geq 0} \chi_{SU(2)_p \times U(1)_q \times U(1)_z} \left( \Sym(V_m^\vee) \right) = \prod_{m \geq 0} \sum_{k \geq 0} z^k q^{km} \chi_{SU(2)_p}\left(\Sym^k(V_m^\vee)\right)
\]
The sum on the right hand side is the standard generating function for the determinant, so we can rewrite this as
\[
\prod_{m \geq 0} \frac{1}{\det(1 - z q^m A)}
\]
where the determinant is taken in the $V^\vee_m$ representation and $A$ is the $2 \times 2$ matrix ${\rm diag}(p, p^{-1})$. 
To compute this determinant, we choose the basis $\{z_1^{\ell} z_2^{m-\ell}\}_{\ell = 0}^{m}$ for $V_m^\vee$. 
Since, $A (z_1^{\ell} z_2^{m-\ell}) = p^{2\ell -m} z_1^{\ell} z_2^{m-\ell}$ the expression for the character reduces to
\[
\prod_{m \geq 0} \prod_{\ell = 0}^m \frac{1}{1-z q^m p^{2 \ell -m}} .
\]

Similarly, we can compute the contribution of $\Sym\left(\oplus_{m \geq 0} V_{m+2}\right)[-1]$ to the character which gives
\[
\prod_{m \geq 0} \prod_{\ell = 0}^m \left(1 - u z^{-1} q^{m+2} p^{2 \ell - m}\right)
\]
\end{proof}

\begin{rmk}
Note that the change of variables $q \to (q_1q_2)^{1/2}$ and $p \to (q_1/q_2)^{1/2}$ returns the expression for the character in Proposition \ref{prop: char}. 
This is consistent with the fact that $SU(2) \times U(1)$, whose Cartan we labeled by $q,p$, is a double cover of the group $U(2)$, whose Cartan we labeled by $q_1,q_2$. 
\end{rmk}

\begin{rmk}\label{rmk: flavor}
So far, we have treated the entire target $V$ as weight $+1$ with respect to the flavor symmetry $U(1)_z$.
If $\dim_\CC (V) = N$, then the flavor symmetry is in fact $U(N)$, and we can introduce a flavor fugacity for the entire Cartan subalgebra, thus enhancing the free character.
Labeling the $i$th fugacity by $z_i$, $i=1,\ldots, N$,
this enhanced free character becomes $\prod_{i=1}^N \chi(q_1,q_2; z_i ; u)$. 
\end{rmk}

\subsection{Partition function on Hopf surfaces}

In this section we show how the local character we have computed above is identical to the partition function of the holomorphic twist of $4d$ $\N=1$ chiral multiplet on a particular complex surface called a {\em Hopf surface}. 

We choose to focus on a class of Hopf surfaces which are diagonal.
These compact complex surfaces are defined for any two complex numbers $q_1,q_2$ satisfying $1 < |q_1| \leq |q_2|$ by the quotient
\[
X = \left. \left(\CC^2 \setminus 0\right) \;\; \right/ \;\; \sim
\]
where the relation is $(z_1,z_2) \sim (q^{n}_1 z_1, q^n_2 z_2)$ for $n \in \ZZ$.
As a smooth manifold $X_{q_1,q_2}$ is diffeomorphic to $S^3 \times S^1$, and the Dolbeault cohomology is
\[
H^{0,0}(X_{q_1,q_2}) = H^{0,1} (X_{q_1,q_2}) = H^{2,1}(X_{q_1,q_2}) = H^{2,2}(X_{q_1,q_2}) = \CC 
\]
with all other Dolbeault cohomology groups zero. 
In particular, $X_{q_1,q_2}$ is not K\"ahler.

Our goal is to compare the formula for the local character of the holomorphic theory computed in the last section to the partition function of the theory on Hopf manifolds. 
The relation between the two quantities is evidence for a higher dimensional {\em state-operator} correspondence.
There are some key differences between the usual CFT picture that we wish to point out. 
Firstly, in CFT one uses Weyl transformations to transform $\RR^n$ to $S^{n-1} \times \RR$ and then traces out the remaining direction to obtain the partition function on $S^{n-1} \times S^1$. 
For us, the holomorphic theory on $\CC^2$ restricts to one on $\CC^2 \setminus 0$ which we can then descend to one on the complex manifold $X_{q_1,q_2} \cong S^3 \times S^1$. 
On the other hand, if we perform the reduction in two stages:
\[
\CC^2 \setminus 0 \xto{\cong} S^3 \times \RR_{>0} \to S^3 \times S^1 = X_{q_1,q_2}
\]
then we can think about the partition function as related to Hochschild homology of the algebra obtained from the theory on $\CC^2 \setminus 0$. 
Thus, in the holomorphic case, the relation between the partition function and the trace of local operators in manifest.
The calculations in this section are an explicit test of this relationship. 

The variables involved in the local character consisted of $q_1,q_2$, which labeled the Cartan of the $U(2)$ symmetry group acting on $\CC^2$ by rotations. 
At the level of the partition function, these variables label the complex structure moduli on the Hopf manifold. 
The other local symmetry which we wish to match up with the partition function is the $U(1)_z$-flavor symmetry. 
Globally, we can encode this symmetry by working with a background $U(1)$ connection, which by holomorphicity we can take to be of type $(0,1)$.
That is, we consider the free $\beta\gamma$ system on $X_{q_1,q_2}$ in the presence of a background $(0,1)$-gauge field $A_f \in \Omega^{0,1}(X_{q_1,q_2})$, encoding the $U(1)$-flavor symmetry:
\[
\int_{X_{q_1,q_2}} \beta\dbar \gamma + \int_{X_{q_1,q_2}} \beta A_f \gamma .
\]
Globally, $A_f$ corresponds with a generator of the cohomology group $H^{0,1}(X_{q_1,q_2}) = \CC \cdot a_f$, and to be consistent with the formulas above, we will label the holonomy of $A_f$ by the variable $z$. 
In turn, the partition function will be a function of the variables $q_1,q_2,z$, just as the local character is.

The Hopf manifold can be viewed as the total space of a holomorphic fibration
\[
\begin{tikzcd}
T^2 \ar[r] & X_{q_1,q_2} \ar[d] \\
& \PP^1 
\end{tikzcd}
\]
which is topologically obtained from the Hopf fibration $S^1 \to S^3 \to S^2$ by taking the product with a circle $S^1$. 
In particular, there is a natural (smooth, not holomorphic) map $\pi : X_{q_1, q_2} \to S^3$. 

We will compute this partition function by first compactifying along $\pi$ to obtain a $3$-dimensional theory on $S^3$, with an infinite tower of fields corresponding to the winding modes around $S^1$.
Then, we use a formula for the partition function of partially holomorphic theory on $S^3$, which turns out to be equal to the reduction of our holomorphic theory on $X_{q_1,q_2}$. 
The spirit of our calculation is very similar to the approach in \cite{Spiri3d} at the level of the holomorphic twist. 

When we compactify, we must remember the higher Kaluza--Klein modes, but it is perhaps easier to imagine first the situation of dimensional reduction. 

\subsubsection{Dimensional reduction}

The dimensional reduction of the four-dimensional $\N=1$ supersymmetry algebra to three dimensions is the three-dimensional $\N=2$ supersymmetry algebra, and the four-dimensional chiral multiplet reduces to the $\N=2$ chiral multiplet in three dimensions.
Upon choosing a holomorphic twist $Q \in S_+^{4d}$, two of three translations remain exact upon reduction. There is then a natural description of the resulting theory, the minimal twist of the three-dimensional $\N=2$ chiral multiplet, as follows. 
Following~\cite{ACMV}, we will refer to the twist as ``holomorphic/topological matter.''
The twists and dimensional reductions fit into the following diagram of theories:
\[
\begin{tikzcd}
  4d: &  \text{$\N=1$ chirals} \ar[rrr,"\text{holomorphic twist}"] \ar[d, "\text{dimensional reduction}"'] & & & \beta\gamma {\rm \; system} \ar[d, "\text{dimensional reduction}"] \\
  3d: & \text{$\N = 2$ chirals} \ar[rrr,"\text{minimal twist}"] & & & \text{ hol./top.\ matter} 
\end{tikzcd}
\]

For simplicity, we will give a local description of the three-dimensional theory, that we refer to as holomorphic/topological matter, on the $3$-manifold $\CC \times \RR$.\footnote{The notation is to remind us that we are using the complex structure on $\CC$ and just the real structure on $\RR$.}
In general, a choice of nilpotent supercharge in the three-dimensional $\N=2$ algebra corresponds to a {\em transverse holomorphic foliation} structure, which can be used to give a coordinate independent description of the theory.
We refer the reader to~\cite{ACMV} for more details.  
The fields of the twisted $3d$ theory (in the BV formalism) are given by
\begin{align}
\label{gamma1} \gamma^{3d} & \in \Omega^{0,*}(\CC_z) \hattensor \Omega^*(\RR_t) \tensor V \\
\label{beta1} \beta^{3d} & \in \Omega^{1,*}(\CC_z) \hattensor \Omega^*(\RR_t) \tensor V^\vee [1]
\end{align}
and the action functional is
\beqn\label{3dmatter}
S(\gamma^{3d}, \beta^{3d}) = \int_{\CC \times \RR} \beta^{3d} \d \gamma^{3d} 
\eeqn
where $\d$ is the total de Rham operator on $\CC \times \RR = \RR^3$. 
Notice that only the piece $\d \zbar \partial_{\zbar} + \d t \partial_t$ of the de Rham operator contributes to the action.

\subsubsection{Compactification}

The compactification of the twist of $4d$ $\N=1$ is a bit more subtle.
When we compactify along the map $\pi : X_{q_1,q_2} \to S^3$, we want to remember not just the zero modes but all of the higher modes as well.
The resulting $3$-dimensional theory is of a similar form to the holomorphic/topological matter system, but there is an infinite tower of fields corresponding to the winding modes along $S^1$.
Locally, we can write the fields in the form
\begin{align}
\gamma^{3d}_\pi = \sum_{n\in \ZZ} \gamma_n^{3d}(z,\zbar, t) e^{in \theta} \\
\beta_\pi^{3d} = \sum_{n \in \ZZ} \beta_n^{3d} (z,\zbar,t) e^{i n \theta} .
\end{align}
where $\gamma_n^{3d}, \beta_n^{3d}$ are $3$-dimensional fields as in~\eqref{gamma1} and~\eqref{beta1} for each $n \in \ZZ$, and $\theta$ is a coordinate on~$S^1$. 

The theory, upon compactification, is of the form 
\beqn\label{Sn}
S_n(\gamma_n^{3d}, \beta_n^{3d}) = \int_{S^3} \beta_n^{3d} \d \gamma_n^{3d} + \beta_n^{3d} a_{f} \gamma_n^{3d} + \beta_{n}^{3d} A_{b, n} \gamma_n^{3d}
\eeqn
Here, $a_f$ is the constant background $U(1)$-flavor connection, and $A_{b,n}$ is a background connection that is proportional to the winding mode
\[
A_{b} = i n \d t + \cdots .
\]
This background connection has the effect of shifting the energy of the modes by the integer $n$. 
In particular, when $n=0$ this background field vanishes (up to a gauge transformation) and we are left with the dimensionally reduced theory from before. 

In \S 6.2 of \cite{ACMV} the partition function along $S^3$, equipped with a THF structure, of the 3d holomorphic/topological matter theory is computed. 
For the free theory (\ref{3dmatter}), with constant background connection $a_f$, the answer is
\[
Z_{3d}^{\rm chiral}(S^3_{q_1,q_2}) = \prod_{n_1,n_2 \geq 0} \frac{n_1 \tau_1 + n_2 \tau_2 - i (\tau_1 + \tau_2) + i a_f}{n_1 \tau_1 + n_2 \tau_2 - i a_f}
\]
where $q_i = e^{2\pi i \tau_i}$.

Now, the full partition function of the theory on $X_{q_1,q_2}$ is written as a product of the $3$-dimensional partition functions over the winding modes
\beqn\label{4dpartition}
Z_{4d}^{\rm chiral}(X_{q_1,q_2}) = \prod_{n \in \ZZ} Z_{3d}^{{\rm chiral}, (n)}(S^3_{q_1,q_2}) 
\eeqn
where $Z_{3d}^{{\rm chiral}, (n)}(S^3_{q_1,q_2})$ is the partition function of the $3d$ theory corresponding to the $n$th mode. 

For each $n \in \ZZ$, labeling the winding mode, we have the similar looking formula of the partition function of the theory (\ref{Sn}):
\[
Z_{3d}^{{\rm chiral}, (n)}(S^3_{q_1,q_2}) = \prod_{n_1,n_2 \geq 0} \frac{n_1 \tau_1 + n_2 \tau_2 + n - i (\tau_1 + \tau_2) + i a_f}{n_1 \tau_1 + n_2 \tau_2 + n - i a_f} .
\]
After simplifications involving a regularization scheme---see \cite{AsselHopf} for instance---the product (\ref{4dpartition}) reduces to the elliptic $\Gamma$-function
\[
Z_{4d}^{\rm chiral}(X_{q_1,q_2}) = \Gamma_{ell}(q_1,q_2 ; z) 
\] 
where $z = e^{2 \pi i a_f}$ is the holonomy of the $U(1)$-flavor connection $a_f$. 

\begin{rmk}
The holomorphic twist of the $\N=(1,0)$ hypermultiplet in six dimensions is equivalent to the $\beta\gamma$ system on $\CC^3$; see~\cite{Butson} and~\cite{ESW}.
One can compute the holomorphic local character in a way completely similar to the method here.
It would be useful to compare the calculation of the partition function on the Hopf $3$-fold diffeomorphic to $S^5 \times S^1$ to to the local character as we've just done for Hopf surfaces.
\end{rmk}

\section{Turning on a superpotential}
\label{sec:Fterms}

In this section we perform the calculation of the local character of the $\beta\gamma$ system on $\CC^2$, valued in $V$, in the presence of a holomorphic interaction defined by the holomorphic potential $W : V \to \CC$. 
Recall, we have seen how this theory arises as the twist of the theory of the $\N=1$ chiral multiplet deformed by the $F$-term superpotential determined by $W$.

Our method to compute the local character is direct. 
In the BV formalism, the differential on observables can alway be split up as $s = s_{\rm free} + s_{\rm int}$ where $s_{\rm free}$ comes from the free part of the action and $s_{\rm int}$ from the interacting part. 
In turn, this induces a spectral sequence whose first page computes the cohomology with respect to $s_{\rm free}$ and the differentials on the higher pages are determined by $s_{\rm int}$. 
For holomorphic local operators, when $Q^\text{hol} = 0$, we have implicitly already taken the cohomology with respect to the free part, which is always equal to the Dolbeault operator of some holomorphic vector bundle. 
Thus, all that remains to compute is the cohomology with respect to $s_{\rm int}$. 

\subsection{The holomorphic $\sigma$-model}

The theory in the presence of superpotential interactions can be understood geometrically. Our results can be interpreted as saying, analogous to Remark~\ref{rmk:sigma-model}, that the theory is equivalent to a holomorphic sigma model whose target space is the derived critical locus of~$W$.

Let us unpack this a bit. 
As we've already remarked, the free $\beta\gamma$ system valued in $V$ makes sense in any complex dimension $d$ as the cotangent theory to the moduli space of holomorphic maps ${\rm Map}^\text{hol}(\CC^d, V)$. 
One way of writing this is as 
\[
T^*[-1] \; {\rm Map}(\CC^d_{\dbar} , V)
\]
where $\CC^d_{\dbar}$ indicates the derived manifold whose dg ring of functions is $\Omega^{0,*}(\CC^d)$, which resolves holomorphic functions on $\CC^d$.
Thus, the $\beta\gamma$ system is the cotangent theory of a very natural holomorphic $\sigma$-model into a vector space target.

In particular, the $\beta\gamma$ system on $\CC^2$ on with values in the complex vector space $V$ is the BV theory whose fields are
\begin{align*}
T^*[-1] \; {\rm Map}(\CC^2_{\dbar}, V) & = T^*[-1] \left(\Omega^{0,*}(\CC^2) \tensor V \right) \\ & = \Omega^{0,*}(\CC^2) \tensor V \oplus \Omega^{2,*}(\CC^2) \tensor V^\vee [1] .
\end{align*}

We can rewrite the fields once we choose a Calabi--Yau form on $\CC^2$, which we always assume is simply~$\d^2 z$. 
Indeed, up to issues of compact support, we have
\[
T^*[-1] \; {\rm Map}(\CC^2_{\dbar}, V) = {\rm Map}\left((\CC^2_{\dbar}, \d^2 z) , T^*[1] V\right)
\]
where on the right hand side the $(+1)$-shifted cotangent bundle to $V$ has appeared. 
This gives us an AKSZ description of the $\beta\gamma$ system.
If we choose boundary conditions at $\infty$ in $\CC^2$, the Calabi--Yau form endows $\CC^2$ with the structure of a $2$-oriented dg manifold. 
The standard pairing on $T^*[1] V$ endows it with the structure of a $1$-shifted symplectic structure, and the free $\beta\gamma$ system is equivalent to the resulting AKSZ theory.

Thus, once we choose a holomorphic volume form, the free $\beta\gamma$ system is the holomorphic $\sigma$-model with target the derived manifold $T^*[1] V$ equipped with the dg ring of functions
\beqn\label{koszul1}
\left(\sO(T^*[1] V), Q = 0\right) = \left(\Sy^* (V^\vee) \tensor \Lambda^*(V), 0\right)
\eeqn
with {\em zero} differential.
Here, $\Wedge^k V$ is placed in degree $+k$. 

Using this description, it is easy to see what happens when we turn on the interaction $L_W = \int \d^2 z \,W(\gamma)$ determined by the holomorphic potential $W \in \sO(V)$. 
Recall, the Koszul resolution is a cochain complex which computes the derived critical locus of $W$.
It has the form
\[
\left(\sO(T^*[-1] V), \{W, -\}\right)
\]
where $\{-,-\}$ is the $(-1)$-shifted Poisson bracket associated to the standard $(-1)$-shifted symplectic structure on $T^*[-1] V$. 
Notice that the graded ring $\sO(T^*[1] V)$ in (\ref{koszul1}) differs from the graded ring of the standard Koszul resolution, but they agree if we take the gradings modulo $2$. 
Thus, we can identify the $\ZZ/2$ graded ring of functions on the target of the holomorphic $\sigma$-model with the underlying $\ZZ/2$-graded ring of the Koszul resolution of $W$. 

The effect of turning on $W$ deforms the $\ZZ/2$-graded dg ring (\ref{koszul1}) to 
\beqn\label{koszul2}
\left(\sO(T^*[1] V), \{W, -\}\right)
\eeqn
which we can, in turn, identify with the underlying $\ZZ/2$-graded dg ring of the standard Koszul resolution. 

In conclusion, the $\beta\gamma$ system deformed by $L_W$ is equal to the $\ZZ/2$-{\em graded} theory given by the AKSZ $\sigma$-model
\[
{\rm Map}(\CC^2_{\dbar}, T_W^*[1] V)
\]
where $T_W^*[1] V$ denotes the odd symplectic $\ZZ/2$-graded dg manifold whose ring of functions is (\ref{koszul2}). 
Thus, we can interpret the theory in the presence of a superpotential $W$ as the $\sigma$-model of $\CC^2$ mapping into the {\em derived critical locus} of $W$. 
Here, we forget the $\ZZ$-graded dg manifold $\CC^2_{\dbar}$ down to a $\ZZ/2$ dg manifold in the obvious way.
This has the effect of only remembering the Dolbeault form degree modulo $2$. 

\subsection{The chiral Jacobi ring}

To perform the calculation of the character, we first collect some facts about the $\beta\gamma$ system in the presence of a superpotential.
\begin{itemize}
\item
As above, denote the fields of the $\beta\gamma$ system on $\CC^2$ by
\deq{
\gamma_i = \gamma^0_i + \gamma^{1}_i + \gamma^{2}_i \in \Omega^{0,*}(\C^2), \qquad \beta^{0 ; i} + \beta^{1 ; i} + \beta^{2 ; i} \in \Omega^{2,*}(\C^2) [1]
}
where $i = 1, \ldots, N = \dim_\CC V$. 

\item The BV pairing between fields is the evaluation pairing on $V$ together with the wedge product of forms $(\beta, \gamma) \mapsto [\<\beta\wedge\gamma\>_V]_\text{top}$, so that the antifield to~$\gamma^{a}_i$ is~$\beta^{0, 2-a ; i}$.

\item The classical BV action of the twisted theory is
\deq{
L = i \langle \beta, \bar\partial \gamma \rangle + L_\text{int}(\gamma),
}
where
\deq{
L_\text{int}(\gamma) = \frac{1}{2} \left(W_{ij}(\gamma^0) \gamma_i^{0,1}\wedge \gamma_j^{0,1} + W_{i}(\gamma^0) \gamma_i^{0,2}\right) \d^2 z .
}

\item The local operators, as in~\S\ref{sec: bg ops}, only depend on the lowest component of the $\beta,\gamma$ fields and are denoted by $\bgamma_{n_1,n_2 ; i}, \bbeta_{n_1+1,n_2+1}^j$. 
\end{itemize}

To see the interaction spectral sequence for the local operators supported at $w=0 \in \CC^2$, we will write the bicomplex arising from separating the BV differential $s = \{S,-\}$ into two parts, corresponding to the decomposition $S= S_{\rm free} + S_\text{int}$. 
The cochain complex of local operators at $0$ are of the form
\[
\left(\Obs_0 , s = s_{\rm free} + s_{\rm int}\right)
\]
where $s_0$ is the piece of the BV differential coming from the free part of the action, and $s_{\rm int}$ comes from the interactions. 
It is clear that
\deq{
s_0  = \pm i \bar\partial,
}
acting with opposite signs on fields and antifields, so that the free BV complex is a copy of the Dolbeault complex ($\gamma$) in ghost number zero, together with another copy of the Dolbeault complex ($\beta$), shifted by the Calabi--Yau form, in ghost number $-1$.
Thus, the $E_1$ page of the spectral sequence is precisely the holomorphic local operators $\Obs_0^\text{hol}$ of the free $\beta\gamma$ system, see Remark \ref{rmk: ss1}.
Thus, we can identify the first page of the spectral sequence is
\beqn\label{ss2}
\left(\Obs_0^\text{hol} , s_{\rm int}\right) .
\eeqn
Recall that the holomorphic local operators of the free system have been described in~\S\ref{sec: hol op}.

The cohomology of this complex is the next page of the spectral sequence, and is governed solely by the superpotential $W$. 
Since $L_\text{int}$ only depends on the  fields $\gamma$, $s_\text{int}$ can only act nontrivially on antifields. From the form of the antibracket, the differential on local operators is as follows:
\deq{
s_\text{int}: \bbeta_{n_1,n_2}^i \mapsto \pm \frac{\delta L_\text{int}}{\delta \bgamma_{n_1,n_2;i}}. 
}
To write down an explicit formula for this differential, we set up some notation. 

\subsubsection{} \label{sec: construction1}

Suppose $B$ is any commutative algebra and let $\xi : B \to \CC$ be an algebra map. 
Given an element $F \in \sO(V)$, we can extend it to one $\Tilde{F}(\xi) \in \sO(B \tensor_\CC V)$ as follows. 
Let $F_n : V^{\tensor n} \to \CC$ be the $n$th homogenous component of $F$. 
Consider the composition
\[
(V \tensor B)^{\tensor n} \xto{F_n \tensor \id_{B}^{\tensor n}} B^{\tensor n} \xto{m} B \xto{\xi} \CC
\]
where $m$ is the multiplication on $B$. 
Symmetrizing, we obtain an element $\Tilde{F}_{n}(\xi) \in \Sym^n(B \tensor V)^\vee$. 
Then, the extension is defined as the sum $\Tilde{F}(\xi) = \sum_n \Tilde{F}_{\xi, n}$. 

This construction 
\[
F \in \sO(V) \rightsquigarrow \Tilde{F}(\xi) \in \sO(B \tensor V)
\]
has the following geometric interpretation. 
Think of $B$ as defining the affine scheme $\Spec(B)$. 
Then, $\sO(B \tensor V)$ can be thought of as observables on the space of maps $\Spec(B) \to V$ (here, $V$ is just a linear space, but this construction can be easily modified for general varieties or even stacks). 
The map of algebras $\xi : B \to \CC$ defines a linear map $\xi : B \times V \to V$.
Equivalently, $\xi$ can be thought of as an element in $\Spec(B)$ and hence determines a map
\[
  {\rm ev}_\xi : \Map(\Spec(B) , V) \to V
\]
by evaluation.
Then, given $F \in \sO(V)$, we can pull-back along ${\rm ev}_\xi$ to obtain an element ${\rm ev}^*(F) = \Tilde{F}(\xi) \in \sO(\Map(\Spec(B) , V)) = \sO(B \tensor_\CC V)$.

Now, specialize to the case $B = \CC[\![z_1, z_2]\!]$.
Geometrically, we are looking at the space of maps from the formal disk to $V$
\[
\hD^2 \to V .
\]
We use the notation $\Tilde{F}_{n_1,n_2} := F(\xi)$ where $\xi : \CC[\![z_1,z_2]\!] \to \CC$ is the linear functional $f \mapsto \partial_{z_1}^{n_1} \partial_{z_2}^{n_2} f$. 
Via the above construction, the potential $W \in \sO(V)$ defines an observable 
\[
\Tilde{\partial_i W}_{n_1,n_2} \in \sO(\CC[\![z_1,z_2]\!] \tensor V)  .
\]
Note that this composite observable is a polynomial in the $\gamma_{n_1,n_2 ; j}$ observables we've defined above, 

Using this notation, the differential $s_{\rm int}$ can be read off as
\[
s_{\rm int} : \bbeta^{i}_{n_1 + 1, n_2+1} \mapsto \Tilde{\partial_i W}_{n_1, n_2} (\bgamma) .
\]
On the right-hand side, we use the notation $(\bgamma)$ to stress that the resulting operator only depends on the local field $\gamma$. 

\subsubsection{A remark on the Jacobian ring}

Classically, the Jacobian ring of a polynomial
\[
W \in \CC[x_1,\ldots,x_N] = \sO(V)
\]
is
\[
{\rm Jac}(W) = \CC[x_1,\ldots, x_N]  / \<\partial_1 W, \ldots, \partial_N W\> .
\]
This ring appears as the $B$-model chiral ring of a two-dimensional $\N=(2,2)$ Landau--Ginzburg theory~\cite{LVW}. 
We introduce a slight holomorphic variant of the Jacobi ring that one might think of as the Jacobian ring for the mapping space $\hD^2 \to V$. 

We have already seen that a piece of the space of local operators of the $\beta\gamma$ system on $\CC^2$ is of the form $\sO(\CC[\![z_1,z_2]\!] \tensor V)$.
The ring we consider is a quotient of this by some ideal we now describe. 
As above, let $W \in \sO(V)$ be any polynomial.

Given an algebra map $\xi : B \to \CC$, we have seen in~\S\ref{sec: construction1} how to extend a polynomial $F \in \sO(V)$ to one $\Tilde{F}(\xi) \in \sO(B \tensor V)$.
In the case $B = \CC[\![z_1,z_2]\!]$, consider the polynomials $\partial_1 W ,\ldots, \partial_N W \in \sO(V)$ and the associated observables
\[
\Tilde{\partial_1 W}_{n_1, n_2} (\bgamma), \ldots, \Tilde{\partial_N W}_{n_1, n_2} (\bgamma) \in \sO(\CC[\![z_1,z_2]\!] \tensor V) \subset \Obs^\text{hol}_0 .
\]

\begin{dfn}
Define the two-dimensional {\em chiral Jacobi ring} of $W \in \sO(V)$ to be
\[
\Jac^\text{hol} (W) = \left. \sO(\CC[\![z_1,z_2]\!] \tensor V) \; \middle/ \; \left\<\Tilde{\partial_1 W}_{n_1, n_2} (\bgamma), \ldots, \Tilde{\partial_N W}_{n_1, n_2} (\bgamma) \; | \; n_1,n_2 \geq 0 \right\> \right. .
\]
\end{dfn}

Note that the usual Jacobian ring sits inside the chiral one as the $U(2)$-weight $(0,0)$ subspace and is of the form $\sO(V) / \<\partial_1 W , \ldots, \partial_N W\>$. 
With this definition in hand, the following lemma is easy to prove. 

\begin{lem}\label{lem: potential cohomology}
The cohomology of the local holomorphic observables of the $\beta\gamma$ system on $\CC^2$ in the presence of a potential $W$, see (\ref{ss2}), is isomorphic to the chiral Jacobi ring:
\[
H^*(\Obs^\text{hol}_0 , s_{\rm int}) \cong \Jac^\text{hol} (W) .
\]
\end{lem}

\begin{proof}
Indeed, the subspace which is killed by $s_{\rm int}$ is $\sO(\CC[z_1,z_2] \tensor V) \subset \Obs^\text{hol}_0$. 
These are the local operators generated by $\bgamma_{n_1,n_2;i}$.
By definition, the image of $s_{\rm int}$ is the subspace $\<\Tilde{\partial W}\>$. 
\end{proof}

\begin{rmk}
In particular, if $W$ is a non-degenerate quadratic polynomial then we see that the cohomology vanishes so that there are no non-trivial local operators in this case.
This is familiar to the case of the $B$-model and the ordinary Jacobi ring. 
\end{rmk}

\begin{rmk}
One can make a similar definition in the complex dimension one case.
We define the one-dimensional chiral Jacobi ring to be
\[
\left. \sO(\CC[\![z]\!] \tensor V) \; \middle/ \; \left\<\Tilde{\partial_1 W}_{n}, \ldots, \Tilde{\partial_N W}_{n} \; | \; n \geq 0 \right\> \right. 
\]
This ring appears as the cohomology of the observables of the holomorphic twist of the $(2,2)$ supersymmetric $\sigma$-model into $V$. 
Note that the further quotient of this ring by the ideal $\sO(z \CC[\![z]\!] \tensor V)$ is precisely the ordinary Jacobi ring of $W$, which is the chiral ring of the Landau--Ginzburg model of $W$. The elliptic genus of Landau--Ginzburg models was first computed in~\cite{Witten-LG}, while \cite{Kolya} gave an analysis of the chiral algebra appearing in the holomorphic twist. 
\end{rmk}

\subsection{The homogenous character}

We now turn to compute the character of the holomorphic theory in the presence of the superpotential.
We rely on Lemma~\ref{lem: potential cohomology}. 

Note that since the theory is no longer $\ZZ$-graded, we no longer have the $U(1)_{u}$-symmetry coming from ghost number.
Thus, the character is only a function of the fugacities $(p,q;z)$, where we use the notation of \S\ref{sec: SU2}.

\begin{prop} \label{prop: potential char}
Let $V = \CC$ and suppose $W$ is homogenous of degree $N+1$. 
Then, the $SU(2) \times U(1) \times U(1)$ character of the holomorphic local operators is given by
\[
\chi_W (p,q;z) = \prod_{m \geq 0} \prod_{\ell = 0}^m \frac{1 - z^N q^m p^{2\ell-m}}{1-zq^m p^{2\ell-m}} .
\]
\end{prop}

\begin{proof}
We compute this using the description of the cohomology in terms of the chiral Jacobi ring in the proposition above.
We have written the cohomology as a quotient $\sO(\CC[[z_1,z_2]] \tensor V) / \<\Tilde{\partial W}\>$.
Before taking the quotient, the character for $\sO(\CC[[z_1,z_2]] \tensor V)$ contributes
\beqn\label{W1}
\prod_{m \geq 0} \prod_{\ell = 0}^m \frac{1}{1-zq^m p^{2\ell-m}} ,
\eeqn
see the proof of Proposition \ref{prop: SU2}. 

It suffices to compute the character of the subspace $\<\Tilde{\partial W}\>$. 
Note that as a $\sO(\CC[[z_1,z_2]] \tensor V)$-module this subspace is generated by the image of $\sO(\CC[z_1,z_2] \tensor V^\vee [1])$ under the differential $s_{\rm int}$.
Since $s_{\rm int}$ has flavor $U(1)_{z}$-weight $N$, we see that the character of $\<\Tilde{\partial W}\>$ is
\beqn\label{W2}
\prod_{m \geq 0} \prod_{\ell = 0}^m \frac{1}{1-zq^m p^{2\ell-m}} \cdot z^N q^m p^{2\ell -m} .
\eeqn
Taking the difference of (\ref{W1}) and (\ref{W2}) yields the result. 
\end{proof}
The reader will recognize the method above as a close analogue of the computation of the Hilbert series of a complete intersection. Indeed, nondegenerate superpotentials are defined by the criterion that the critical points of~$W$ be isolated---in other words, that the critical locus is a complete intersection. Based on our earlier results, however, it should be clear that neither our method nor our identification with the holomorphic $\sigma$-model on the derived critical locus depends on the nondegeneracy of~$W$. This should lead to interesting relations with index computations for four-dimensional $\N=1$ $\sigma$-models, as well as possibly to theorems relating elliptic genera for $\sigma$-models and Landau--Ginzburg theories in two dimensions; we look forward to exploring this in future work.

There is another, efficient way of arriving at the formula for the character in the presence of a superpotential. 
In the free theory, we computed the character $\chi(q_1,q_2 ; z ;u)$ of the graded vector space of holomorphic local operators with respect to the action of the group
\[
U(2)_{q_1,q_2} \times U(1)_z \times U(1)_{u}
\]
where the first factor comes from the natural action on $\CC^2$ by rotations, the second factor is the flavor symmetry, and the last factor encodes the ghost number grading. 
Alternatively, we can replace $U(2)$ by its $2$-fold cover $SU(2)_p \times U(1)_q$ as in~\S\ref{sec: SU2}.

Notice that the differential $s_{\rm int}$ does not preserve the full symmetry group $SU(2)_p \times U(1)_q \times U(1)_u$. 
First off, $s_{\rm int}$ depends on the choice of a volume form, which we take to be the standard one. 
Though this preserves $SU(2)_p$, the operator $s_{\rm int}$ has $U(1)_q$ weight $-2$. 
Secondly, in the case that $W$ is a homogenous polynomial of degree $N+1$, we see that $s_{\rm int}$ has $U(1)_z$ weight $N+1$. 
Finally, $s_{\rm int}$ has $\ZZ$-ghost number $-1$, so it has weight $-1$ under $U(1)_u$.

Putting all this together, we see that only a total $SU(2)_p \times U(1) \times U(1)$ symmetry survives in the case that $W$ is homogenous of degree $N+1$, and in terms of the variables of the free theory we can obtain the interacting character by substituting
\[
\chi (p_{\rm free} ,q_{\rm free} ; z_{\rm free} ; u_{\rm free}) \to \chi_W(p_{\rm int} , q_{\rm int} ; z_{\rm int}) 
\]
\[
p_{\rm free} = p_{\rm int} \; \;  , \;\; q_{\rm free} = q_{\rm int} \; \;  , \;\;  z_{\rm free} = z_{\rm int} \;\;  , \;\; u_{\rm free} = q_{\rm int}^{-2} p_{\rm int}^{N+1} .
\]
One can check immediately that this is compatible with the calculation above. This is the analogue, in two complex dimensions, of the method used for elliptic genera of Landau--Ginzburg models in~\cite{Witten-LG} and (in closer notation) in~\cite{GukovGadde}; the corresponding spectral sequence was studied in~\cite{BPS-SS}.

\begin{rmk}
As in Remark~\ref{rmk: flavor}, in the case that $V = \CC^n$, there is an enhancement of the character where we weight each of the flavor directions with its own fugacity $z_i$, $i=1,\ldots,n$.
The free character is given as a product $\prod_{i=1}^n \chi(p,q ; z_i ; u)$. 
In the case of a potential which is non-degenerate and quasi-homogenous, one has the relation
\[
W(\lambda^{w_1} x_1, \ldots, \lambda^{w_n} x_n) = \lambda^{N} W(x_1,\ldots x_n);
\]
we then arrive at the character by substituting $u = z_i^{(N+1)/w_i} q^{-2}$ in the $i$-th term of the product, obtaining $\prod_{i=1}^N \chi_W(p,q ; z_i, u = z_i^{(N+1)/w_i}q^{-2})$. 
\end{rmk}
 
\section{Holomorphic flavor symmetry}
\label{sec:symmetry}

Consider the $\N=1$ chiral multiplet with matter fields valued in the vector space $V$. 
There is a natural flavor symmetry on the theory by the Lie algebra $\fg = \lie{gl}(V)$ that acts globally.
In fact, this symmetry becomes \emph{local} in the holomorphic twist: it is clear that the action is invariant under any local transformation that depends holomorphically on the spacetime. 
We therefore are interested in the infinite-dimensional symmetry algebra
\deq{
  \sO^\text{hol}_X \otimes_\C \fg ,
} 
where $X$ is the complex manifold on which the $\beta\gamma$ system has been placed.
For the most part, in this section we will take $X = \C^d \setminus 0$, and of course will mostly be interested in the case $d = 2$. 

In the case of a general field theory on $\RR^n$, we have explained how the operators restricted to spheres in $\RR^n \setminus 0$ are endowed with an algebra structure via the OPE in the radial direction. 
A similar result holds true for the OPE of the current algebra~$\sO^\text{hol}_X \otimes_\C \fg$, or really its derived replacement
\begin{equation}
  \Omega^{0,*}_X \tensor_\CC \fg,
\end{equation}
where the differential is the $\bar\partial$ operator on the Dolbeault complex---making it into a resolution of the sheaf of holomorphic functions---and the bracket comes from the bracket on~$\fg$ together with the product structure on the Dolbeault complex. 
One can think of this as a free resolution of the above sheaf of Lie algebras, where ``free'' refers to free modules over functions on spacetime.

In \cite{GwilliamWilliams}, a {\em factorization algebra} is associated to this current algebra, which we call $\Cur(\fg)$, on any complex manifold $X$.
In particular, it exists on the complex manifold $X = \CC^d \setminus 0$. 
The factorization product in the radial direction produces a dg associative algebra from $\Cur (\fg)$, which in the case $d = 1$ is isomorphic to the enveloping algebra of the ordinary Kac--Moody algebra $\sO(\CC^\times) \tensor \fg$. 
For $d > 1$, we thus obtain higher dimensional versions of the Kac--Moody algebra \cite{FHK, GwilliamWilliams}. 

Of course, when acting on a field theory, corrections to the current algebra may arise when the symmetry is quantized. 
For the radial operators, this manifests as a central extension of the classical current algebra. 
When $d=1$, this is the usual central extension of the Kac--Moody vertex algebra, but for general $d$ central extensions are labeled by elements of the space
\[
\Sym^{d+1} (\fgl(n))^{{\rm GL}(n)} .
\]
As described in \cite{FHK, GwilliamWilliams}, such an element defines an $L_\infty$ central extension 
\[
0 \to \CC \to \Tilde{\fg}^*_\theta \to A^*_d \tensor \fg \to 0
\]
where $A_d$ is a certain algebraic dg module for the punctured disk $D^d \setminus 0$.
If $\theta$ is such an element, we will denote the centrally extended current algebra by $\Cur_\theta(\fg)$ which is explicitly given by the $A_\infty$-enveloping algebra of $\Tilde{\fg}_\theta$; see \cite{GwilliamWilliams}.

Note that when $d=1$, the ordinary Kac--Moody extension simply arises as a central extension of a Lie algebra. 
In higher dimensions, the term $\theta$ deforms the classical Lie algebra of symmetries to an $L_\infty$ algebra with nontrivial higher operation in degree $(d+1)$. 
This central extension is controlled by a holomorphic analog of the Konishi anomaly.
We will see the explicit instance of this in the case $d=2$ below. 

Finally, for a general field theory we have described how the radial operators act on the local operators. 
This means that for a theory with a classical current symmetry by $\Cur(\fg)$, like the $\beta\gamma$ system, the quantized local observables will be a module for the deformed current algebra $\Cur_\theta(\fg)$.
This phenomena is familiar in two-dimensional conformal field theories: the naive action by a finite-dimensional Lie algebra on the local observables is promoted to an action by an infinite-dimensional current algebra. We wish to emphasize that this is a general feature of holomorphic theories in any complex dimension; essentially, this is because derivatives in the action are the obstruction to global symmetries being local, and a subset of the derivative operators become nullhomotopic and disappear from the action in the holomorphic twist.

\subsection{Free field realization}
\label{sec: ff}

Based on the work \cite{GwilliamWilliams}, we spell out how the current algebra outlined above witnesses a local enhancement of the flavor symmetry found in the holomorphic twist of the $4d$ $\N=1$ chiral multiplet. 
Following this, we specialize to the complex surface $\CC^2 \setminus 0$ and extract the algebra of $S^3$ operators and its action on the local holomorphic operators of the $\beta\gamma$ system on $\CC^2$.

On any complex manifold $X$, one has the sheaf of commutative dg algebras $\Omega^{0,*}(X)$ which is a fine resolution of the sheaf holomorphic functions on $X$. 
Further, we can tensor with $\fgl(n)$ to obtain a sheaf of dg Lie algebras $\Omega^{0,*}(X) \tensor \fgl(n)$. 
This sheaf comprises the linear generators of the current algebra described above. 

Let's specialize to the case $\dim_\CC X = 2$.
If $V$ is an complex $n$-dimensional vector space, we can explicitly describe the symmetry of the current algebra by coupling the fields $\Omega^{0,*}(X) \tensor \fgl(n)$ as background gauge fields of the free $\beta\gamma$ system with values in $V$. 
Indeed, if $\alpha \in \Omega^{0,*}(X) \tensor \fgl(n)$, then the action functional of the $\beta\gamma$ system is deformed by the term
\[
\int_X \<\beta, \alpha \cdot \gamma\>_V .
\]
Here, $\alpha \cdot \gamma$ extends the natural action of $\fgl(n)$ on $V$ together with the wedge product of Dolbeault forms.
Also, as usual, $\<-,-\>$ denotes the linear pairing between $V$ and its dual.

Writing out the full action, we see that this is nothing but the $\beta\gamma$ system where we have deformed the $\dbar$ operator to $\dbar + \alpha$. 
We can interpret this as the induced deformation on the associated bundle obtained from deforming the trivial principal holomorphic $G$-bundle on $X$. 
Without much more work, one can study deformations of non-trivial holomorphic bundles as well, but it will play no role for us. 

The machinery of \cite{CG2} associates a factorization algebra to any QFT. 
There is also a factorization algebra associated to the symmetries of a QFT, which in this example is the current factorization algebra $\Cur(\fg)$.
To an open set $U \subset X$ it assigns the cochain complex
\begin{equation}
  \label{eq:CurG}
\clieu_*\left(\Omega^{0,*}_c (U , \fg)\right) = \left(\Sym\left(\Omega^{0,*}_c(U, \fg)[1]\right) , \kappa \right)
\end{equation}
where $\kappa$ is the Chevalley-Eilenberg differential for the dg Lie algebra $\Omega^{0,*}_c(U, \fg)$. 
In \cite{CG2} a version of Noether's theorem for factorization algebras is formulated, which from the classical setup above produces a map of factorization algebras from $\Cur(\fg)$ to the factorization algebra of observables of the $\beta\gamma$ system. 
Below, we focus on the value of the factorization algebra on punctured affine space $\CC^d \setminus 0$, and what this map of factorization algebras tells us about the symmetries of the holomorphic local operators.

\subsection{A model for $\CC^2 \setminus 0$}

We now specialize to the case $X = \CC^2 \setminus 0$. 
It's possible to choose a model for the resolution of holomorphic functions on $\CC^d \setminus 0$ that turn out to be convenient formulating the above structures algebraically.
In the case $X = \C^d \setminus 0$, a good one was considered in~\cite{FHK}.
It is constructed as follows: First off, let 
\deq{
  R = \C[z_1,\bar{z}_1,\ldots,z_d,\bar{z}_d]\left[\frac{1}{z\bar{z}}\right]
}
be functions on the punctured affine space, and consider $\Tilde{R}^*$ the free graded-commutative $R$ algebra generated by~$d\bar{z}_i$ in degree $1$.
We take $A^*_d$ to be the graded subalgebra of $\Tilde{R}$ consisting of elements satisfying the following two conditions:
\begin{itemize}
  \item[(a)] \label{conditiona} The overall $\bar{z}$-degree of all elements is zero, and
  \item[(b)] The contraction with the Euler vector field
    \deq{
      \eta = \bar{z}_i \frac{\partial}{\partial \bar{z}_i}
    }
    vanishes.
\end{itemize}
Letting $\xi = z_i \bar{z}_i$ be the squared radius, the result after the first step consists of elements in degree $k$ of the form
\deq{
  f_K d\bar{z}^K,  \qquad f_K \in \frac{1}{\xi^k}\C[z_1,\ldots,z_d] \left[ \frac{\bar{z}_1}{\xi}, \ldots, \frac{\bar{z}_d}{\xi} \right] .
}
Here $K \subseteq \{1,\ldots, d\}$ is a multi-index, with $k = \# K$. 
The Euler vector field condition means that
\deq{
  \sum_K \sum_{i \in K}  \pm f_K \bar{z}_i d\bar{z}^{K\setminus i} = 0,
}
where $\pm$ is the parity of the number of elements of~$K$ preceding~$i$. In particular, this means that our algebra is only nonzero in degrees between zero and~$(d-1)$. 
For the sake of brevity, we will define the subalgebra
\deq{
  R_d = \C[z_1,\ldots,z_d] \left[  \frac{\bar{z}_1}{\xi}, \ldots, \frac{\bar{z}_d}{\xi} \right] \subset R 
}
of elements which satisfy condition (a) above.
Note that $R_d$ is generically not a polynomial algebra, since its generators satisfy the relation
\deq{
  z_1 \frac{\bar{z}_1}{\xi} + \cdots + z_d \frac{\bar{z}_d}{\xi} = 1.
  }
For instance, when $d=2$, it is isomorphic to the quotient algebra
  \deq{
    \C[a,b,c,d]/\langle ac + bd = 1 \rangle,
  }
which we can think of as a quadric in a weighted projective space.

\begin{eg}
  Let $d=1$. Then $\bar{z}/\xi = z^{-1}$, so that the algebra reduces to 
  \deq{
    R_1 = \C[z,z^{-1}]
  }
  concentrated in degree zero with zero differential. This recovers the usual story of Kac--Moody symmetry for theories on the punctured complex plane.
\end{eg}

Let's now consider the case $d=2$ in detail. 
The algebra $A_2^*$ is supported only in degrees zero and one; in degree one, it consists of elements of the form 
\deq{
    f_1\, d\bar{z}_1 + f_2\, d\bar{z}_2,
  }
subject to the condition that
\deq[eq:EVF2]{
    f_1 \bar{z}_1 + f_2 \bar{z}_2 = 0.
}
The differential then maps such an element to
  \deq[eq:dbar2]{
    \left( - \pdv{f_1}{\bar{z}_2} + \pdv{f_2}{\bar{z}_1} \right) d\bar{z}_1 \wedge d\bar{z}_2,
  }
which must be zero for consistency after the Euler vector field condition is imposed. But that condition~\eqref{eq:EVF2} just means that each monomial term in $f_1$ corresponds to another monomial term in~$f_2$, of the form
  \deq{
    f_1 \ni \frac{\bar{z}_1^a \bar{z}_2^{b+1}}{\xi^{a+b+2}} \quad \Longleftrightarrow \quad - f_2 \ni \frac{\bar{z}_1^{a+1} \bar{z}_2^b}{\xi^{a+b+2}}.
  }
  More briefly, we can write
  \deq{
    A^1 = \omega \cdot R_2, 
  }
  where
\deq{
  \omega = \frac{ \bar{z}_2 \, d\bar{z}_1 - \bar{z}_1 \, d\bar{z}_2 } {\xi^2} 
}
is the Bochner--Martinelli kernel in complex dimension two.
It is then easy to verify by direct computation that
  \deq{
    \pdv{f_2}{\bar{z}_1} = 
    \pdv{f_1}{\bar{z}_2} = \frac{\bar{z}_1^a \bar{z}_2^b}{\xi^{a+b+3}} \left[ (b+1) z_1 \bar{z}_1 - (a+1) z_2 \bar{z}_2 \right],
  }
so that~\eqref{eq:dbar2} vanishes. 
It remains to compute the image of the $\bar\partial$ differential inside of the degree-one piece of the algebra. 
Similar to above, we can consider the action of the differential on an allowed monomial in degree zero; this is mapped to
\deq{
    \frac{\bar{z}_1^a \bar{z}_2^b}{\xi^{a+b}}
    \mapsto
    d\bar{z}_1 \, \frac{\bar{z}_1^{a-1} \bar{z}_2^b}{\xi^{a+b+1}}\left[ a z_2 \bar{z}_2 - b z_1 \bar{z}_1 \right]
    +
    d\bar{z}_2\, \frac{\bar{z}_1^a \bar{z}_2^{b-1}}{\xi^{a+b+1}} \left[ b z_1 \bar{z}_1 - a z_2 \bar{z}_2 \right].
 }
  This just has the effect of setting the generators $z_1$ and~$z_2$ to zero in the cohomology of~$\bar\partial$, so that we can identify the cohomology $H^1(A_2^*)$ with the space of elements
  \deq{
h \cdot \omega, 
    \qquad
    h \in \C\left[ \frac{\bar{z}_1}{\xi}, \frac{\bar{z}_2}{\xi} \right].
  }
Note that 
\[
L_{\partial / \partial z_i} \omega = \frac{\zbar_i}{\xi} \omega
\]
where $L_{(-)}$ is the Lie derivative. 
Thus, we can equivalently write the first cohomology as the free $\CC[\partial_{z_1}, \partial_{z_2}]$-module generated by $\omega$, which we can further identify with the dual of power series in two variables 
\[
\CC[z_1,z_2]^\vee = \CC[\partial_{z_1}, \partial_{z_2}] \cdot \omega .
\]
  
This also makes the computation in degree zero easy. 
In that degree, the Euler vector field condition becomes vacuous, so that $H^0(A_2^*)$ just consists of elements 
\deq{
    f \in R_2. 
  }
The computation above further shows that the second set of generators fail to be closed, so that the kernel is precisely polynomials in~$z_1$ and~$z_2$.

\subsubsection{Higher central extensions}

With the dg algebra $A_2^*$ understood, we can now recall the definition of the $d=2$ higher Kac--Moody algebra as in \cite{FHK, GwilliamWilliams}.
We start with the dg Lie algebra obtained from tensoring the ordinary Lie algebra $\fg$ with $A_2^*$. 
We can think of this as an $L_\infty$ algebra with operations: $\ell_1 (a \tensor X) = (\dbar a) \tensor X $ and $\ell_2 (a \tensor X, b \tensor Y) = ab \tensor [X,Y]$. 

For any $\theta \in \Sym^3(\fg^\vee)^\fg$ the dg Lie algebra $A_2^* \tensor \fg$ has an $L_\infty$ algebra central extension
\beqn\label{ext}
0 \to \CC \cdot K \to \Tilde{\fg}^*_{\theta} \to A^* \tensor \fg \to 0
\eeqn
where the $1$-ary and $2$-ary operations are 
\[
\ell_1 (a \tensor X) = (\dbar a \tensor X)  \;\; , \;\; \ell_2 (a \tensor X, b \tensor Y) = ab \tensor [X,Y] \;\; , \;\; \ell_1(K) = \ell_1 (K, a \tensor X) = 0 
\]
and the $3$-ary operation is
\[
\ell_3 (a \tensor X, b \tensor Y, c \tensor Z) = \theta(X,Y,Z) \oint_{S^3} a \partial b \partial c
\]
where $a,b,c \in A_2^*$ and $X,Y, Z \in \fg$.
Here $\oint_{S^3}$ denotes the higher residue pairing, which agrees with the contour integration of a $(2,1)$ differential form along the $3$-sphere.

\subsection{The symmetry multiplet}

Just as the $\beta\gamma$ system on $\CC^2$ arises as the minimal twist of the $4d$ $\N=1$ chiral multiplet, we can also realize the aforementioned current algebra from a of a certain $\N=1$ multiplet.
In~\S\ref{ssec:gauge}, we argued that the holomorphic twist of the $\cN=1$ vector multiplet returns the Dolbeault complex on $\CC^2$ with values in the gauge Lie algebra.

At the level of the BV theory, we have seen that the holomorphic twist of the $\cN=1$ supersymmetric gauge theory is equivalent to holomorphic BF theory with fields
\[
A \in \Omega^{0,*}(\CC^2, \fg)[1] \;\; , \;\; B \in \Omega^{2,*}(\CC^2, \fg^\vee)  .
\]

Using the BV formalism and a bit of trickery, we can use this computation to arrive at an understanding of the holomorphically twisted current multiplet. 
In general, in the BV formalism, antifields generate equations of motion for fields under the action of the BV differential:
\begin{equation}
  s \phi^* = \{ S, \phi^* \} = \frac{\delta S}{\delta \phi}.
\end{equation}
Since one knows that the coupling between the gauge field and the corresponding symmetry current takes the form $A_\mu J^\mu$ in the untwisted theory, it is apparent that one should identify the twist of the current multiplet with the twist of the antifield multiplet, up to a shift by one originating with the homological degree of the bracket. 
That is, we should take 
\deq{
  j \in \Omega^{2,*}(\C^2, \lie{g}^\vee)[1]
}
as the definition of the twisted current multiplet. It is easy to see that a pairing of the form $\langle j,\sA \rangle$ is well-defined and can appear in the action. 

We then find that operators in the fields $j$, so functions on $\Omega^{2,*}(\C^2, \lie{g}^\vee)[1]$, is precisely the current algebra $\Cur(\fg)$ that we consider in this section, as defined in~\eqref{eq:CurG}.
In summary, the (non-centrally extended version of the) current algebra $\Cur(\fg)$ arises as the holomorphic twist of a shift of the anti-field piece of the $\N=1$ vector multiplet which describes flavor symmetries of the untwisted theory.

\begin{rmk}
The same analysis can be done to compute the holomorphic and topological twists of symmetry multiplets in $\N=2$ and~$\N=4$ supersymmetry. 
In general they are given by deformations of the holomorphic current algebra we have just introduced. 
For example, in the holomorphic twist, the flavor multiplet of $\N=2$ supersymmetry is of the form
\[
\Omega^{0,*}(\CC^2) \tensor \fg [\epsilon]
\]
where $\epsilon$ is a formal variable of cohomological degree $+1$.
\end{rmk}

\subsection{Local module}

Finally, we argue why the local holomorphic operators of the twist of $\cN=1$ form a representation for the current algebra we have just introduced. 
As above, we look at the $\beta\gamma$ system with values in the $\fg$-representation $V$. 
This is not a representation in the ordinary since; we have already seen that $\fg_{\theta}^*$ is most naturally exhibited as an $L_\infty$ algebra.
Correspondingly, the local operators form a $L_\infty$-module for this $L_\infty$ algebra. 
\footnote{By a $\fg$-$L_\infty$-module $V$, we mean a map of $L_\infty$ algebras $\rho_V : \fg \rightsquigarrow \End(V)$}.

The action of the higher Kac--Moody current is through the higher ``modes" algebra of the $\beta\gamma$ system. 
This algebra is obtained from placing the $\beta\gamma$ system on $\CC^2 \setminus 0$ and projecting out the radial direction. 
So, as a vector space, it consists of operators supported on the $3$-sphere $S^3$. 
As we've already discussed, the algebra arising from the OPE of sphere operators is a dg analog of the Weyl algebra defined as follows. 
Again, take the algebra $A_2^*$ and consider the abelian dg Lie algebra
\[
A^*_2 \tensor V^\vee [1] \oplus A_2^* \tensor V .
\]
The pairing between $V$ and $V^\vee$ together with the residue pairing between $A_2^*$ and itself defines a central extension of dg Lie algebras
\beqn\label{sesh}
0 \to \CC \cdot K \to \sH_V \to A^*_2 \tensor V^\vee [1] \oplus A_2^* \tensor V .
\eeqn
The differential on $\sH_V$ comes from the differential on $A_2^*$ and the bracket is
\[
[a \tensor v^\vee, b \tensor v] = \hbar \<v, v^\vee\> \oint_{S^3} a \wedge b \d^d z 
\]
for $a,b \in A_2^*$ and $v \in V, v^\vee \in V^\vee$.
The dg Weyl algebra is the enveloping algebra $U(\sH_V)$, which we refer to as the $S^3$-modes algebra.

Consider the dg algebra $A^*$ modeling derived algebraic functions on $\CC^2 \setminus 0$ that we introduced above.
Define the $A^*$-module of {\em negative modes}
\[
A_{2,-} = H^{1}(A^*) .
\] 
Since the cohomology of $A_2^*$ is concentrated in degrees $0,1$, there is a quotient map of dg $A^*$-modules $A_2^* \to A_{2,-}[-1]$. 
Define the dg ideal of {\em positive modes} $A^*_{+} = \ker\left(A_2^* \to A_{2,-}[-1]\right)$, so that there is a short exact sequence of dg vector spaces
\[
A_{2,+}^* \to A_2^* \to A_{2,-}[-1] .
\]

The holomorphic local operators arise as the vacuum Verma module associated to the short exact sequence above.
Indeed, if we replace $A_2^*$ by $A_{2,+}$ in (\ref{sesh}) we obtain an abelian dg Lie algebra $\sH_{V,+}$. 
It's abelian since the residue pairing vanishes when restricted to the positive modes.
There is an isomorphism of vector spaces
\[
\Obs^\text{hol}_{0} \cong U(\sH_V) \tensor_{U(\sH_{V,+})} \CC_{K=1}
\]
thus endowing the holomorphic local operators with the structure of a dg module for $\sH_V$. 

The action of the higher Kac--Moody algebra factors through the $S^3$-modes algebra.
At the quantum level, the current algebra is built from the $L_\infty$ algebra $\Tilde{\fg}^*_{\theta_V}$ as in $(\ref{ext})$ where $\theta_V$ is determined by a certain $1$-loop anomaly analogous to the Konishi anomaly for $\N=1$ SUSY.
The anomaly arises from trying to lift the free field realization of~\S\ref{sec: ff} to the quantum level. 
For the $\beta\gamma$ system with values in $V$, it is shown in~\cite[Corollary 3.13]{GwilliamWilliams} that $\theta_V$ is a multiple of the $\fg$-invariant cubic functional
\[
X, Y, Z \in \Sym^3(\fg) \to \Tr_V(XYZ) .
\]
For this $\theta_V$ one can construct an $L_\infty$-morphism from $\Tilde{\fg}^*_{\theta_V}$ to the algebra of $S^3$-modes of the $\beta\gamma$ system
\[
\Tilde{\fg}^*_{\theta_V} \rightsquigarrow U(\sH_V) 
\] 
which we interpret as a higher dimensional analog of the free field realization in CFT. 
For an explicit formula we refer to~\cite[Proposition 3.14]{GwilliamWilliams}.
Since $\sH_V$ acts on $\Obs^\text{hol}_0$ by definition, there is an induced representation of the $L_\infty$ algebra $\Tilde{\fg}_{\theta_V}^*$ on $\Obs_0^\text{hol}$ through the cited $L_\infty$ map. 

One way to understand this free field realization more explicitly is to construct a version of holomorphic local operators in the current algebra itself. 
Just as in the $\beta\gamma$ case, we can define the $L_\infty$ vacuum module
\[
{\rm Vac}_{\theta_V} (\fg) = U\left(\Tilde{\fg}^*_\theta\right) \tensor_{U(A_+ \tensor \fg \oplus \CC \cdot K)} \CC_{K = 1}
\]
where $\CC_{K = 1}$ is the module for which $K$ acts by $1$. 
We call this $\Tilde{\fg}_{\theta_V}$-module the vacuum module at level $1$.

There is an embedding of the higher Kac--Moody vacuum module inside the holomorphic operators of the $\beta\gamma$ system described as follows. 
First off, as a vector space, we can identify ${\rm Vac}_{\theta_V} (\fg)$ with
\[
\Sym(\CC[z_1,z_2]^\vee \tensor \fg [-1]) .
\]
If $X \in \fg$ we write $X_{n_1, n_2}$ for the linear element $(z_1^{n_1} z_2^{n_2})^\vee \tensor X \in {\rm Vac}_{\theta_V}(\fg)$. 
The embedding
\[
{\rm Vac}_{\theta_V} (\fg) \to \Obs_0^\text{hol}
\]
is defined on linear elements by
\[
X_{n_1,n_2} \mapsto \sum_{k_1, k_2 \geq 0} \sum_{i,j} \bbeta_{k_1, k_2}^j \rho(X)^i_j \bgamma_{n_1-k_1, n_2-k_2 ; i}
\]
where $\rho : \fg \to \End(V)$ denotes the representation. 
This map is compatible with the free field realization above in the sense that it is a map of $\Tilde{\fg}_{\theta_V}^*$-modules. 
In future work, we aim to use the description of the holomorphic local operators as a module for the higher current algebra to decompose the local character into characters of the current algebra.

\section{Dimensional reduction}
\label{sec:dimred}

In this section, we consider the dimensional reduction of the $\beta\gamma$ system to two dimensions. Upon reduction, four-dimensional minimal supersymmetry becomes $\N=(2,2)$ supersymmetry in two dimensions; one can therefore obtain the holomorphic twist of an $\N=(2,2)$ theory with $F$-term interactions (i.e., a Landau--Ginzburg theory) by dimensionally reducing the $\beta\gamma$ system on~$\C^2$ along~$\C$, or considering the $\beta\gamma$ system on a complex four-manifold which is a product of two Riemann surfaces.

More generally, we can consider the dimensional reduction of the $\beta\gamma$ system on flat space, along a plane that may not coincide with a complex subspace of~$\C^2$. Under such a reduction, there is an inclusion map of nilpotence varieties between that of the higher-dimensional and that  of the lower-dimensional theory. However, this map is \emph{not} a map of stratified spaces; the lower-dimensional nilpotence variety will, in general, be stratified more finely than the higher-dimensional one, so that (for example) holomorphic twists in four-dimensional minimal supersymmetry  may reduce to either holomorphic or $B$-type topological twists of two-dimensional $\N=(2,2)$ supersymmetry.

In terms of the complexification $V_\C \cong \C^4$, a complex structure on~$V$ corresponds to a maximal isotropic subspace with respect to the standard bilinear (complex) inner product on~$\C^4$, which can be thought of as the space of homotopically trivial translation operators in the holomorphic twist. A dimensional reduction, on the other hand, corresponds to the choice of a two-dimensional subspace of~$V$, corresponding to the translations that are set to zero. The inner product on the complexification of this space will always be nondegenerate. As such, the intersection of these spaces may have dimension either zero (generically) or one, so that their span (the space of all translations which act trivially in the  twisted, dimensionally reduced theory) has dimension either four (generically) or three. The  former case is a topological twist of $B$-type, whereas the  latter is the two-dimensional holomorphic twist.

In the following subsections, we address each of these constructions in turn, beginning with a spacetime that is the product of two complex curves. 

\subsection{Product of Riemann surfaces}
Let~$\Sigma_1$ and~$\Sigma_2$ be two Riemann surfaces, and consider the $\beta\gamma$ system in two complex dimensions on the product $\Sigma_1 \times \Sigma_2$. 
There is a slight generalization of the $\beta\gamma$ system, as we've introduced it, that will be relevant in this section:
we can replace functions on the complex surface by sections of an arbitrary holomorphic vector bundle, and require that the fields live in the Dolbeault resolution of holomorphic sections of that bundle. 
For the case at hand, we take this bundle to be the pullback of a holomorphic bundle~$\sV$ on~$\Sigma_2$ along the obvious projection $\pi_2 : \Sigma_1 \times \Sigma_2 \to \Sigma_2$. 
The BV fields of the $\beta\gamma$ system on $\Sigma_1 \times \Sigma_2$ with values in the bundle $\pi_2^* \sV$ are then of the form
\beqn\label{bgv}
(\gamma, \beta) \in \Omega^{0,*} (\Sigma_1 \times \Sigma_2, \pi_2^* \sV) \oplus \Omega^{2,*}(\Sigma_1 \times \Sigma_2 , \pi_2^* \sV^\vee)[1]
\eeqn
with free action functional $\int \<\beta, \dbar \gamma\>_{\sV}$.\footnote{The pairing $\<-,-\>_{\sV}$ is the fiberwise linear pairing between $\sV$ and $\sV^\vee$.}

\begin{prop}\label{prop: dimred}
The compactification along~$\Sigma_2$ of the two-dimensional $\beta\gamma$ system on the complex surface $\Sigma_1 \times \Sigma_2$ with values in $\pi_2^* \sV$ is equivalent to the one-dimensional $\beta\gamma$ system on~$\Sigma_1$, with values in the graded vector space
\deq{
H_{\dbar}^*(\Sigma_2,\sV).
}
\end{prop}

\begin{rmk}
As a word of caution, by the ``$\beta\gamma$ system" with values in a graded vector space we allow for the possibility for fields of both even and odd parity.
In one complex dimension, the sectors of odd parity are commonly referred to as ``$bc$ systems" in the literature.
\end{rmk} 

We recognize the $\beta\gamma$ system in the proposition as the holomorphic twist of the $(0,2)$-supersymmetric sigma model with values in the graded vector space $H_{\dbar}^*(\Sigma_2,\sV)$.
It can be checked directly that the the twist of the $(0,2)$ sigma model is given by such a $\beta\gamma$ system; see for instance \cite{WittenCDO}.

\begin{proof}[Proof of Proposition \ref{prop: dimred}]
This is completely direct. 
By Dolbeault formality of Riemann surfaces, we can replace the Dolbeault cochain complex $\Omega^{0,*}(\Sigma_2)$ with its cohomology. 
Thus, the complex fields of the $\beta\gamma$ system~\eqref{bgv} are quasi-isomorphic to
\begin{equation}
  \Omega^{0,*}(\Sigma_1) \tensor H_{\dbar}^*(\Sigma_2 , \sV) \oplus \Omega^{1,*}(\Sigma^2) \tensor H_{\dbar}^*(\Sigma_2, K_{\Sigma_2} \tensor \sV^\vee)[1],
\end{equation}
where we have used the fact that $\Omega^{2,*}(\Sigma_1 \times \Sigma_2) \cong \Omega^{1,*}(\Sigma_1) \hattensor \Omega^{1,*}(\Sigma_2)$.\footnote{$\hattensor$ denotes the completed tensor product, which agrees with the ordinary tensor product for finite dimensional vector spaces.}

By Serre duality, $H_{\dbar}^*(\Sigma_2, K_{\Sigma_2} \tensor \sV^\vee) \cong \left(H^*_{\dbar}(\Sigma_2, \sV)\right)^\vee[-1]$, hence the fields can be written as
\[
\Omega^{0,*}(\Sigma_1) \tensor H_{\dbar}^*(\Sigma_2 , \sV) \oplus \Omega^{1,*}(\Sigma_1) \tensor \left(H^*_{\dbar}(\Sigma_2, \sV)\right)^\vee .
\]
These are the fields of the (ordinary) $\beta\gamma$ system on $\Sigma_1$ with values in $H_{\dbar}^*(\Sigma_2 , \sV)$. 
To check that the action functional is the correct one amounts to observing that the induced BV pairing on this space of fields comes from integration along $\Sigma_1$ together with the linear pairing on $H_{\dbar}^*(\Sigma_2 , \sV)$, which is obvious. 
\end{proof}

\subsubsection{The case $\Sigma_2 = T^2$ and dimensional reduction}

Consider the specific case that $\Sigma_2 = T^2$, and $\sV$ is the trivial bundle with fiber $V$.
The compactified theory is equivalent to the $\beta\gamma b c$ system on~$\Sigma_1$, whose fields are
\[
(\gamma_{1}, \beta_{1}) \in \Gamma(\Sigma_1, \sO_{\Sigma_1} \tensor V \oplus K_\Sigma \tensor V^\vee)
\]
and
\[
(c, b) \in \Gamma(\Sigma, \sO_{\Sigma_1} \tensor V [1] \oplus K_{\Sigma_1} \tensor V^\vee[-1])
\]
with action functional
\deq{
S(\gamma_{1}, \beta_{1}, c, b) = \int_{\Sigma_1} \<\beta_{1}, \bar\partial \gamma_{1}\>_{V} + \int_{\Sigma_1}\<b, \bar\partial c\>_V .
}
We can also obtain this by naive dimensional reduction of the $\beta\gamma$ system on~$\C^2$: we simply take all Dolbeault forms to be independent of the $z_2$ coordinate, obtaining
\begin{equation}
  \Omega^{0,*}(\C, V)[d \bar{z}_2] \oplus dz_2\cdot \Omega^{1,*}(\C, V^\vee) [d\bar{z}_2].
\end{equation}
Identifying the antifield of~$\gamma_1 \in \Omega^{0,*}(\C)$ with $\beta_1 \in dz_2 \, d\bar{z}_2 \cdot \Omega^{1,*}(\C)$, and similarly the antifield of $c \in d\bar{z}_2\cdot \Omega^{0,*}(\C)$ with~$b \in dz_2 \cdot \Omega^{1,*}(,\C)$, we recover precisely the $\beta \gamma b c$ system above, with $\C[d\bar{z}_2]$ playing the role of~$H^*_{\bar{\partial}}(T^2)$.

This $\beta\gamma bc$ system is the holomorphic twist of the $(2,2)$-supersymmetric $\sigma$-model in two dimensions.
In other words, compactification of the holomorphic theory along $T^2$ in $\Sigma_1 \times T^2$ is equivalent to dimensional reduction. Furthermore, dimensional reduction commutes with the holomorphic twist. 

\subsubsection{The case $\Sigma_2 = \PP^1$}

Let $\sR$ be a fixed line bundle on $\PP^1$ and consider the following version of the higher dimensional $\beta\gamma$ system on $\CC_z \times \PP^1$ whose fields are
\[
(\gamma, \beta) \in \Omega^{0,*}(\CC_z \times \PP^1 , \pi^* \sR \tensor \ul{V}) \oplus \Omega^{2, *}(\CC_z \times \PP^1 , \pi^* \sR^\vee \tensor \ul{V}^\vee)[1] 
\]
where $\pi : \CC_z \times \PP^1 \to \PP^1$ is the projection, and $\ul{V}$ denotes the trivial bundle with fiber $V$. 
Just as above, it is easy to read off the compactification of this theory along $\PP^1$ using Dolbeault formality (or by specializing Proposition~\ref{prop: dimred}).
The fields are
\begin{equation}
  \begin{aligned}[c]
    \gamma_{1} &\in \Omega^{0,*}(\CC_z) \tensor H^{*}(\PP^1, \sR) \tensor V , \\
    \beta_{1} &\in \Omega^{1,*}(\CC_z) \tensor H^*(\PP^1, K_{\PP^1} \tensor \sR^\vee) \tensor V^\vee [1].
  \end{aligned}
\end{equation}
By Serre duality, this is precisely the one-dimensional $\beta\gamma$ system on $\CC_z$ with values in the vector space $H_{\bar\partial}^*(\PP^1 , \sR) \tensor V$. 
This confirms the results of~\cite{Closset1} at the level of the twist, which is a special case of Proposition~\ref{prop: dimred}.

\subsection{General reduction}
In this section, we return to considering dimensional reduction of the flat $\beta\gamma$ system along a plane that does not necessarily define a complex subspace of~$\C^2$. 
\begin{prop}
The dimensional reduction of the $\beta\gamma$ system on~$\R^4$, with respect to a fixed real subspace $\R^2 \subseteq \R^4$, is a family over one component of the nilpotence variety (which is equivalently the space of complex structures on~$\R^4$). Its fields are
\deq{
  \left[  \Omega^{*,*}(\C), \epar_+ \bar\partial + \epar_- \partial \right] ,
}
so that the spectral sequence from the two-dimensional holomorphic twist to the B-model is nothing other than the Hodge-to-de-Rham spectral sequence. 
\end{prop}
\begin{proof}
We consider the fields of the $\beta\gamma$ system, at the point on the nilpotence variety with coordinates $\epar_+,\epar_-$. As we have computed above, the differential of the free theory is then determined by a BV action of the form 
\deq{
  L = \langle \beta, \bar\partial_t \gamma\rangle + L_\text{int},
  \qquad
  \bar\partial_\epar = d\bar{z}_i \left( \epar_+ \pdv{}{\bar{z}_i} + \epar_- \epsilon_{ij} \pdv{}{z_j} \right),
}
so that the complex before dimensional reduction is just the sum of Dolbeault complexes with respect to a deformed complex structure:
\deq[eq:t-deformed]{
\Omega^{0,*}(\C^2, \bar\partial_\epar) \oplus \Omega^{2,*}(\C^2, \bar\partial_\epar )[1].
}
Upon dimensional reduction, we simply take the fields that appear to be independent of~$z_2$ and~$\bar z_2$, and replace corresponding derivatives by zero, so that the differential reduces to
\deq{
  \bar\partial_\epar \rightarrow \epar_+\, d\bar{z}_1 \pdv{}{\bar{z}_1} - \epar_- \, d\bar{z}_2 \pdv{}{z_1}.
}
This means that we can rewrite the fields~\eqref{eq:t-deformed} after dimensional reduction as a sum of  \emph{total} de~Rham complexes of~$\C$, after reinterpreting the odd generator  $d\bar{z}_2$ as~$dz_1$:
\begin{equation}
  \begin{aligned}[c]
    \Omega^{0,*}(\C^2, \bar\partial_\epar)  &\rightarrow \left[ \Omega^{0,*}(\C)[d\bar{z}_2], \bar\partial_\epar \right] \\
    & \cong \left[ \Omega^{*,*}(\C), \epar_+ \bar\partial  - \epar_-  \partial \right].
  \end{aligned}
\end{equation}
In the free theory, there is therefore a spectral sequence from the local  operators in  the holomorphic twist (contributing to the elliptic genus) to those contributing to the $B$-model chiral ring, which naturally  appears from  the family of complexes computing the local operators over the four-dimensional nilpotence variety. It is nothing other than the Hodge-to-de-Rham spectral sequence on~$\C$. A similar spectral sequence passes from the holomorphic twist to the $A$-model chiral ring, but cannot be seen by dimensional reduction from four dimensions; for theories of chiral superfields, this is simply a cancelling differential.
\end{proof}

\printbibliography

\end{document}